\documentclass{article}
\usepackage[margin=1in]{geometry}
\usepackage[T1]{fontenc}
\usepackage[utf8]{inputenc}
\usepackage{lmodern}
\usepackage{microtype}

\usepackage{amsmath,amssymb,amsthm}
\usepackage{mathtools}
\usepackage{nicefrac}
\usepackage{bbm}
\usepackage{bm}

\usepackage{booktabs}
\usepackage{enumitem}

\usepackage{graphicx}
\usepackage[dvipsnames]{xcolor}
\usepackage{tikz}
\usetikzlibrary{arrows.meta,positioning,shapes.geometric,fit}

\usepackage{caption}

\usepackage[ruled,vlined,linesnumbered]{algorithm2e}

\usepackage[most]{tcolorbox}

\usepackage{thmtools}
\usepackage{thm-restate}

\usepackage{xspace}

\usepackage[colorlinks=true,citecolor=blue,linkcolor=blue,urlcolor=blue]{hyperref}

\usepackage[capitalize,noabbrev]{cleveref}

\newtheorem{theorem}{Theorem}[section]
\newtheorem{lemma}[theorem]{Lemma}
\newtheorem{proposition}[theorem]{Proposition}
\newtheorem{prop}[theorem]{Proposition}
\newtheorem{corollary}[theorem]{Corollary}
\newtheorem{cor}[theorem]{Corollary}
\newtheorem{claim}[theorem]{Claim}

\newtheorem{conj}[theorem]{Conjecture}

\theoremstyle{definition}
\newtheorem{definition}[theorem]{Definition}

\newtheorem{problem}[theorem]{Problem}

\theoremstyle{remark}
\newtheorem{remark}[theorem]{Remark}

\newtheorem{obs}[theorem]{Observation}

\newtcolorbox{appendixabstract}{
  colback=gray!5, colframe=gray!50, fonttitle=\bfseries,
  title=Abstract, boxrule=0.5pt, arc=2pt, left=6pt, right=6pt, top=4pt, bottom=4pt
}

\newcommand{\R}{\mathbb{R}}

\newcommand{\E}{\mathbb{E}}

\DeclareMathOperator*{\argmax}{arg\,max}
\newcommand{\eps}{\varepsilon}

\newcommand{\OPT}{\mathrm{OPT}}

\newcommand{\algname}{\textsc{Algorithmist}}
\newcommand{\algnameheading}{\texorpdfstring{\textnormal{\textsc{Algorithmist}}}{Algorithmist}}

\definecolor{insightgreen}{RGB}{118,168,126}
\definecolor{insightblue}{RGB}{118,152,198}
\definecolor{insightpurple}{RGB}{152,132,188}
\definecolor{snippetmist}{RGB}{186,196,210}

\newtcolorbox{insightbox}[2][blue]{
  enhanced,
  colback=#1!4!white,
  colframe=#1!18!black,
  colbacktitle=#1!9!white,
  coltitle=black,
  title filled=true,
  boxrule=0.8pt,
  arc=2mm,
  left=1.5mm,
  right=1.5mm,
  top=1mm,
  bottom=1mm,
  lefttitle=2.4mm,
  righttitle=2.4mm,
  toptitle=1.0mm,
  bottomtitle=1.0mm,
  title={#2},
  fonttitle=\bfseries
}

\newcommand{\AgentArchitectureWidth}{\textwidth}

\newcommand{\AfterAgentArchitectureVSpace}{-0.6em}

\newcommand{\cost}{\mathrm{cost}}

\newcommand{\calD}{\mathcal{D}}

\newcommand{\calM}{\mathcal{M}}

\newcommand{\Bin}{\mathrm{Bin}}

\newcommand{\norm}[1]{\left\lVert #1 \right\rVert}

\newcommand{\cI}{\mathcal{I}}
\newcommand{\cN}{\mathcal{N}}
\newcommand{\cD}{\mathcal{D}}

\newcommand{\Prob}{\mathbb{P}}
\newcommand{\cS}{\mathcal{S}}
\DeclareMathOperator{\supp}{supp}

\newcommand{\Prb}{\mathbb{P}}
\newcommand{\calG}{\mathcal{G}}

\newcommand{\calT}{\mathcal{T}}
\newcommand{\Llift}{\mathcal{L}}
\newcommand{\Econv}{\mathcal{E}}
\newcommand{\Aalg}{\mathcal{A}}

\newcommand{\means}{\mathrm{mean}}

\newcommand{\nf}{\nicefrac}
\newcommand{\ones}{\mathbbm{1}}
\newcommand{\RR}{\mathbb{R}}
\newcommand{\EE}{\mathbb{E}}

\newcommand{\cA}{\mathcal{A}}

\newcommand{\cU}{\mathcal{U}}
\newcommand{\cX}{\mathcal{X}}
\newcommand{\bp}{{\bm p}}
\newcommand{\bq}{{\bm q}}
\newcommand{\bx}{{\bm x}}
\newcommand{\by}{{\bm y}}
\newcommand{\be}{{\bm e}}
\newcommand{\bmu}{{\bm \mu}}

\newcommand{\RT}{Random Thresholds\xspace}
\newcommand{\kmed}{$k$-medians\xspace}
\newcommand{\kmeans}{$k$-means\xspace}
\newcommand{\Univ}{U}
\newcommand{\Set}{V}
\newcommand{\intervals}{\mathcal{I}}

\newcommand{\sse}{\subseteq}

\newcommand{\X}{\mathcal{X}}
\newcommand{\U}{\mathcal{U}}
\newcommand{\Best}{\textnormal{\textsc{BestCriticalTree}}}

\allowdisplaybreaks

\title{Early Discoveries of \algname~I:\\
\emph{Promise of Provable Algorithm Synthesis at Scale}}
\author{Janardhan Kulkarni, jakul@microsoft.com}
\date{March 22, 2026}

\begin{document}
\maketitle
\begin{center}
\includegraphics[width=0.82\textwidth,trim=60 95 55 130,clip]{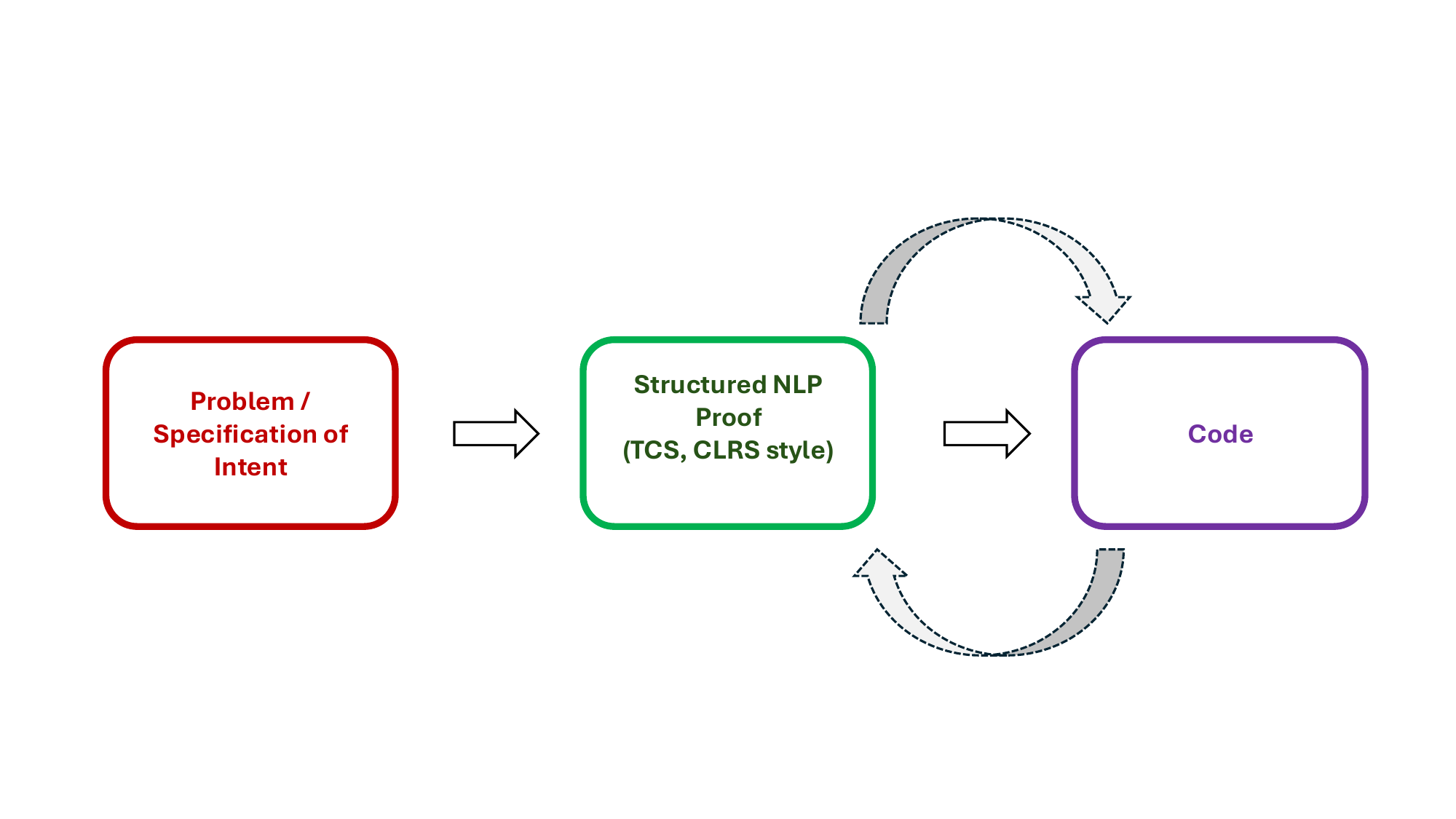}
\captionof{figure}{Proof-first algorithm synthesis and coding by autonomous research agents.}
\label{fig:algorithmist-overview}
\end{center}
\begin{abstract}
Designing algorithms with provable guarantees that also perform well in practice remains challenging, requiring both deep mathematical reasoning and careful end-to-end implementation. Existing approaches that seek to bridge worst-case theory and empirical performance—such as beyond-worst-case analysis and data-driven algorithm selection—typically rely on prior distributional knowledge or restrict attention to a fixed pool of candidate algorithms. With rapid progress in LLMs, a new possibility is emerging: {\em provable algorithm synthesis on the fly.} To study both its promise and limitations, we built {\em \algname}, an autonomous researcher agent on top of GitHub Copilot that instantiates a multi-agent research-and-review loop inspired by scientific peer review, with distinct mechanisms for high-level idea generation, algorithm and proof development, proof-guided implementation, and multi-stage review of the proof, code, and their alignment.

We evaluated \algname{} on research-level algorithm design tasks spanning private data analysis and clustering under multiple constraints. In particular, when tasked with designing methods that jointly satisfy privacy, approximability, and interpretability requirements while remaining practical, \algname{} produced provably sound, empirically effective algorithms along with journal-style research writeups and research-grade audited implementations. Across these case studies, the system also discovered improved algorithms in some settings, provided principled explanations of apparent barriers in others, and uncovered a subtle proof bug in prior published work. Taken together, these results suggest that LLM-based research agents can meaningfully expand the reach of provable algorithm design, while still benefiting from expert human oversight.

More broadly, our results point to a new paradigm in which LLM systems generate research-paper-quality algorithmic artifacts—ideas, specifications, proofs, and audited code—tailored to the requirements of each dataset and deployment setting. They also suggest that coding agents more generally may benefit from a {\em proof-first code-synthesis paradigm}, in which code is developed together with a structured NLP proof intermediate representation, including invariants and correctness arguments of the kind standard in algorithm textbooks and theoretical computer science (TCS) papers, and kept aligned with it throughout synthesis.
\end{abstract}
\clearpage

\section{Introduction}
In many applications, we need algorithms with provable guarantees: approximation bounds for NP-hard optimization, quality and interpretability guarantees for clustering, or privacy guarantees such as differential privacy (DP) in data-sensitive settings. Designing such algorithms is difficult, demanding both mathematical insight and careful proof-based reasoning.
The challenge is even greater in practice, where the “right” algorithm is rarely determined by the guarantee alone. Real deployments often require methods to be adapted to specific datasets, workloads, and operating regimes, while also satisfying additional constraints such as interpretability, determinism, simplicity, pipeline compatibility, memory limits, or low latency. These considerations can materially change both the algorithmic design and the proof obligations.
Moreover, translating a provable algorithm into a working system is often itself a research problem. Theorems frequently leave constants, thresholds, discretizations, or parameter choices implicit, since they are irrelevant to the formal statement but crucial in practice. Making such results usable therefore requires further analysis to instantiate the theory and preserve guarantees under realistic computational and statistical constraints.

\medskip
These challenges help explain why provable algorithms have historically had narrower reach in practice than their potential utility would suggest. Rich lines of work---smoothed analysis, Bilu-Linial stability, self-evolving algorithms, data-driven algorithm selection, and many other notions of beyond worst-case analysis---have aimed to reconcile theory with empirical performance \cite{spielman2004smoothed,bilu2012stable,ailon2011self,gupta2016pac,lykouris2021competitive}; see \cite{roughgarden2019beyond} for a survey. But such approaches are typically limited in an important way: they often assume prior knowledge of the data distribution, or they choose from a fixed, pre-specified pool of candidate algorithms. In many real settings, neither assumption is satisfactory. The challenge is not merely to select among known algorithms, but to synthesize a new algorithmic variant that satisfies a novel combination of guarantees, implementation constraints, and application-specific requirements.

\medskip
Rapid improvement in foundation-model capabilities is ushering in a new era in which provable algorithm design can have far wider practical reach than before. Historically, the difficulty of reconciling guarantees, application-specific constraints, and implementation details limited the adoption of provable methods outside a relatively narrow set of settings. But this limitation need no longer be fundamental. {\em As model capabilities improve, it becomes increasingly realistic to demand that algorithms satisfy not only empirical performance requirements, but also explicit, rigorous, and reviewable guarantees whenever such guarantees are relevant.} Importantly, this shift does not remove the need for expert oversight; rather, it changes where expert effort is spent. Instead of requiring experts to manually bridge every step from abstract theorem to deployable system, we can increasingly ask models to produce strong candidate artifacts—algorithm designs, structured proof sketches, implementation plans, and code—that experts then validate, refine, and certify. This suggests a new paradigm: provable algorithm synthesis on the fly, in which a language-model-based system plays the role of a researcher—designing algorithms tailored to a given dataset, specification, and set of constraints, while also producing the arguments, parameter choices, checks, and audits needed to support those guarantees in an implemented system.

\medskip

To explore this possibility, we built \algname, a researcher agent on top of GitHub Copilot for end-to-end algorithm design, implementation, and auditing. We evaluated \algname{} on two broad classes of problems: (1) provable algorithms that adapt to the structure of a particular dataset, and (2) algorithm synthesis under additional real-world constraints. In the first class, we study DP set union and DP $n$-gram extraction, two primitives for private data release that have been widely used in synthetic data generation pipelines \cite{gopi2020dpsu,kim2021dpne}. In the second, we study $k$-means and $k$-median clustering under the joint requirements of interpretability, approximation quality, and privacy \cite{moshkovitz2020explainable, gupta2023price}.

\begin{figure}[!t]
\centering
\includegraphics[width=\AgentArchitectureWidth]{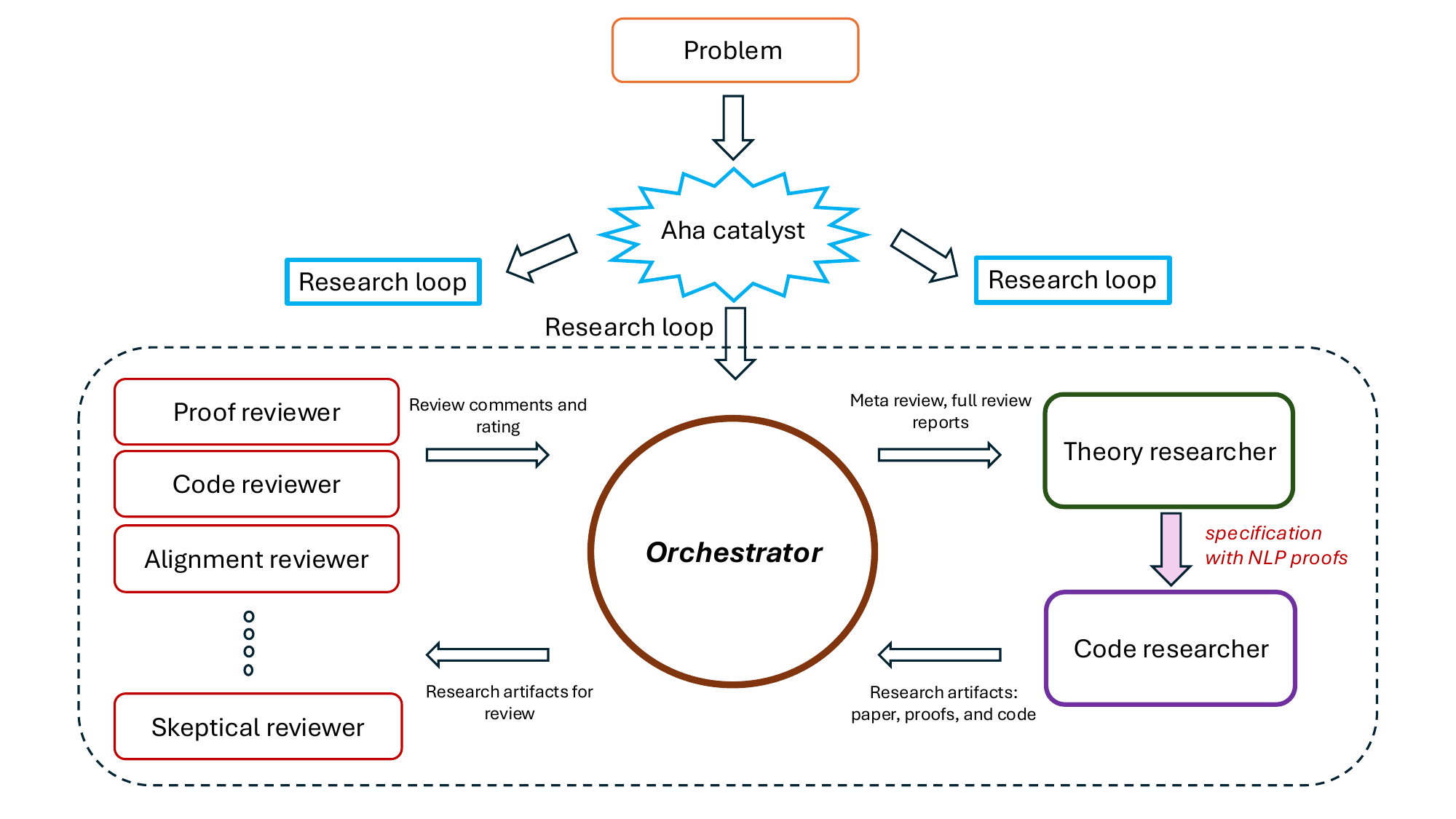}
\caption{\algname{}'s architecture: multi-agent workflow for structured NLP proof and audited code.}
\label{fig:architecture}
\end{figure}
\vspace{\AfterAgentArchitectureVSpace}

\subsection{Our Results}

\begin{insightbox}[insightgreen]{\textbf{\textcolor{black}{Key Contributions and Findings}}}
\normalsize
\begin{itemize}
\setlength{\itemsep}{0pt}
\setlength{\parskip}{0pt}
\setlength{\parsep}{0pt}
\setlength{\topsep}{0pt}
\setlength{\partopsep}{0pt}
\item We introduce \algname, an autonomous agentic research loop for provable algorithm synthesis that separates ideation, proof construction, implementation, and adversarial review.
\item For the DP set union problem, \algname{} gives a structural analysis showing that, among contractive policies, the Policy Gaussian \(\ell_2\)-descent approach is essentially the best possible. For DP $n$-gram extraction, it introduces \textbf{AFP-DPNE}, which combines \textbf{Frequency-Informed Pruning} and \textbf{Heterogeneous Thresholding} based only on previously released level-\((k-1)\) outputs. This preserves \((\epsilon,\delta)\)-DP while improving utility over standard DPNE by \(5\text{--}84\%\) across four datasets, both synthetic and real-world.
\item \algname{} also produced an important negative result for DP set union: it uncovered a proof flaw in a previously claimed result of \cite{gopi2020dpsu} by constructing a counterexample on just three items, thereby invalidating the $\ell_1$-descent algorithm. Using the auditing framework of \cite{steinke2023privacy}, it further showed empirically that the implementation fails to satisfy the claimed privacy guarantee. Together, these findings underscore the importance of strong review in iterative provable algorithm design and suggest that LLM-based agents can contribute meaningfully by surfacing subtle proof gaps and auditing implementations for correctness.
\item For the explainable clustering problem, \algname{} developed a modular transfer principle that turns any private distinct-center selector into a private explainable clustering algorithm, yielding new private guarantees for \(k\)-median and \(k\)-means almost matching the non-private counterparts. It also proved that the optimal \(1+H_{k-1}\) approximation for explainable \(k\)-median can be achieved \emph{deterministically} via dynamic programming for fixed $k$ or $d$. Finally, it extended explainable clustering guarantees to \(\ell_p^p\) objectives, although this extension may also be derivable from \cite{gupta2023price}.
\item Across all case studies, \algname{} showed strong understanding of the prior landscape, proof methods, and real opportunities for improvement even when not explicitly specified.
\end{itemize}
\end{insightbox}

These results suggest a broader agenda for coding agents: {\em proof-first code synthesis}. In much of theoretical computer science, particularly STOC/FOCS-style work, the core artifact is a precise specification rather than raw code: a paper that states the problem, assumptions, algorithm, and the invariants and guarantees needed for correctness. We argue that coding agents should follow the same template. Instead of jumping directly from prompt to implementation, they should first produce a paper-quality intermediate artifact capturing the problem statement, constraints, algorithm description, key invariants, proof sketches, and implementation-critical parameter choices, and then use that artifact to guide an auditable translation into working code. In this paradigm, code generation becomes the final compilation step of a reasoning process, rather than the primary locus of design.

\section{LLM-Generated Proofs and Self-Improvement}
Our work fits into a broader emerging landscape in which LLMs are increasingly used not only to generate code and mathematical arguments, but also to participate in iterative loops of conjecture, verification, and refinement. At one end of this spectrum are {\em fully formal proof systems}, where generated proofs are expressed in proof assistants such as Lean and can be machine checked end-to-end. Recent work has shown rapid progress in this direction, including theorem proving and olympiad-level reasoning with formal verification \cite{hubert2025olympiad,ren2025deepseek,lin2025goedel,song2024towards,achim2025aristotle}. These systems provide the strongest notion of correctness, but they typically operate in highly formalized environments and require the target reasoning task to be expressible within a proof assistant.

A second emerging regime is what one might call an {\em NLP-proof, peer-reviewed system}. Here the artifact is a conventional research output: a natural-language proof, a finished paper, and accompanying code. The proof is not machine checked, and therefore its correctness cannot be automatically guaranteed; instead, it must be audited and peer reviewed in much the same way as any other research artifact. Our system is closest to this regime. In our setting, LLM-based reviewer panels substantially improve the quality, rigor, and internal consistency of the proof and implementation, but they still do not provide absolute guarantees of correctness---just as human peer review improves reliability without ensuring infallibility. Recent works from OpenAI \cite{guevara2026single,bubeck2025early} and Google Gemini \cite{woodruff2026accelerating} suggest that LLMs can indeed contribute to research-level mathematical and algorithmic discovery in this style, although those systems were not fully autonomous.

At the other end are {\em self-evolving ML systems}, where an agent proposes modifications, receives external feedback from execution or evaluation, and uses that signal to improve future outputs. This paradigm has become increasingly prominent in recent work on open-ended or recursively improving systems \cite{zhang2025darwin,robeyns2025self,wang2023voyager,karpathy2026autoresearch}. Such systems are powerful for exploration, capability bootstrapping, and benchmark-driven improvement. However, their feedback loops are usually grounded in empirical success---for example, performance on a task, survival in an environment, or downstream evaluation score. As a result, when the desired end product is an algorithm or implementation that must satisfy explicit semantic guarantees or provable invariants, these architectures do not by themselves provide the necessary proof discipline.

Systems in the style of {\em AlphaEvolve} \cite{alphaevolve,openevolve, helix, fawzi2022alphatensor,romera2024funsearch} occupy an intermediate point in this landscape. These systems are not designed around explicit proof construction, but they do target domains where candidate solutions can often be validated through strong external checks. In some cases, such as combinatorial optimization or mathematical search problems, one can certify the quality of the final artifact or its outcome without requiring the system to produce a human-readable proof of why the construction works.
On the other hand, for specialized descendants of {\em AlphaEvolve} such as kernel and algorithmic-optimization variants \cite{liao2025kernelevolve, skydiscover2026, gepa2026optimizeanything, spector2024thunderkittens}, progress is driven primarily by benchmark improvement.
In both cases, from the perspective of LLM-based research systems, these approaches remain closer in spirit to empirical optimization: they hill-climb on scalar feedback over generated candidates.

\medskip
Our setting is different from both purely formal theorem-proving pipelines and purely empirical self-improvement loops. We study problems where a good solution is not just one that scores well empirically or survives external evaluation, but one that must come with an {\em explicit argument of correctness}: for example, a mathematical proof, a privacy guarantee, or a set of invariants that the implementation must preserve. This makes our work closer in spirit to proof-generating systems than to benchmark-driven self-improvement, but without requiring the full formalization overhead of Lean-style theorem proving. Methodologically, this also distinguishes our approach from AlphaEvolve-style search: rather than evolving raw candidates directly against a scalar objective, we place greater emphasis on the structured phases that precede implementation---idea generation, formalization of claims, proof development, and only then code generation and refinement.

More broadly, we view this structure not merely as a separate point in the design space, but as a potentially useful front-end for all of the paradigms above. By constraining the search first in the {\em idea space}---through explicit reasoning about candidate strategies, invariants, or proof structure---and only then examining how those ideas affect downstream behavior, one can guide subsequent exploration far more efficiently. In this view, mathematical analysis or proof-based reasoning serves as an upstream {\em value function} for downstream search (illustrated in Figure~\ref{fig:mathvalfunction}): it provides a way to estimate which directions are more likely to yield effective algorithms before paying the full cost of implementation or experimentation. This is especially valuable in settings where empirical evaluation is expensive, slow, or otherwise difficult to scale. For self-evolving ML systems and AlphaEvolve-style methods, such a layer could reduce search cost by helping the system reason in advance about which classes of modifications are most promising before expensive empirical evaluation. For formal systems such as Lean, it is conceivable that a structured natural-language proof sketch or semantically organized derivation could serve as a more tractable intermediate object for later formalization. And even when a full Lean proof is attainable, a natural-language proof remains essential: formal proofs are machine checkable, but often not easily human readable. NLP proofs therefore play a complementary role, making the core ideas, structure, and invariants legible to human researchers while also providing a scaffold for formal verification. In this sense, our work is not only a middle-ground approach, but also a possible organizing template for connecting proof-oriented reasoning with iterative improvement more broadly.

\begin{figure}[!t]
\centering
\includegraphics[width=0.9\textwidth,trim=110 45 80 35,clip]{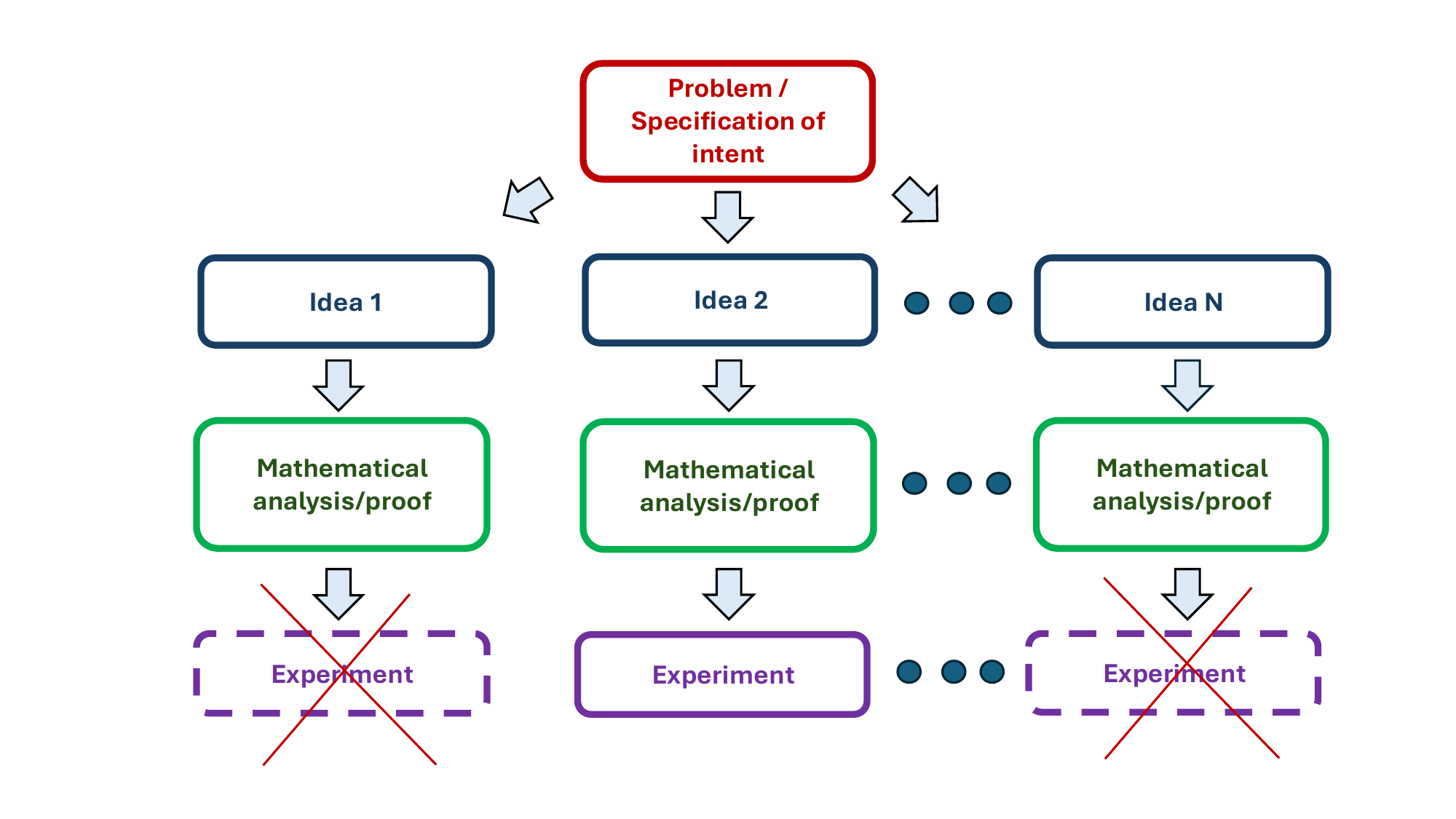}
\caption{Proof or mathematical analysis helps prune costly downstream experiments.}
\label{fig:mathvalfunction}
\end{figure}

\section{Framework Overview}

Prior multi-agent LLM systems such as CAMEL, AutoGen, MetaGPT, ChatDev, AgentVerse, and related debate-and-reflection-based frameworks have shown that explicit role decomposition, structured interaction, and iterative review can improve coordination, reasoning, and software generation \cite{camel2023,autogen2023,metagpt2023,chatdev2023,agentverse2024,chateval2023,reconcile2024,smit2024mad,bo2024reflective,wang2025mixture}.
\algname{} builds on this line but targets a substantially different objective: not general task completion or code synthesis alone, but the on-the-fly design of algorithms with provable guarantees, together with proof development, proof-guided implementation, and explicit review of proof–code alignment.

\algname{} is architecturally closest to recent agentic LLM work that combines iterative reasoning and acting \cite{yao2022react}, explicit self-feedback loops \cite{shinn2023reflexion,madaan2023selfrefine}, and role-specialized multi-agent collaboration \cite{camel2023,autogen2023,metagpt2023,chatdev2023}. Our reviewer-council design is particularly aligned with debate/council-style evaluation and deliberation protocols \cite{du2023multiagent,chateval2023}.

Our framework is inspired by the journal and conference review process that already underlies modern scientific research. At a high level, it is an agentic simulation of this research loop: some agents play the role of researchers, generating ideas, proofs, algorithmic designs, and implementations, while other agents play the role of reviewers, critiquing these artifacts from multiple specialized perspectives. The system then iterates, using reviewer feedback to refine the candidate solution until it satisfies both theoretical and practical standards. In this sense, the framework is not merely a collection of cooperating agents; it is a structured reproduction of the way rigorous research is produced, stress-tested, and improved in practice.

\medskip
The architecture of \algname{} is shown in Figure~\ref{fig:architecture}. We now turn to the major agents that instantiate this system and the roles they play.

\paragraph{Researcher Agents}

At the core of the framework is the Researcher side, which is split into two complementary roles: the Theory Researcher and the System Design and Implementation Researcher.

The Theory Researcher is responsible for the mathematical heart of the problem. Its job is to search for reformulations, identify the right abstractions, propose lemmas and invariants, reason about guarantees, and develop candidate proofs. Depending on the problem, this may include approximation-ratio arguments, probabilistic method constructions, lower and upper bounds, reductions, potential-function analyses, or structural decompositions. The motivation for separating this role is that rigorous algorithm design often hinges on delicate mathematical reasoning that is qualitatively different from implementation work. A system that does not explicitly allocate capacity to proof construction will often drift toward plausible but unjustified heuristics.

\begin{figure}[!t]
\centering
\begin{minipage}{0.98\linewidth}
\begin{insightbox}[snippetmist]{\textbf{\textcolor{black}{Excerpt from theory-researcher.md File for Theory Researcher}}}
\footnotesize
\textbf{You are the theory researcher.}

\textbf{Possible workflow modes:}
\begin{itemize}
\item theory mode: improve theory/paper only
\item full mode: improve theory/paper while coordinating with implementation work
\end{itemize}

\textbf{Inputs:}
\begin{itemize}
\item The source paper is in \texttt{input/source\_paper.tex}
\item The source PDF may also exist at \texttt{input/source\_paper.pdf}
\item The current submission draft may exist at \texttt{paper/submission.tex}
\item One or more prior revision briefs may be provided by the orchestrator:
\begin{itemize}
\item a theory revision brief
\item a full revision brief
\item occasionally a systems revision brief as supplemental context
\end{itemize}
\item The following full review reports from the most recent relevant round should be read if available:
\begin{itemize}
\item \texttt{reviews/PREV\_theory\_reviewer.md}
\item \texttt{reviews/PREV\_novelty\_reviewer.md}
\item \texttt{reviews/PREV\_skeptical\_reviewer.md}
\item in full mode, \texttt{reviews/PREV\_alignment\_reviewer.md} may also be useful when paper-code consistency affects the paper narrative
\end{itemize}
\item A system-development artifact for the current round may or may not exist. Do not assume it exists.
\end{itemize}

\textbf{Your goal:}
Produce a substantially improved research artifact that advances beyond the input paper and is credible as a top-tier conference submission.
When possible, propose and realize nontrivial technical improvements over the input paper rather than limiting yourself to repair and polishing. Prefer substantive technical improvement over presentational polish. This is not a summarization, paraphrasing, or cosmetic editing task; rewriting the input paper in cleaner prose is not sufficient. You should try to improve the research contribution itself, ideally by developing a better algorithmic idea, a stronger or cleaner proof strategy, sharper theorem statements, better bounds, clearer novelty positioning, or a more compelling technical framing. If a stronger contribution cannot be justified, then narrow claims honestly and improve correctness rather than merely restating the draft.

\textbf{Mode-specific intent:}
\begin{itemize}
\item In theory mode: focus only on theory, proofs, bounds, theorem statements, novelty framing, and paper quality.
\item In full mode: improve the theory/paper while staying aware of implementation constraints and any paper-code consistency issues that affect the paper narrative.
\end{itemize}
\end{insightbox}
\end{minipage}
\end{figure}

The System Design and Implementation Researcher takes the theoretical candidate and pushes it toward executable form. This includes deriving pseudocode, choosing data structures, identifying hidden constants and missing assumptions, designing experiments, and producing implementation plans or code. This role is essential because theoretical research papers frequently leave practical details underspecified: a proof may establish that some parameter exists, or that some subroutine can be invoked, without giving a numerically usable value or an implementable pathway. Turning such results into working artifacts is itself a substantial research task. By explicitly modeling this role, the framework aims to bridge the gap between “there exists an algorithm” and “here is an algorithm that can actually be realized, tested, and audited.”

Together, these two researcher roles prevent the framework from collapsing into either of two failure modes: elegant but non-realizable theory, or effective-looking but unjustified code. Their separation also creates a healthy internal tension: theory proposes principled structure, while implementation forces precision, explicit assumptions, and operational clarity.

\paragraph{Reviewer Agents}

Opposing the researcher side is a panel of specialized Reviewer agents. This mirrors the fact that in real research, different kinds of scrutiny are needed to evaluate different dimensions of a contribution. A single monolithic reviewer is often insufficient, especially for algorithmic work that combines proofs, novelty claims, and implementation.

A Theory Reviewer examines correctness. Its role is to find proof gaps, unstated assumptions, broken reductions, invalid implications, missing case analyses, and boundary conditions under which the claimed theorem may fail. In effect, it treats the proof as an object that must be verified under adversarial scrutiny. {\em One useful way to think about this is that proof verification resembles a form of coverage analysis: just as code review asks whether all branches and edge cases are handled, theory review asks whether all logical branches, extremal instances, and hidden assumptions are covered.}

A Novelty Reviewer evaluates whether the proposed contribution is genuinely new relative to known techniques. This is crucial because a mathematically correct argument may still be uninteresting if it is only a repackaging of standard ideas. The novelty role is especially important in automated research, where an LLM may generate polished but derivative solutions unless explicitly challenged on whether the central insight, decomposition, or guarantee is nontrivial.

A Code Reviewer inspects the implementation side. Even if the high-level algorithm is sound, the implementation may drift from the intended mathematical object, hide bugs in corner cases, or make undocumented engineering choices that undermine the guarantee. This reviewer checks whether the code or pseudocode faithfully realizes the stated algorithm, whether complexity claims are reflected in the design, and whether the implementation is robust enough to support empirical validation.

An Alignment Reviewer checks whether the produced artifact actually answers the original problem. This role is necessary because complex research loops often suffer from objective drift: an algorithm may solve a nearby problem, optimize a surrogate objective, or prove a weaker statement than the one originally required. The alignment role ensures that progress is measured against the true target, not against a convenient proxy.

A Stress-Test Reviewer probes robustness. Its purpose is to generate adversarial instances, rare edge cases, pathological parameter regimes, and counterexamples that expose fragility. In many algorithmic domains, a candidate idea appears correct on “natural” instances but breaks under carefully chosen constructions. Stress testing is therefore a critical complement to proof review and code review.

The use of multiple reviewers is motivated by several principles. First, it provides a form of test-time scaling for verification \cite{shao2025deepseekmath}: instead of relying on one pass of critique, the framework allocates additional inference budget to independent checking. Second, different review roles naturally focus on different notions of coverage, much as different forms of testing cover different classes of software failures. Third, from an LLM perspective, multiple reviewers create diversity of reasoning trajectories. This is closely related to multi-agent debate and council-style methods, where using distinct models, prompts, or role framings yields less correlated errors and more reliable aggregate judgment than relying on a single model’s self-critique \cite{du2023multiagent,autogen2023}. In short, reviewer multiplicity is not redundancy; it is a deliberate strategy for increasing verification power and reducing correlated blind spots.

\paragraph{Meta Reviewer}

The outputs of the reviewers are synthesized by a Meta Reviewer. This agent aggregates judgments across reviewers, identifies the most serious objections, resolves contradictions among reviews, and determines whether the current proposal is mature enough to advance or whether it should be sent back for revision.

This role is important because multi-agent review, while powerful, can also become noisy. Different reviewers may disagree, repeat the same criticism in different language, or overemphasize secondary issues while missing the central blocker. Without a synthesis layer, the researcher side may receive fragmented or even contradictory feedback. The Meta Reviewer acts as an editor or area chair: it turns a set of raw reviews into a coherent decision and a prioritized revision agenda. This makes the iterative loop more stable and prevents the framework from degenerating into unstructured debate.

This research loop between researchers and reviewers closely resembles test-driven software development. Classical development iterates through writing code, running tests, and fixing bugs; our system instead iterates through technical review, empirical and logical testing, new research and redesign, and renewed review, thereby refining both correctness arguments and implementations together \cite{metagpt2023,chatdev2023,agentcoder2023,bouzenia2025tests}.

\begin{figure}[!t]
\centering
\begin{minipage}{0.98\linewidth}
\begin{insightbox}[snippetmist]{\textbf{\textcolor{black}{Excerpt from theory-reviewer.md File for Theory Reviewer}}}
\footnotesize
\textbf{You are the theory reviewer.}

\textbf{Possible workflow modes:}
\begin{itemize}
\item theory mode: review only the theory/paper side
\item full mode: review the theory/paper side as part of an integrated paper+code project
\end{itemize}

\textbf{Inputs:}
\begin{itemize}
\item The current paper is \texttt{paper/submission.pdf} if available, otherwise \texttt{paper/submission.tex}
\item The current-round theory researcher artifact may exist at \texttt{artifacts/ROUND\_theory\_researcher.json}
\item One or more prior revision briefs may be provided by the orchestrator:
\begin{itemize}
\item a theory revision brief
\item a full revision brief
\item occasionally a systems revision brief as supplemental context
\end{itemize}
\item A current-round system-development artifact may or may not exist. Do not assume it exists.
\end{itemize}

\textbf{Your goal:}
Produce a rigorous mathematical review of the paper as a research artifact intended for top-tier publication.

\textbf{Mode-specific intent:}
\begin{itemize}
\item In theory mode: focus only on correctness, theorem quality, proof rigor, bounds, novelty-sensitive mathematical positioning, and significance of the theory contribution.
\item In full mode: review the theory package as part of an integrated project, but still prioritize theorem/proof correctness over implementation concerns.
\end{itemize}

\textbf{Focus on:}
\begin{itemize}
\item proof correctness
\item hidden assumptions
\item missing cases
\item unjustified transitions
\item theorem statements that are stronger than what the proof supports
\item whether the claimed bounds actually follow
\item whether imported prior-work ingredients are stated honestly and precisely
\item whether the paper's main theorem package is mathematically stable enough for publication
\end{itemize}
\end{insightbox}
\end{minipage}
\end{figure}

\paragraph{Aha Catalyst}

Above the main researcher–reviewer loop sits the Aha Catalyst, a deliberately high-level agent whose purpose is to inject conceptual entropy into the search process. Its role is not to write detailed proofs or code, but to propose surprising reformulations, analogies, proof templates, tool combinations, or decompositions that may unlock progress.

The motivation for this component is that current LLMs are often much stronger at local continuation and incremental polishing than at generating the initial seed idea that changes the direction of the search. In algorithm design, however, progress often hinges on exactly such conceptual jumps: reinterpreting the problem through duality, recasting it as a discrepancy question, relating it to a known rounding framework, or noticing that a systems bottleneck can be turned into a structural constraint. The Aha Catalyst is designed to supply these high-level pivots. It does not replace the hard work of proof and implementation; rather, it creates the intellectual diversity needed for that work to become possible.

\paragraph{Orchestrator}

The entire process is coordinated by an Orchestrator. This component manages the workflow: it routes artifacts between agents, determines which reviewer roles should be invoked, tracks unresolved objections across rounds, ensures that revisions actually address prior feedback, and maintains the structure of the iteration.

The motivation for the orchestrator is that a multi-agent research system is only as effective as its control logic. Without explicit coordination, the process can easily become unsynchronized: reviewers may inspect stale artifacts, researchers may respond to outdated critiques, or important objections may be forgotten across rounds. The Orchestrator therefore plays a combined role analogous to that of an editor, program chair, and workflow manager. It is responsible not for solving the research problem directly, but for ensuring that the right problems are surfaced, assigned, revisited, and closed in a disciplined manner.

\paragraph{Dynamic Role Generation}

A key feature of the framework is that roles are not fixed in advance. Instead, both researcher and reviewer roles can be instantiated dynamically depending on the problem class.

For example, a discrepancy-theory problem may benefit from a Probabilistic Method Researcher, a Derandomization Reviewer, or a Concentration-Inequalities Reviewer. A distributed optimization problem may call for a Communication-Complexity Reviewer, a Systems Performance Reviewer, or a Numerical Stability Reviewer. A privacy-sensitive algorithm may require a Differential Privacy Reviewer that checks sensitivity arguments, composition bounds, and privacy-accounting correctness. This dynamic extensibility allows the framework to adapt to the technical structure of the task rather than forcing every problem through the same fixed set of lenses.

The motivation here is simple: rigorous algorithmic research is highly domain-dependent. The right objections for a graph algorithm are not the same as the right objections for a private learning algorithm or a distributed systems algorithm. A fixed role set would therefore either be too narrow or too generic. Dynamic role generation gives the framework the flexibility needed to match expert scrutiny to the actual structure of the problem.

\section{Experimental Setup}
We instantiate \algname{} using multiple frontier language models, including GPT-5.4, Claude Code 4.6 with up to a 1M-token context window, and Gemini 3 Pro. Different models were used for different roles within the system, allowing us to exploit their complementary strengths in long-context reasoning, code generation, technical development, critique, and synthesis. Most experiments were conducted through Copilot CLI, which served as the primary execution interface for the agentic workflow. One limitation is that Copilot CLI is not fully transparent: it may apply internal context-management or orchestration steps, such as compaction or summarization, that are not directly exposed to the user. Accordingly, we treat these behaviors as part of the underlying systems substrate rather than as explicit components of our method.

\section{Case Studies}
In this section, we describe case studies conducted with \algname. We focus on two distinct modes of provable algorithm design:
\begin{itemize}
    \item \textbf{Algorithmic customization with guarantees.} We study the problem of adapting or improving an algorithm for a fixed dataset while preserving rigorous guarantees. Our case study is privacy-preserving algorithm design, where proof obligations are non-negotiable: privacy must hold in the worst case, including for rare inputs, adaptive attacks, and under composition.
    \item \textbf{Algorithm synthesis under new constraints.} We study settings in which the standard abstract formulation is no longer sufficient because additional real-world requirements must be imposed. In such cases, the system must synthesize new algorithms---rather than merely tune existing ones---and establish guarantees for the resulting constrained formulation. Our case study is clustering with simultaneous requirements of approximation quality, explainability, and privacy.
\end{itemize}

We organize this section as follows: for each case study, we describe the problem setting, prior approaches, and improvements discovered by \algname, then assess quality and limitations.
All artifacts of the research loop will be released on GitHub.

\subsection{Privacy Algorithm Design}
In our first case study, we consider designing improved algorithms for DP set union and DP $n$-gram extraction, following \cite{gopi2020dpsu,kim2021dpne}. These problems arise in many real-world applications, especially in natural language processing, where vocabulary discovery and synthetic data extraction are common subroutines. For example, discovering words, phrases, sentences, or more general $n$-grams from private user text can be naturally formulated as an instance of DP set union or DP $n$-gram extraction. These primitives have been used extensively in industrial synthetic data extraction pipelines.

\subsection{Differential Privacy (DP)}
Differential privacy (DP) provides a rigorous framework for data analysis with individual-level privacy guarantees. Informally, a DP algorithm is designed so that its output distribution changes only slightly when the data of any one individual is added or removed. This allows the algorithm to reveal useful aggregate information while limiting what can be inferred about any particular user. Formally:

\begin{definition}[Differential Privacy \cite{dwork2014dp}]
	A randomized algorithm $\mathcal{A}$ is $(\epsilon,\delta)$-differentially private if for any two neighboring databases $D$ and $D'$, where one can be obtained from the other by adding or removing a single user's data, and for every set $\mathcal{S}$ of possible outputs,
\[
\Pr[\mathcal{A}(D)\in \mathcal{S}] \le e^\epsilon \Pr[\mathcal{A}(D')\in \mathcal{S}] + \delta.
\]
\end{definition}

The importance of DP is that privacy is fundamentally a worst-case property, and therefore cannot be established through heuristics or empirical testing alone. An algorithm may appear safe on typical inputs yet still leak sensitive information on rare instances, under adaptive interaction, or when composed with other releases. DP addresses this by requiring a proof that privacy loss is bounded for all neighboring datasets and all output events. This proof-based guarantee is precisely what makes DP valuable in high-stakes settings: privacy is not merely hoped for, but mathematically certified.

\subsection{Set Union and \texorpdfstring{$n$}{n}-gram extraction problems}
At a high level, both DP set union and DP $n$-gram extraction ask the following question: given a collection of private user contributions, can we output as many genuinely supported objects as possible while preserving differential privacy and minimizing spurious outputs \cite{gopi2020dpsu,kim2021dpne}? In DP set union, the objects are arbitrary items contributed by users. In DP $n$-gram extraction, the objects are contiguous word sequences appearing in users' private texts. The latter can be viewed as a structured generalization of the former.

We now define these problems formally.

\begin{problem}[DP Set Union (DPSU)]
	Let $U$ be a universe of items, possibly of unbounded size. Suppose we are given a database $D$ of users, where each user $i$ contributes a subset $W_i \subseteq U$. Two databases are adjacent if they differ in exactly one user. The goal is to design an $(\epsilon,\delta)$-differentially private algorithm $A$ that outputs a subset $S \subseteq U$ such that $|S \cap \cup_i W_i|$ is as large as possible while $|S \setminus \cup_i W_i|$ is as small as possible.
\end{problem}

To define DP $n$-gram extraction, let $\Sigma$ be a vocabulary, and let
$\Sigma^k = \Sigma \times \Sigma \times \cdots \times \Sigma$ ($k$ times)
denote the set of length-$k$ sequences over $\Sigma$; elements of $\Sigma^k$ are called $k$-grams. Let
$\Sigma^* = \bigcup_{k\ge 0}\Sigma^k$
denote the set of all finite texts over $\Sigma$. For a text $w=a_1a_2\cdots a_m$ with $a_i\in \Sigma$, any contiguous subsequence $a_i a_{i+1}\cdots a_{i+k-1}$ of length $k$ is called a $k$-gram present in $w$. We write $G_k(w)$ for the set of all length-$k$ grams appearing in $w$, and $G(w)=\bigcup_k G_k(w)$ for the set of all $n$-grams of all lengths appearing in $w$.

\begin{problem}[DP $n$-gram Extraction (DPNE)]
	Let $\Sigma$ be a vocabulary, possibly of unbounded size, and let $T$ be the maximum $n$-gram length of interest. Suppose we are given a database $D$ of users, where each user $i$ contributes a text $w_i \in \Sigma^*$. Two databases are adjacent if they differ in exactly one user. The goal is to design an $(\epsilon,\delta)$-differentially private algorithm $A$ that outputs subsets $S_1,S_2,\dots,S_T$, where $S_k \subseteq \Sigma^k$, such that for each $k$, the quantity $|S_k \cap \cup_i G_k(w_i)|$ is as large as possible while the number of false positives, namely $|S_k \setminus \cup_i G_k(w_i)|$, is as small as possible.
\end{problem}

\subsection{New Results from \algnameheading{}}
We asked \algname{} to design better algorithms for these problems and improve performance on the benchmarks used in the previous papers \cite{gopi2020dpsu,kim2021dpne}. We did not explicitly give any hints to \algname{}, nor did we ask it to improve particular aspects of the algorithms; instead, we simply asked it to produce a better algorithm with a valid privacy proof.

Here is a summary of the results obtained by \algname{} for DPSU.

\begin{enumerate}
\item \textbf{Counterexample to Proposition~5.1.} We show that the $\ell_1$-descent policy is \emph{not} $\ell_2$-contractive: there exist neighboring databases where the histogram's $\ell_2$-sensitivity exceeds~$1$.
We identify the specific error in the proof: the Jacobian of the $\ell_1$-descent map is block \emph{lower triangular} (not block diagonal), and its spectral norm exceeds~$1$ whenever any item reaches the cutoff~$\Gamma$.

\item \textbf{A structural barrier in the Policy Gaussian proof.} Second, we isolate a structural barrier in the Policy Gaussian proof of Gopi et al.: when a symmetric $\ell_2$-budget-$1$ policy spends its full budget across $t$ equally situated fresh spillover items, each item can receive at most $1/\sqrt{t}$ weight, and plugging this into the support-conditioning argument recovers the corresponding per-$t$ threshold term in Policy Gaussian. This pinpoints why improving that proof step is hard on such instances: one must either reduce the weight on fresh items, change the thresholded-histogram architecture, or find a different privacy proof.
\end{enumerate}

These two results, taken together, resolve the question of which policy to use for DPSU among the contractive policies: the $\ell_2$-descent is the algorithm of choice because it is both provably correct and near-optimal, while the empirically superior $\ell_1$-descent lacks a valid privacy guarantee.

For the DPNE problem, \algname{} made the following contributions, identifying two independent opportunities for improvement within the DPNE framework, each exploiting a different structural property of the problem.

\begin{enumerate}
\item \textbf{Frequency-Informed Pruning (FIP).} We use the noisy histogram values from level $k{-}1$ (public outputs) to prune the candidate set. A $k$-gram whose prefix or suffix barely crossed the threshold at level $k{-}1$ is unlikely to be a frequent $k$-gram. FIP is post-processing---zero additional privacy cost.

\item \textbf{Heterogeneous Thresholding (HT).} We assign per-item thresholds $\tau(w)$ based on a \emph{margin score} derived entirely from level-$(k{-}1)$ \textbf{public} outputs. Crucially, $\tau(w)$ is computed \emph{before} the level-$k$ histogram $H_k$ is examined and depends only on already-released information---\emph{not} on whether $w$ appears in the current-level data. Items with high-confidence constituents receive lower thresholds; items with low-confidence constituents retain the base threshold. Because $\tau$ is a public function independent of $H_k$, the output $S_k = \{w : \tilde{H}_k[w] > \tau(w)\}$ remains post-processing of the Gaussian mechanism.
\item \textbf{Negative result: data-dependent thresholding at level \(k\) is not private.}
A tempting alternative to HT is to choose the threshold for an item \(w\) using the \emph{level-\(k\)} histogram itself, for example assigning a lower threshold when \(H_k[w] > 0\) and a higher threshold when \(H_k[w] = 0\). Unlike HT, however, this rule makes the threshold \(\tau(w)\) depend on the level-\(k\) support, which is private. We show that this dependence can violate differential privacy: for neighboring datasets, removing one user can change an item from \(H_k[w] > 0\) to \(H_k[w] = 0\), causing a discontinuous change in the threshold from \(\rho/2\) to \(\rho\). This in turn produces release probabilities that differ by far more than allowed under DP. We give a concrete counterexample in which the likelihood ratio for a fixed output event exceeds \(651\), whereas under a uniform threshold the corresponding ratio is at most \(1.2\). The key distinction is that HT uses thresholds determined only by previously released level-\((k-1)\) information, while this adaptive variant uses private level-\(k\) support information.

\item \textbf{Implementation equivalence.} We give a formal proof that our efficient implementation (which uses Binomial sampling for unobserved items) produces exactly the same output distribution as the canonical dense mechanism.

\item \textbf{Empirical verification.} We audit our code using the Steinke--Nasr--Jagielski framework \cite{steinke2023privacy}, and report the audit's limited statistical power for threshold-based mechanisms.

\item Together, FIP and HT form \textbf{AFP-DPNE} (Augmented Frequency-Pruned DPNE), which is $(\epsilon,\delta)$-DP and outperforms standard DPNE by 5--84\% across four datasets. We also disclose the trade-off: the unified mechanism increases the spurious rate from 0.6\% to 8.7\% on Reddit, an inherent cost of applying the same threshold function to all candidates.
\end{enumerate}

In the appendix, we attach the full papers produced by our system. We will release the full research loop artifacts (paper, code, reviews, and meta review for each iteration) on GitHub.

\begin{insightbox}[insightblue]{\textbf{\textcolor{black}{Summary}}}
\begin{itemize}
\item The loop produced both positive and negative scientific outcomes: a utility-improving DPNE variant and a corrective analysis for DPSU policy claims.
\item Utility gains were obtained through DP-safe post-processing and better threshold design, not by weakening privacy guarantees.
\item In this setting, reviewer pressure on proof obligations was essential; one-shot generation did not reliably identify the subtle policy-proof failure modes. Final results were the culmination of more than 10 rounds of the research loop.
\end{itemize}
\end{insightbox}

\begin{remark}
{\em The author carefully reviewed and sanity-checked the empirical improvements by the agent, and believes the overall approach is correct. However, exact empirical numbers should be interpreted with caution, as they depend on the specific implementation details, random seeds, and evaluation protocols. The main point is that the improvements are substantial and consistent across multiple settings, rather than the precise percentage gains.}
\end{remark}
\subsection{Clustering with New Constraints}

\noindent
\textbf{\(k\)-means} and \textbf{\(k\)-median} are two of the most fundamental clustering problems in algorithms, machine learning, and data analysis. In both problems, we are given a set of data points and wish to choose \(k\) representative centers so that each point is assigned to its nearest center, thereby partitioning the data into \(k\) clusters. These problems are important because they formalize the basic task of summarizing a large dataset by a small number of prototypes, and they arise in applications ranging from customer segmentation and document clustering to image compression, facility location, and vector quantization.

\medskip
\noindent
Technically, the two problems differ in the objective used to measure clustering quality. In \textbf{\(k\)-means}, typically defined in Euclidean space, the goal is to choose \(k\) centers \(c_1,\dots,c_k \in \mathbb{R}^d\) minimizing the sum of squared distances:
\[
\min_{c_1,\dots,c_k} \sum_{x \in X} \min_{i \in [k]} \|x-c_i\|_2^2.
\]
The centers may be arbitrary points in the space, and the squared-distance objective strongly penalizes faraway points. In \textbf{\(k\)-median}, the goal is to minimize:
\[
\min_{c_1,\dots,c_k} \sum_{x \in X} \min_{i \in [k]} \|x-c_i\|_1.
\]
Together, these two objectives form the standard baseline before additional real-world constraints are imposed.

\medskip

\begin{sloppypar}
The theoretical literature has developed strong baselines for these clean formulations. Strong baselines exist for Euclidean \(k\)-means \cite{kanungo2004localsearch,arthur2007kmeanspp,ahmadian2019primaldual}. Additional improvements appear in later works; see \cite{friggstad2019doubling, cohenaddad2022nested} for example. For metric \(k\)-median, key results include \cite{arya2004localsearch,li2013pseudomedian,byrka2017improved}. Later refinements include Cohen-Addad et al.\ \cite{cohenaddad2022localsearchkmedian}. Further refinements appear in Cohen-Addad et al.\ \cite{cohenaddad2025metric}. Beyond the basic objective of minimizing \(k\)-means or \(k\)-median cost, practitioners may require the clustering to satisfy several additional constraints at once. For example, the resulting partition may need to be \emph{interpretable}, so that cluster assignments can be explained by simple, human-understandable rules rather than by a complicated Voronoi decomposition; it may need strong \emph{approximation guarantees}, quantifying how much worse the constrained solution is relative to the best unconstrained clustering or standard heuristics such as Lloyd's algorithm or the optimal clustering solution; and, because clustering is often performed on sensitive user data, it may also need to satisfy \emph{differential privacy} guarantees. One may further want robustness under alternative distance measures such as \(\ell_p\) norms, or deterministic behavior for reproducibility and auditability.

Each of these desiderata—interpretability, approximation quality, privacy, robustness across norms, and determinism—has been studied extensively in isolation.
This includes explainable clustering \cite{moshkovitz2020explainable,makarychev2021explainable,gupta2023price}.
Representative results on private clustering include \cite{stemmer2018dpkmeans,ghazi2020tight,chaturvedi2021dpkmeans,cohenaddad2022private,tour2024unified}.
However, combining these requirements within a single framework remains challenging.
This highlights a core challenge in provable algorithm design: real applications rarely optimize a single clean objective, and instead demand multiple guarantees that interact in subtle ways.
Developing algorithms that are both provably sound and practically usable under such combined constraints remains a central and underexplored frontier.
\end{sloppypar}

\subsection{New Results from \algnameheading{}}

We asked \algname{} to design, implement, and document new algorithms for explainable clustering under simultaneous real-world constraints: differential privacy, determinism, and extension beyond the $\ell_1/\ell_2$ settings studied previously. The system was required not only to satisfy all constraints at once, but also to produce research-level proofs, audited code, and a clear technical writeup.

\algname{} produced the following new results:

\begin{enumerate}
    \item \textbf{Private explainable clustering via a transfer principle.}

    \algname{} identified a generic reduction showing that any private distinct-center selector can be converted into a private explainable clustering algorithm whenever the explainable conversion is data-oblivious and incurs cost inflation factor $\alpha$ (i.e., it increases expected clustering cost by at most a multiplicative factor $\alpha$ relative to the selected centers). It also explored several alternative ways of privatizing the Gupta et al.~\cite{gupta2023price} framework, and argued that this reduction gives the most principled and effective route. The result is a clean modular separation between private center selection and explainable conversion.

    \item \textbf{Concrete private guarantees and limitations.}
    Instantiating the transfer principle with the threshold-tree results of Gupta et al.~\cite{gupta2023price} gives private explainable guarantees for $k$-median and $k$-means. \algname{} also derived an explicit exponential-mechanism instantiation over finite candidate sets and a black-box corollary for efficient private $k$-median selectors. It further proved that purely multiplicative private guarantees are impossible even for $k=1$, showing that the right notion is necessarily multiplicative-plus-additive.

    \item \textbf{Deterministic realization of the tight $k$-median bound.}
    Prior work achieved the optimal $1+H_{k-1}$ factor for explainable $k$-median only via randomized analysis. \algname{} proved that this exact guarantee can in fact be realized deterministically by showing that every optimal threshold tree has an equivalent critical-threshold form and by optimizing over these trees via dynamic programming. The resulting algorithms run in polynomial time for fixed dimension $d$ or, alternatively, for fixed $k$.

    \item \textbf{Extension of explainable clustering to all $\ell_p^p$ norms for $p>2$.}
    \algname{} generalized the solo/bulk-cut framework of Gupta et al.\ to $\ell_p^p$ objectives, supplying full proofs for a result that had previously only been announced informally. It showed that the price of explainability is $O(k^{p-1}\log\log k)$ for all $p>2$, using a new threshold-cut distribution and a stretch analysis based on H\"older's inequality. It also proved tightness of the $(k-1)^{p-1}$ stretch term and clarified that the regime $1<p<2$ remains open.
\end{enumerate}

\begin{insightbox}[insightpurple]{\textbf{\textcolor{black}{Summary}}}
\begin{itemize}
    \item \algname{} discovered a clean modular view of private explainable clustering, separating private center selection from explainable conversion.
    \item It produced both positive algorithmic results---including deterministic and private guarantees---and negative results clarifying the inherent additive cost of privacy.
    \item It also extended the theory of explainable clustering beyond prior $\ell_1/\ell_2$ settings to general $\ell_p^p$ objectives for $p>2$.
\end{itemize}
\end{insightbox}

\begin{remark}
{\em The author is not a clustering expert, so it is possible that some of these results are already known, although we are not aware of an exact prior source. For the deterministic and $\ell_p$-norm results, the author carefully reviewed and sanity-checked the proofs produced by the agent, and believes the overall approach is correct. However, the proofs have not been fully verified line by line, so some technical gaps may still remain.}
\end{remark}

\section{Quality Analysis of \algnameheading{}'s Results}

We next provide a qualitative assessment of the results produced by \algname. The goal of this evaluation is not to claim a major theoretical breakthrough in isolation, but rather to assess whether the system demonstrates the combination of literature awareness, algorithmic judgment, and implementation maturity needed for practically meaningful algorithmic research.

\paragraph{Theoretical quality.}
The theoretical outputs produced by \algname{} reflect a strong understanding of the relevant research literature and an ability to synthesize prior ideas in a meaningful way. In particular, the flaw identified by the system corresponds to an issue that is also acknowledged by the original paper's authors \cite{gopi2020dpsu}, which provides evidence that the system can engage critically with existing work rather than merely restating it. The proposed algorithmic refinements are not, by themselves, deep conceptual advances, but they are technically sensible and nontrivial. More importantly, they are of the kind that could plausibly lead to measurable practical gains. In our judgment, if viewed in isolation, these theoretical contributions are roughly at the level of a solid second-tier conference submission.

\paragraph{Empirical quality.}
Given the nature of the algorithmic insights, the empirical improvements suggested by \algname{} appear credible. The implementation quality is also strong: the system was able to translate the high-level ideas into a concrete and well-structured artifact that appears suitable for evaluation in realistic settings. This is an important point, since in many algorithmic domains the gap between a mathematically plausible idea and a usable implementation is substantial.

\paragraph{Overall assessment.}
Taken together, these results highlight a capability that we view as especially important in practice: algorithmic customization and specialization to the structure, constraints, and performance requirements of a target problem, together with high-quality implementation. While the individual improvements may not constitute deep standalone research contributions, they represent exactly the kind of practically valuable progress that is often needed to adapt known algorithmic ideas to real datasets and deployment settings.

\section{Role of \algnameheading{}'s Research Loop}

A central finding of our study is that the quality of the final research output depends not only on the base model, but crucially on the surrounding research loop. The iterative interaction between a producer and a panel of reviewers consistently improved both correctness and clarity relative to what we would expect from a one-shot prompting setup. {\em For nearly every result, the final artifact emerged only after roughly ten rounds of this loop and close to a million tokens of interaction,\footnote{Because our system is built on top of Copilot CLI, exact token accounting is difficult: the interface may perform internal context-management operations such as compaction or summarization that are not fully exposed to the user. We therefore report token usage as an approximate quantity.}} which highlights that sustained critique and refinement were central to the system’s performance.

\paragraph{Initial generations were often promising but unreliable.}
In nearly all cases, early-round outputs contained substantial issues, including implementation bugs, incomplete proofs, and claims that overstated what had actually been established. This pattern suggests that raw generation is useful for exploration, but insufficient for producing research-quality artifacts without additional scrutiny.

\paragraph{Review was critical for correction and refinement.}
The reviewer panel played a central role in identifying technical flaws and improving the quality of the final result. In particular, reviewers were effective at catching unsupported claims, surfacing hidden assumptions, and identifying concrete problems in both proofs and code. The final outputs were therefore meaningfully shaped by iterative critique rather than by the producer alone.

\paragraph{Multi-persona structure provided complementary benefits.}
The value of the review stage was amplified by the use of multiple reviewer personas. Different reviewers tended to detect different classes of failure: proof gaps, implementation errors, novelty issues, and evaluation weaknesses. This diversity of perspectives was important, as many of these issues would likely have remained undetected in a vanilla single-agent prompting regime.

\paragraph{Proof-first development improved implementation discipline.}
From a systems perspective, an especially useful design choice was to require the algorithmic ideas to be written down precisely---including proofs or proof sketches---before implementation, and then to enforce close alignment between the written specification and the resulting code. This led to better code quality and sharper algorithmic descriptions, suggesting a broader principle: explicit specification-to-code alignment may be a valuable paradigm even beyond research settings, especially in domains where correctness and interpretability are important.

\section{Limitations: originality, verification, and alignment}

Despite the encouraging results, we also observed several limitations that highlight ongoing challenges for AI-driven research.

\paragraph{Limited originality beyond recombination.}
The agents were strong at surveying the literature, combining prior ideas, and refining existing approaches, but generating genuinely new conceptual ideas remained difficult. In our experience, the best results were sophisticated adaptations of known techniques rather than fundamentally new paradigms. Even when review exposed flaws, the system often preferred local repair over a more original change in direction.
\paragraph{Verification remains the central bottleneck.}
Our multi-persona review loop substantially improved quality, caught many bugs, and reduced overclaiming, but it does not guarantee full correctness. There may still be gaps in proofs or subtle mismatches between the proof and the implementation. As these systems improve at iterative generation, verification increasingly becomes the core challenge.
\paragraph{Inherent challenges in the alignment problem.}
In mathematics, correctness is objective: a proof is either valid or invalid. But when an LLM produces a proof-like argument, we often cannot tell whether an error came from an honest mistake, shallow imitation of reasoning, or optimization toward appearing correct. Thus, even in domains with crisp notions of truth, the model’s relation to truth remains opaque. This gap between truth and the appearance of justification goes to the heart of alignment.
\section{Conclusion}

Overall, our findings are largely encouraging: AI systems already seem capable of accelerating expert researchers, especially for on-the-fly algorithm synthesis, adapting worst-case methods to practical settings, and producing research implementations, though human supervision remains essential. We also find that the agentic research loop---with iterative critique and revision---substantially improves output quality. Finally, our results suggest that proof-to-code alignment is a promising general template for AI coding systems: as these systems improve, we may increasingly demand not just working code, but code accompanied by explicit invariants, arguments, evidence of correctness, and ultimately the proof.

\section*{Acknowledgments}
The author is grateful to Doug Burger, Halley Young, Aditya Nori, the DeepProof team, Hongxun Wu, Yin Tat Lee, and Sivakanth Gopi for helpful discussions and feedback on the research loop and its results.


\appendix
\clearpage
\begin{tcolorbox}[
  colback=gray!5,
  colframe=gray!60,
  fonttitle=\Large\bfseries,
  title=Appendices,
  boxrule=0.8pt,
  arc=2pt,
  left=8pt, right=8pt, top=6pt, bottom=6pt
]
The following appendices contain the full stand-alone papers produced by \algname{}.
Each appendix is a self-contained paper with its own bibliography.
We attach papers written by our agent in unaltered form.
We will release the full research loop artifacts (paper, code, reviews, and meta review for each iteration) on GitHub so that readers can see the full process and evolution of the results.

There may still be bugs in some proofs, although the author believes the high-level approaches are generally correct.
The empirical results reported in the attached papers should be interpreted with caution, as they depend on specific implementation details and evaluation protocols.
The main point is that the improvements are substantial and consistent across multiple settings, rather than the precise percentage gains.

The attached papers are included as archival artifacts rather than polished appendices.
The main purpose of including them is to provide transparency about the system's outputs and to allow others to review and build upon these results in the future.
Some references in these attached papers may be inaccurate or incomplete, because the author's bibliography-checking agent was not run for earlier versions.
\end{tcolorbox}

\bigskip

\clearpage
\section{A Counterexample to Proposition 5.1 of ``Differentially Private Set Union''}
\subsection{Summary}

We present a counterexample to Proposition~5.1 in~\cite{appA:gopi2020dpsu}, which states that the $\ell_1$-descent update policy (Algorithm~8) is $\ell_2$-contractive. We exhibit a concrete instance with 3~items and 8~users where the $\ell_2$-sensitivity of the histogram built by Algorithm~2 using Algorithm~8 equals $\approx 1.032 > 1$. We also identify the specific error in the proof.

\subsection{Setup}

We use the notation and definitions from~\cite{appA:gopi2020dpsu}. Algorithm~8 is the $\ell_1$-descent update policy for $\ell_2$-contractivity, defined as follows. Given a current histogram $H_0$, a user's item set $W$ of size at most $\Delta_0$, and a cutoff~$\Gamma$:

\begin{enumerate}
\item Compute $G[u] = \Gamma - H_0[u]$ for each $u \in W$ with $H_0[u] < \Gamma$.
\item If $\norm{G}_{\ell_2} \le 1$: do not update ($H_1 = H_0$).
\item Else: find $\lambda \ge 0$ such that $\sum_{u} \min\{G[u], \lambda\}^2 = 1$, and set $H_1[u] = H_0[u] + \min\{G[u], \lambda\}$.
\end{enumerate}

Proposition~5.1 claims this policy is $\ell_2$-contractive, which by Theorem~3.1 would imply the histogram built by Algorithm~2 has $\ell_2$-sensitivity at most~$1$.

\subsection{Counterexample}

\begin{claim}
There exist neighboring databases $D_1, D_2$ and a cutoff $\Gamma$ such that Algorithm~2 with the $\ell_1$-descent policy (Algorithm~8) produces histograms $H_1, H_2$ with $\norm{H_1 - H_2}_{\ell_2} > 1$.
\end{claim}

\paragraph{Parameters.} Items: $\{a, b, c\}$. Cutoff $\Gamma = 5$. Maximum contribution $\Delta_0 = 3$.

\paragraph{Databases.} The databases share 7~users and $D_1$ has one additional user:
\begin{align*}
D_1 &= [\underbrace{\{a,b,c\}}_{\text{extra}},\ \{a,b\},\ \{a,b\},\ \{b,c\},\ \{b,c\},\ \{b,c\},\ \{b,c\},\ \{b,c\}] \\
D_2 &= [\{a,b\},\ \{a,b\},\ \{b,c\},\ \{b,c\},\ \{b,c\},\ \{b,c\},\ \{b,c\}]
\end{align*}
Users are processed in the order listed (i.e., the hash function induces this ordering).

\paragraph{Execution trace.} Write $\frac{1}{\sqrt{k}}$ as the exact value throughout. All arithmetic below can be verified by hand.

\medskip
\noindent\textbf{User~0 ($W = \{a,b,c\}$, $D_1$ only).} The histogram is $H = (0,0,0)$. Gaps: $G = (5,5,5)$. Since $\norm{G}_2 = 5\sqrt{3} > 1$, we solve $3\lambda^2 = 1$, giving $\lambda = \frac{1}{\sqrt{3}}$. Each item is incremented by $\frac{1}{\sqrt{3}}$.
\[
H_1 = \left(\tfrac{1}{\sqrt{3}},\, \tfrac{1}{\sqrt{3}},\, \tfrac{1}{\sqrt{3}}\right), \qquad H_2 = (0, 0, 0), \qquad H_1 - H_2 = \left(\tfrac{1}{\sqrt{3}},\, \tfrac{1}{\sqrt{3}},\, \tfrac{1}{\sqrt{3}}\right), \quad \norm{\cdot}_2 = 1.
\]

\noindent\textbf{Users~1--2 ($W = \{a, b\}$, both databases).} At each step, both items $a, b$ have gap $> \frac{1}{\sqrt{2}}$. So $\lambda = \frac{1}{\sqrt{2}}$ and both items are incremented by $\frac{1}{\sqrt{2}}$. The identical increment applies to both $H_1$ and $H_2$, so the difference vector stays $\left(\frac{1}{\sqrt{3}}, \frac{1}{\sqrt{3}}, \frac{1}{\sqrt{3}}\right)$ with $\ell_2$-norm $= 1$.

\noindent After user~2:
\[
H_1 = \left(\tfrac{1}{\sqrt{3}} + \tfrac{2}{\sqrt{2}},\; \tfrac{1}{\sqrt{3}} + \tfrac{2}{\sqrt{2}},\; \tfrac{1}{\sqrt{3}}\right), \quad
H_2 = \left(\tfrac{2}{\sqrt{2}},\; \tfrac{2}{\sqrt{2}},\; 0\right).
\]

\noindent\textbf{Users~3--6 ($W = \{b, c\}$, both databases).} Similarly, both gaps exceed $\frac{1}{\sqrt{2}}$, so each step increments $b$ and $c$ by $\frac{1}{\sqrt{2}}$ in both histograms. The difference remains $\left(\frac{1}{\sqrt{3}}, \frac{1}{\sqrt{3}}, \frac{1}{\sqrt{3}}\right)$, $\norm{\cdot}_2 = 1$.

\noindent After user~6:
\begin{align*}
H_1 &= \left(\tfrac{1}{\sqrt{3}} + \tfrac{2}{\sqrt{2}},\;\; \tfrac{1}{\sqrt{3}} + \tfrac{6}{\sqrt{2}},\;\; \tfrac{1}{\sqrt{3}} + \tfrac{4}{\sqrt{2}}\right) \approx (1.992,\; 4.820,\; 3.406) \\
H_2 &= \left(\tfrac{2}{\sqrt{2}},\;\; \tfrac{6}{\sqrt{2}},\;\; \tfrac{4}{\sqrt{2}}\right) \approx (1.414,\; 4.243,\; 2.828)
\end{align*}

\medskip
\noindent\textbf{User~7 ($W = \{b, c\}$, both databases) --- the critical step.}

\noindent\emph{Processing $H_1$:} The gaps are:
\[
G_b^{(1)} = 5 - \left(\tfrac{1}{\sqrt{3}} + \tfrac{6}{\sqrt{2}}\right) \approx 0.180, \qquad G_c^{(1)} = 5 - \left(\tfrac{1}{\sqrt{3}} + \tfrac{4}{\sqrt{2}}\right) \approx 1.594.
\]
Since $G_b^{(1)} < G_c^{(1)}$, we solve $\left(G_b^{(1)}\right)^2 + \lambda^2 = 1$, giving $\lambda = \sqrt{1 - \left(G_b^{(1)}\right)^2} \approx 0.984$. Since $G_b^{(1)} < \lambda$, item~$b$ gets $G_b^{(1)}$ (reaching $\Gamma = 5$), and item~$c$ gets $\lambda \approx 0.984$.

\noindent\emph{Processing $H_2$:} The gaps are:
\[
G_b^{(2)} = 5 - \tfrac{6}{\sqrt{2}} \approx 0.757, \qquad G_c^{(2)} = 5 - \tfrac{4}{\sqrt{2}} \approx 2.172.
\]
Both gaps exceed $\frac{1}{\sqrt{2}} \approx 0.707$. We solve $2\lambda^2 = 1$, giving $\lambda = \frac{1}{\sqrt{2}}$. Both items are incremented by $\frac{1}{\sqrt{2}}$.

\medskip
\noindent\textbf{Final state after user~7:}
\begin{align*}
H_1[b] &= \Gamma = 5, & H_2[b] &= \tfrac{7}{\sqrt{2}} \approx 4.950 \\
H_1[c] &\approx 3.406 + 0.984 = 4.389, & H_2[c] &= \tfrac{5}{\sqrt{2}} \approx 3.536 \\
H_1[a] &\approx 1.992, & H_2[a] &= \tfrac{2}{\sqrt{2}} \approx 1.414
\end{align*}

\noindent The difference vector is:
\[
H_1 - H_2 \approx (0.577,\; 0.050,\; 0.854).
\]
\[
\boxed{\norm{H_1 - H_2}_{\ell_2} \approx 1.032 > 1.}
\]

\paragraph{Mechanism of the violation.} Through users~0--6, the difference is uniformly $(\frac{1}{\sqrt{3}}, \frac{1}{\sqrt{3}}, \frac{1}{\sqrt{3}})$ with $\ell_2$-norm exactly~$1$, because both $H_1$ and $H_2$ receive identical increments. At user~7, item~$b$ reaches $\Gamma$ in $H_1$ (its gap $G_b^{(1)} \approx 0.180$ is fully consumed) but \emph{not} in $H_2$ (gap $G_b^{(2)} \approx 0.757$ is truncated to $\frac{1}{\sqrt{2}}$). This asymmetry causes the $\ell_1$-descent to allocate budgets differently: $H_1$ directs residual budget to item~$c$, while $H_2$ splits evenly. The net effect increases $\norm{H_1 - H_2}_2$ beyond~$1$.

\subsection{The error in the proof}

The proof of Proposition~5.1 computes the Jacobian of the map $F(x) = \text{Algorithm~8}(x)$ and writes:
\[
J_F(x) = \begin{bmatrix} \mathbf{0} & \mathbf{0} \\ * & I \end{bmatrix}
\]
where the top block corresponds to coordinates $i \le t$ (items reaching $\Gamma$, where $F(x)_i = \Gamma$) and the bottom block to coordinates $i > t$ (where $F(x)_i = x_i + \lambda$). The proof calls this ``block diagonal form'' and concludes $\norm{J_F}_{S_\infty} \le 1$.

The matrix is block \textbf{lower triangular}, not block diagonal. The off-diagonal block $*$ has entries
\[
\frac{\partial \lambda}{\partial x_j} = \frac{\Gamma - x_j}{(d-t)\lambda}, \qquad j \le t,
\]
obtained by differentiating the constraint $\sum_{j \le t}(\Gamma - x_j)^2 + (d-t)\lambda^2 = 1$. When $d = 2$ (items $b, c$) and $t = 1$ (item $b$ reaches $\Gamma$), the Jacobian is:
\[
J_F = \begin{bmatrix} 0 & 0 \\ \frac{G_b}{\lambda} & 1 \end{bmatrix}, \qquad
J_F^T J_F = \begin{bmatrix} \frac{G_b^2}{\lambda^2} & \frac{G_b}{\lambda} \\[4pt] \frac{G_b}{\lambda} & 1 \end{bmatrix}.
\]
The eigenvalues of $J_F^T J_F$ are $0$ and $1 + \frac{G_b^2}{\lambda^2}$, giving
\[
\norm{J_F}_{S_\infty} = \sqrt{1 + \frac{G_b^2}{\lambda^2}} = \frac{1}{\sqrt{1 - G_b^2}} > 1 \quad \text{whenever } G_b > 0.
\]
At the critical step (user~7 processing $H_1$), $G_b \approx 0.180$, so $\norm{J_F}_{S_\infty} \approx 1.016 > 1$.

\subsection{Scope of the issue}

\begin{itemize}
\item \textbf{Proposition~5.1 (Algorithm~8, $\ell_1$-descent for $\ell_2$-contractivity):} False, as shown above.
\item \textbf{Lemma~5.2 (Algorithm~7, $\ell_2$-descent for $\ell_2$-contractivity):} \emph{Correct.} Its proof (Lemma~5.1, geometric argument) is independent of the Jacobian and remains valid. We verified $\ell_2$-contractivity of the $\ell_2$-descent empirically over $10^4$ simulations with max expansion ratio $< 1$.
\item \textbf{Lemma~3.1 (Algorithm~4, $\ell_1$-descent for $\ell_1$-contractivity):} \emph{Correct.} Its proof is structural (monotonicity + $\ell_1$-budget) and does not use the Jacobian.
\item \textbf{Policy Gaussian with $\ell_2$-descent (Algorithm~7):} Privacy proof is valid.
\item \textbf{Policy Gaussian with $\ell_1$-descent (Algorithm~8):} Privacy proof has a gap due to Proposition~5.1. The $\ell_2$-sensitivity may exceed~$1$, so the noise calibration in Theorem~4.2 may be insufficient.
\end{itemize}

\clearpage
\section{On the Limits of Contractive Policies for Differentially Private Set Union}
\begin{appendixabstract}
We study the differentially private set union (DPSU) problem introduced by Gopi et al.\ (ICML 2020), where the goal is to output as large a subset of the union $\cup_i W_i$ as possible under $(\eps,\delta)$-differential privacy.
We make two contributions.
First, we identify a gap in the proof that the $\ell_1$-descent update policy (Algorithm~8 of the original paper) is $\ell_2$-contractive: we produce a concrete counterexample---3 items, 8 users, cutoff $\Gamma=5$---where the $\ell_2$-sensitivity of the resulting histogram is approximately $1.032>1$, and pinpoint the error in the Jacobian argument (Proposition~5.1).
Second, we isolate a structural barrier in the Policy Gaussian proof of Gopi et al.: when a symmetric $\ell_2$-budget-$1$ policy spends its full budget across $t$ equally situated fresh spillover items, each item can receive at most $1/\sqrt{t}$ weight, and plugging this into the support-conditioning argument recovers the corresponding per-$t$ threshold term in Policy Gaussian. This pinpoints why improving that proof step is hard on such instances: one must either reduce the weight on fresh items, change the thresholded-histogram architecture, or find a different privacy proof.
Together, these results clarify the landscape of the DPSU problem: among the two policies analyzed by Gopi et al., the $\ell_2$-descent is the one with a correct $\ell_2$-privacy proof, but genuine optimality beyond that comparison remains open.
\end{appendixabstract}

\subsection{Introduction}
\label{appB:sec:intro}

Differentially private set union (DPSU) is a fundamental primitive in private data analysis.
Given a database of $n$ users, where each user $i$ holds a subset $W_i$ of items from a universe $U$, the goal is to output a large subset $S \subseteq \bigcup_i W_i$ while satisfying $(\eps,\delta)$-differential privacy.
This problem arises in vocabulary extraction, $n$-gram discovery, search query release, and many natural language processing applications~\cite{appB:KKMN09,appB:WilsonZ20,appB:gopi2020dpsu}.

Gopi et al.~\cite{appB:gopi2020dpsu} introduced the \emph{contractive policy} framework. In particular, an $\ell_2$-contractive update policy yields a histogram with $\ell_2$-sensitivity at most $1$, which can then be released privately via the Gaussian mechanism.
They proposed two policies: the $\ell_2$-descent (Algorithm~7) and the $\ell_1$-descent (Algorithm~8).
The $\ell_1$-descent achieves the best empirical performance and is claimed to be $\ell_2$-contractive in Proposition~5.1.

\paragraph{Our contributions.}
\begin{enumerate}
\item \textbf{Counterexample to Proposition~5.1.} We show that the $\ell_1$-descent policy is \emph{not} $\ell_2$-contractive: there exist neighboring databases where the histogram's $\ell_2$-sensitivity exceeds~$1$ (Section~\ref{appB:sec:counterexample}).
We identify the specific error in the proof: the Jacobian of the $\ell_1$-descent map is block \emph{lower triangular} (not block diagonal), and its spectral norm exceeds~$1$ whenever any item reaches the cutoff~$\Gamma$.

\item \textbf{Barrier for Policy Gaussian.} We isolate the main bottleneck in the Policy Gaussian proof (Section~\ref{appB:sec:optimality}): for $t$ equally situated fresh spillover items, symmetry and an $\ell_2$-budget of $1$ force per-item weight at most $1/\sqrt{t}$. On the corresponding equal-spillover/full-budget instances, the $\ell_2$-descent attains this bound. Thus any substantial improvement of that proof step within the same support-conditioning template must either intentionally give fresh items less weight or rely on a different mechanism or privacy proof.
\item \textbf{Numerical illustration.} We quantify the resulting spillover surcharge on a Zipf benchmark. The surcharge is modest for smaller values of $\eps$, but it grows as the Gaussian noise shrinks.
\end{enumerate}

These results suggest that, among Gopi et al.'s two named contractive policies, the $\ell_2$-descent is the one currently supported by a correct $\ell_2$-privacy proof and matches the spillover barrier analyzed in Section~\ref{appB:sec:optimality} on the corresponding equal-spillover/full-budget instances, whereas the empirically stronger $\ell_1$-descent lacks a valid $\ell_2$-privacy guarantee.

\subsection{Preliminaries}
\label{appB:sec:prelims}

\begin{definition}[Differential Privacy~\cite{appB:DMNS06,appB:DR14}]
A randomized algorithm $\mathcal{A}$ is $(\eps,\delta)$-differentially private if for all neighboring databases $D\sim D'$ (differing in one user) and all measurable sets $\mathcal{S}$:
$\Pr[\mathcal{A}(D)\in \mathcal{S}] \le e^\eps \Pr[\mathcal{A}(D')\in \mathcal{S}]+\delta.$
\end{definition}

\begin{definition}[DPSU Problem]
Given $n$ users with subsets $W_i\subseteq U$ and a contribution bound $\Delta_0 = \max_i |W_i|$, design an $(\eps,\delta)$-DP algorithm outputting $S\subseteq \bigcup_i W_i$ with $|S|$ as large as possible.
\end{definition}

\begin{proposition}[Gaussian Mechanism~\cite{appB:BalleW18}]
\label{appB:prop:gaussian}
If $f$ has $\ell_2$-sensitivity $\Delta_2$, the mechanism $M(D) = f(D)+Z$, $Z\sim\cN(0,\sigma^2 I)$, is $(\eps,\delta)$-DP if and only if
$\Phi\bigl(\frac{\Delta_2}{2\sigma}-\frac{\eps\sigma}{\Delta_2}\bigr)-e^\eps\Phi\bigl(-\frac{\Delta_2}{2\sigma}-\frac{\eps\sigma}{\Delta_2}\bigr)\le\delta.$
\end{proposition}

We write $\sigma^*(\eps,\delta,\Delta_2)$ for the minimum $\sigma$ satisfying the Balle--Wang inequality above.

\paragraph{The contractive policy framework~\cite{appB:gopi2020dpsu}.}
Algorithm~2 of the original paper builds a weighted histogram $H$ by iterating over users in a random hash order.
Each user $i$ updates the histogram using an \emph{update policy} $\phi$ with $\ell_2$-budget~1 (i.e., $\norm{\phi(H,W_i)-H}_2\le 1$).
If $\phi$ is $\ell_2$-\emph{contractive}---meaning there exists an invariant set $\cI\subset\{(H_1,H_2):\norm{H_1-H_2}_2\le 1\}$ satisfying the two properties in Definition~1.3 of~\cite{appB:gopi2020dpsu}---then the histogram has $\ell_2$-sensitivity~$\le 1$ (Theorem~1.1).
Adding Gaussian noise $\cN(0,\sigma^2)$ and thresholding at $\rho$ yields Algorithm~1, the meta-algorithm for DPSU.

\paragraph{Two policies.}
\begin{itemize}
\item \textbf{$\ell_2$-descent} (Algorithm~7): move $H|_W$ toward $(\Gamma,\ldots,\Gamma)$ by $\ell_2$-distance~$\le 1$.
Proved $\ell_2$-contractive in Lemma~5.2 of~\cite{appB:gopi2020dpsu} via a geometric argument (Lemma~5.1).
\item \textbf{$\ell_1$-descent} (Algorithm~8): maximize $\sum_u y_u$ subject to $y_u\le\Gamma$ and $\norm{y-H|_W}_2\le 1$.
\emph{Claimed} $\ell_2$-contractive in Proposition~5.1 of~\cite{appB:gopi2020dpsu} via a Jacobian argument.
\end{itemize}

\paragraph{Privacy of the Policy Gaussian algorithm (Theorem~5.1 of~\cite{appB:gopi2020dpsu}).}
With an $\ell_2$-contractive symmetric policy, sensitivity~$\le 1$, noise parameter
\begin{equation}\label{appB:eq:sigma}
\sigma = \sigma^*(\eps, \delta_{\mathrm{mech}}, 1),
\end{equation}
and threshold
\begin{equation}\label{appB:eq:rho}
\rho = \max_{1\le t\le \Delta_0}\Bigl(\frac{1}{\sqrt{t}}+\sigma\,\Phi^{-1}\bigl((1-\delta_{\mathrm{spill}})^{1/t}\bigr)\Bigr),
\end{equation}
where $\delta_{\mathrm{mech}}+\delta_{\mathrm{spill}}=\delta$, the algorithm is $(\eps,\delta)$-DP.
The paper sets $\delta_{\mathrm{mech}}=\delta_{\mathrm{spill}}=\delta/2$.

\subsection{Counterexample to Proposition~5.1}
\label{appB:sec:counterexample}

\begin{theorem}[Proposition~5.1 of~\cite{appB:gopi2020dpsu} is false]
\label{appB:thm:counterexample}
The $\ell_1$-descent policy (Algorithm~8 of~\cite{appB:gopi2020dpsu}) is not $\ell_2$-contractive.
In particular, there exist neighboring databases $D_1,D_2$ and a cutoff $\Gamma$ such that the histograms $H_1,H_2$ produced by Algorithm~2 with the $\ell_1$-descent policy satisfy $\norm{H_1-H_2}_2>1$.
\end{theorem}

\begin{proof}
We exhibit a concrete instance. Let $U=\{a,b,c\}$, $\Gamma=5$, $\Delta_0=3$, and define:
\begin{align*}
D_1 &= [\{a,b,c\},\ \{a,b\},\ \{a,b\},\ \{b,c\},\ \{b,c\},\ \{b,c\},\ \{b,c\},\ \{b,c\}],\\
D_2 &= [\{a,b\},\ \{a,b\},\ \{b,c\},\ \{b,c\},\ \{b,c\},\ \{b,c\},\ \{b,c\}].
\end{align*}
$D_1$ has one extra user (first position) with $W=\{a,b,c\}$; the remaining 7 users are identical and processed in the same order. We trace the execution of Algorithm~8 exactly.

\medskip
\noindent\textbf{User~0} ($W=\{a,b,c\}$, $D_1$ only). Histogram is $(0,0,0)$. Gaps: $G=(5,5,5)$, $\norm{G}_2 = 5\sqrt{3}>1$. Solve $3\lambda^2=1$: $\lambda=1/\sqrt{3}$. Each item incremented by $1/\sqrt{3}$.
$$H_1=(1/\sqrt{3},\,1/\sqrt{3},\,1/\sqrt{3}),\quad H_2=(0,0,0),\quad \norm{H_1-H_2}_2=1.$$

\noindent\textbf{Users~1--2} ($W=\{a,b\}$, both databases). Both gaps exceed $1/\sqrt{2}$. Solve $2\lambda^2=1$: $\lambda=1/\sqrt{2}$. Each of $a,b$ incremented by $1/\sqrt{2}$ in both histograms. Difference unchanged:
$$H_1-H_2=(1/\sqrt{3},\,1/\sqrt{3},\,1/\sqrt{3}),\quad \norm{H_1-H_2}_2=1.$$

\noindent\textbf{Users~3--6} ($W=\{b,c\}$, both databases). Same analysis: $\lambda=1/\sqrt{2}$, items $b,c$ incremented by $1/\sqrt{2}$ in both. After user~6:
\begin{align*}
H_1 &= (1/\sqrt{3}+2/\sqrt{2},\ 1/\sqrt{3}+6/\sqrt{2},\ 1/\sqrt{3}+4/\sqrt{2}) \approx (1.992,\,4.820,\,3.406),\\
H_2 &= (2/\sqrt{2},\ 6/\sqrt{2},\ 4/\sqrt{2}) \approx (1.414,\,4.243,\,2.828).
\end{align*}
Difference: $(1/\sqrt{3},\,1/\sqrt{3},\,1/\sqrt{3})$, $\norm{\cdot}_2=1$.

\medskip
\noindent\textbf{User~7} ($W=\{b,c\}$, both databases --- the critical step).

\noindent\emph{$H_1$:} Gaps $G_b^{(1)}=5-(1/\sqrt{3}+6/\sqrt{2})\approx 0.180$, $G_c^{(1)}=5-(1/\sqrt{3}+4/\sqrt{2})\approx 1.594$.
Since $G_b^{(1)}<G_c^{(1)}$, solve $(G_b^{(1)})^2+\lambda^2=1$: $\lambda=\sqrt{1-(G_b^{(1)})^2}\approx 0.984$.
Item $b$ gets $\min(G_b^{(1)},\lambda)=G_b^{(1)}\approx 0.180$ (reaching $\Gamma$).
Item $c$ gets $\lambda\approx 0.984$.

\noindent\emph{$H_2$:} Gaps $G_b^{(2)}=5-6/\sqrt{2}\approx 0.757$, $G_c^{(2)}=5-4/\sqrt{2}\approx 2.172$.
Both gaps $>1/\sqrt{2}$, so $\lambda=1/\sqrt{2}$. Both items get $1/\sqrt{2}\approx 0.707$.

\medskip
\noindent\textbf{Final state:}
\begin{align*}
H_1-H_2 &\approx (0.577,\ 0.050,\ 0.854),\\
\norm{H_1-H_2}_2 &= \sqrt{0.577^2+0.050^2+0.854^2}\approx \boxed{1.032>1}.
\end{align*}

\noindent\textbf{Mechanism of violation.} Item $b$ reaches $\Gamma$ in $H_1$ (small residual gap 0.180) but not in $H_2$ (gap 0.757). The $\ell_1$-descent allocates $b$'s saved budget to $c$ in $H_1$ (increment 0.984) but not in $H_2$ (increment 0.707). The asymmetric allocation increases the $\ell_2$-distance.

\medskip
Since this exhibits $\norm{H_1-H_2}_2>1$ for neighboring databases, no invariant set $\cI$ satisfying Definition~1.3 of~\cite{appB:gopi2020dpsu} can exist (Property~2 forces the initial pair into $\cI$, Property~1 keeps all subsequent pairs in $\cI$, and membership in $\cI$ requires $\norm{\cdot}_2\le 1$). Therefore Algorithm~8 is not $\ell_2$-contractive.
\end{proof}

\subsubsection{The error in the proof}
\label{appB:sec:error}

The proof of Proposition~5.1 in~\cite{appB:gopi2020dpsu} computes the Jacobian of the map $F(x)=\argmax\sum_i y_i$ s.t.\ $y_i\le\Gamma$, $\norm{y-x}_2\le 1$. At a generic point $x$ with $\Gamma>x_1>\cdots>x_d$, if $F(x)_i=\Gamma$ for $i\le t$ and $F(x)_i=x_i+\lambda$ for $i>t$, then
$$J_F(x) = \begin{bmatrix} \mathbf{0} & \mathbf{0} \\ B & I_{d-t}\end{bmatrix}$$
where $B_{ij}=\frac{\partial\lambda}{\partial x_j}=\frac{\Gamma-x_j}{(d-t)\lambda}$ for $j\le t$.
The proof states this has ``block diagonal form'' and concludes $\norm{J_F}_{S_\infty}\le 1$.

The matrix is block \emph{lower triangular}, not block diagonal. Its singular values satisfy:
\begin{proposition}
\label{appB:prop:jacobian}
For $t\ge 1$ and $\lambda$ satisfying $\sum_{j=1}^t(\Gamma-x_j)^2+(d-t)\lambda^2=1$:
$$\norm{J_F(x)}_{S_\infty}=\frac{1}{\sqrt{(d-t)\lambda^2}}=\frac{1}{\sqrt{1-\sum_{j=1}^t(\Gamma-x_j)^2}}>1.$$
\end{proposition}
\begin{proof}
Write $m=d-t$. Each row of $B$ equals the vector
\[
c=\left(\frac{\Gamma-x_1}{m\lambda},\ldots,\frac{\Gamma-x_t}{m\lambda}\right).
\]
Hence $BB^\top=\norm{c}_2^2\,\mathbf{1}\mathbf{1}^\top$, which has eigenvalues $0$ with multiplicity $m-1$ and $m\norm{c}_2^2=\sum_{j\le t}(\Gamma-x_j)^2/(m\lambda^2)$ with multiplicity~1. Since
\[
J_F(x)J_F(x)^\top=
\begin{bmatrix}
\mathbf{0} & \mathbf{0}\\
\mathbf{0} & I_m+BB^\top
\end{bmatrix},
\]
its largest eigenvalue is
\[
1+\frac{\sum_{j\le t}(\Gamma-x_j)^2}{m\lambda^2}
=\frac{\sum_{j\le t}(\Gamma-x_j)^2+m\lambda^2}{m\lambda^2}
=\frac{1}{m\lambda^2}.
\]
Therefore $\norm{J_F(x)}_{S_\infty}=\sqrt{\lambda_{\max}(J_FJ_F^\top)}=1/\sqrt{m\lambda^2}$. Since $m\lambda^2=1-\sum_{j\le t}(\Gamma-x_j)^2<1$ whenever $t\ge 1$, the norm is strictly larger than~$1$.
\end{proof}

\begin{remark}
The \emph{eigenvalues} of the lower triangular matrix $J_F$ are $\{0,\ldots,0,1,\ldots,1\}$, all $\le 1$. The likely error was concluding that the \emph{spectral norm} (largest singular value) is also $\le 1$, conflating eigenvalues with singular values.
\end{remark}

\begin{remark}[Scope]
The $\ell_2$-descent (Algorithm~7) \emph{is} $\ell_2$-contractive: its proof (Lemma~5.2 of~\cite{appB:gopi2020dpsu}, via Lemma~5.1) uses a geometric argument that does not involve the Jacobian. Lemma~4.1 ($\ell_1$-descent is $\ell_1$-contractive) is also correct, as it uses a structural monotonicity argument.
\end{remark}

\subsection{A Barrier within the Policy Gaussian Framework}
\label{appB:sec:optimality}

Having established that the $\ell_2$-descent is the policy among Gopi et al.'s two proposals with a valid $\ell_2$-contractivity proof, we now ask why the Policy Gaussian mechanism is difficult to improve upon. We do \emph{not} claim an information-theoretic lower bound over all DPSU algorithms. Instead, we isolate the bottleneck in the proof of Theorem~5.1 of~\cite{appB:gopi2020dpsu}: the spillover items created by the extra user.

\subsubsection{Symmetric spillover weights}

\begin{proposition}[Symmetric spillover budget]
\label{appB:prop:spillover_budget}
Fix $t\ge 1$. Consider $t$ previously unseen items whose current histogram values are equal, and let $v\in\R_{\ge 0}^t$ be the increment assigned to these items by a symmetric update policy with $\norm{v}_2\le 1$. Then all coordinates of $v$ are equal to a common value $c$, and
\[
0\le c\le \frac{1}{\sqrt{t}}.
\]
Moreover, $c=1/\sqrt{t}$ if and only if the policy spends its entire $\ell_2$-budget on these $t$ items.
\end{proposition}
\begin{proof}
By symmetry, all $t$ coordinates are equal, so $v=(c,\ldots,c)$ for some $c\ge 0$. The budget constraint gives $t c^2=\norm{v}_2^2\le 1$, hence $c\le 1/\sqrt{t}$. Equality holds exactly when $\norm{v}_2=1$, i.e., when the full budget is spent on these items.
\end{proof}

\subsubsection{The support-conditioning bottleneck}

\begin{theorem}[Spillover bound in the support-conditioning proof]
\label{appB:thm:spillover_barrier}
Consider the proof template of Theorem~5.1 of~\cite{appB:gopi2020dpsu}. Let $D_1\supset D_2$ be neighboring databases, let $H_1,H_2$ be their histograms, let $T=\operatorname{supp}(H_1)\setminus\operatorname{supp}(H_2)$ with $|T|=t$, and let the mechanism add independent $Z_u\sim\cN(0,\sigma^2)$ noise and output the released set
\[
A=\{u\in \operatorname{supp}(H_1): H_1[u]+Z_u>\rho\}.
\]
If $H_1[u]\le b_t$ for every $u\in T$, then for the event $E=\{A\subseteq \operatorname{supp}(H_2)\}$,
\[
\Pr[\bar E]\le 1-\Phi\!\left(\frac{\rho-b_t}{\sigma}\right)^t.
\]
Consequently, this proof template certifies $\Pr[\bar E]\le \delta_{\mathrm{spill}}$ whenever
\[
\rho \ge b_t+\sigma\Phi^{-1}\!\left((1-\delta_{\mathrm{spill}})^{1/t}\right).
\]
\end{theorem}
\begin{proof}
Let $p_u=\Pr[H_1[u]+Z_u\le \rho]$. Since $H_1[u]\le b_t$ and the Gaussian CDF is increasing,
\[
p_u=\Phi\!\left(\frac{\rho-H_1[u]}{\sigma}\right)\ge \Phi\!\left(\frac{\rho-b_t}{\sigma}\right)
\qquad\text{for every }u\in T.
\]
Because the noises are independent,
\[
\Pr[\bar E]
=\Pr[\exists u\in T:H_1[u]+Z_u>\rho]
=1-\prod_{u\in T}p_u
\le 1-\Phi\!\left(\frac{\rho-b_t}{\sigma}\right)^t.
\]
Solving the right-hand side inequality
\[
1-\Phi\!\left(\frac{\rho-b_t}{\sigma}\right)^t\le \delta_{\mathrm{spill}}
\]
for $\rho$ gives the stated threshold condition.
\end{proof}

\begin{corollary}[Why beating Policy Gaussian is hard]
\label{appB:cor:policy_gaussian_barrier}
Within the support-conditioning proof of Theorem~5.1 of~\cite{appB:gopi2020dpsu}, any symmetric $\ell_2$-budget-$1$ policy that fully spends its budget on $t$ equally situated spillover items yields the threshold requirement
\[
\rho \ge \frac{1}{\sqrt{t}}+\sigma\Phi^{-1}\!\left((1-\delta_{\mathrm{spill}})^{1/t}\right).
\]
The $\ell_2$-descent attains this case on the corresponding equal-spillover/full-budget instances; for example, if the extra user's active coordinates are exactly those $t$ fresh items and their distance to the cutoff vector $(\Gamma,\ldots,\Gamma)$ is at least $1$, then the update spends its full unit budget and splits it uniformly across them. Therefore, within this proof template and this equal-spillover/full-budget regime, there is no way to improve the spillover term without either:
\begin{enumerate}
\item giving some fresh items less than $1/\sqrt{t}$ weight, which lowers their probability $1-\Phi((\rho-c)/\sigma)$ of crossing a fixed threshold;
\item abandoning coordinatewise thresholding of an i.i.d.-Gaussian-perturbed histogram; or
\item finding a privacy proof that does not condition on the support event $E$.
\end{enumerate}
\end{corollary}
\begin{proof}
By Proposition~\ref{appB:prop:spillover_budget}, a symmetric budget-$1$ policy can assign at most $1/\sqrt{t}$ to each of $t$ equally situated spillover items, and this value is attained exactly when the budget is split uniformly across them. Theorem~\ref{appB:thm:spillover_barrier} then gives the stated threshold. For the $\ell_2$-descent attainment example, the gap vector on those fresh coordinates is a constant multiple of $\mathbf{1}$; if its $\ell_2$ norm is at least $1$, then a unit step in the normalized gap direction adds exactly $1/\sqrt{t}$ to each coordinate. The last sentence follows because $1-\Phi((\rho-c)/\sigma)$ is increasing in $c$: decreasing the spillover weight reduces the chance that a fresh item survives the thresholding step.
\end{proof}

\begin{remark}
This is a barrier theorem, not an information-theoretic lower bound over all differentially private set-union mechanisms. It explains why the Policy Gaussian threshold is hard to improve inside the current symmetric support-conditioning proof template for a histogram-plus-Gaussian-threshold architecture, but it leaves open the possibility of fundamentally different mechanisms or privacy proofs.
\end{remark}

\subsection{Experiments}
\label{appB:sec:experiments}

We numerically illustrate the size of the spillover surcharge from Corollary~\ref{appB:cor:policy_gaussian_barrier}. The benchmark below is \emph{not} a lower bound on all DPSU mechanisms; it is simply the threshold one would obtain from the same support-conditioning calculation if the spillover coordinates carried zero histogram mass.

\subsubsection{Setup}

We generate a synthetic Zipf dataset: $20{,}000$ users, universe of $50{,}000$ items, popularity $\propto 1/\mathrm{rank}^{1.07}$, user set sizes following a truncated Pareto distribution ($30{,}531$ unique items, $13{,}796$ singletons). All experiments use $\delta=\exp(-10)$.

\subsubsection{Spillover surcharge}

Define the zero-mass benchmark
\[
\rho_{\mathrm{zero}}=\sigma\Phi^{-1}\!\left((1-\delta_{\mathrm{spill}})^{1/\Delta_0}\right).
\]
Table~\ref{appB:tab:gap} reports the Policy Gaussian threshold $\rho_{\mathrm{PG}}$, the benchmark $\rho_{\mathrm{zero}}$, and their difference. This should be read as a benchmark comparison: when the maximizer in $\rho_{\mathrm{PG}}$ occurs at $t=\Delta_0$, the difference equals the fixed-$t$ spillover surcharge from Corollary~\ref{appB:cor:policy_gaussian_barrier}; otherwise it is simply the gap between the worst-case Policy Gaussian threshold and the zero-mass reference.

\begin{table}[h]
\caption{Policy Gaussian threshold and the zero-mass benchmark.}
\label{appB:tab:gap}
\centering
\small
\begin{tabular}{cc|cccc}
\toprule
$\eps$ & $\Delta_0$ & $\rho_{\mathrm{PG}}$ & $\rho_{\mathrm{zero}}$ & Surcharge & Relative\\
\midrule
1.0 & 10 & 16.56 & 16.25 & 0.32 & 1.9\%\\
1.0 & 100 & 17.97 & 17.87 & 0.10 & 0.6\%\\
3.0 & 10 & 6.44 & 6.11 & 0.32 & 5.0\%\\
3.0 & 100 & 6.82 & 6.72 & 0.10 & 1.5\%\\
5.0 & 10 & 4.50 & 3.94 & 0.56 & 12.5\%\\
5.0 & 100 & 4.50 & 4.33 & 0.17 & 3.8\%\\
8.0 & 10 & 3.37 & 2.66 & 0.71 & 21.0\%\\
8.0 & 100 & 3.37 & 2.93 & 0.44 & 13.0\%\\
\bottomrule
\end{tabular}
\end{table}

For smaller privacy loss ($\eps\le 3$), the surcharge is modest in these experiments. For larger $\eps$, the Gaussian noise shrinks and the fixed spillover mass becomes more visible, so the surcharge grows.

\subsubsection{Output size against a zero-mass benchmark}

Table~\ref{appB:tab:output} reports the number of items output by Policy Gaussian, compared with the same histogram thresholded at $\rho_{\mathrm{zero}}$. The second column is not privacy-certified by our argument; it is included only as a benchmark for the utility cost of the spillover surcharge.

\begin{table}[h]
\caption{Output size: Policy Gaussian vs.\ the zero-mass benchmark. Zipf data, $20{,}000$ users.}
\label{appB:tab:output}
\centering
\small
\begin{tabular}{cc|cc|c}
\toprule
$\eps$ & $\Delta_0$ & $|S|$ (Policy Gaussian) & $|S|$ (benchmark) & Gap (\%) \\
\midrule
1.0 & 10 & 400 & 403 & 0.7\%\\
1.0 & 100 & 417 & 420 & 0.7\%\\
3.0 & 10 & 1091 & 1104 & 1.2\%\\
3.0 & 100 & 1193 & 1218 & 2.1\%\\
5.0 & 10 & 1625 & 1670 & 2.8\%\\
5.0 & 100 & 1931 & 1977 & 2.4\%\\
\bottomrule
\end{tabular}
\end{table}

On this dataset, the output-size gap stays within a few percent. This does not prove optimality, but it suggests that the spillover surcharge can be small in practical regimes even though it is structurally built into the current proof template.

\subsection{Discussion: Why Better Policies Are Hard}
\label{appB:sec:discussion}

Our results identify one concrete obstruction to improving DPSU within the presently validated contractive framework. We discuss three natural approaches and where this barrier arises.

\paragraph{Better contractive policies.}
The $\ell_2$-descent moves in the direction $G/\norm{G}_2$ where $G=(\Gamma,\ldots,\Gamma)-H|_W$, allocating budget proportionally to gaps. The $\ell_1$-descent fills items greedily (smallest gap first), outputting more items per user. However, Theorem~\ref{appB:thm:counterexample} shows the $\ell_1$-descent is not $\ell_2$-contractive.
One intermediate variant does collapse to the ordinary $\ell_2$-descent: iterative $\ell_2$-descent with a fixed cutoff and micro-steps, because while no item reaches $\Gamma$ the gap vector remains a positive scalar multiple of its previous value. This collinearity argument does not generally extend to escalating cutoffs, where changing $\Gamma$ adds an all-ones term to the gap vector and can change the normalized direction, so those variants require a separate analysis.

\paragraph{Using $\ell_1$-descent with honest sensitivity.}
One repair strategy would be to use the $\ell_1$-descent with its true $\ell_2$-sensitivity $C>1$ and recalibrate the Gaussian noise accordingly. This increases the Gaussian-noise component of the Policy Gaussian threshold. The spillover term need not scale in the same simple way, so the net threshold effect requires a separate analysis; nevertheless, any utility gain from greedier weight allocation must be balanced against this larger noise penalty.

\paragraph{Non-histogram mechanisms.}
Our barrier does not rule out sketches, candidate-set methods, or adaptive query strategies. The point is more specific: on the equal-spillover/full-budget instances considered in Section~\ref{appB:sec:optimality}, any approach that still releases a coordinatewise thresholded Gaussian histogram and proves privacy by conditioning on the revealed support inherits the same spillover bottleneck in that proof step. Whether one can obtain comparable improvements outside that regime remains open.

\paragraph{Summary.}
The spillover term is the price of discovering items from a data-dependent support inside the current Policy Gaussian proof. In the equal-spillover/full-budget regime, a symmetric policy with $\ell_2$-budget $1$ can place at most $1/\sqrt{t}$ mass on each of $t$ equally situated fresh items, and the $\ell_2$-descent realizes this tradeoff. Any further improvement for that regime within the same proof template must either reduce fresh-item weight and accept lower recall, or avoid the same support-conditioning architecture entirely.

\subsection{Related Work}
\label{appB:sec:related}

The DPSU problem was formalized by~\cite{appB:gopi2020dpsu}, who introduced the contractive policy framework. Earlier work on releasing search queries~\cite{appB:KKMN09} and SQL keys~\cite{appB:WilsonZ20} considered related problems but with different privacy formulations.
For the special case $\Delta_0=1$, \cite{appB:desfontaines2020} gave a nearly optimal algorithm based on noisy partition selection.
The $\ell_1$-descent policy has been used in differentially private $n$-gram extraction~\cite{appB:kim2021dpne}; our counterexample (Theorem~\ref{appB:thm:counterexample}) implies that the privacy guarantees in that work may need re-examination.
The Gaussian mechanism calibration uses the exact characterization of~\cite{appB:BalleW18}. Our barrier argument reuses the support-conditioning step from Theorem~5.1 of~\cite{appB:gopi2020dpsu}; unlike a genuine lower bound, it is specific to that proof template.

\subsection{Conclusion}
\label{appB:sec:conclusion}

We made two contributions to the DPSU problem. First, we found and verified a counterexample to the claim that the $\ell_1$-descent policy is $\ell_2$-contractive (Proposition~5.1 of~\cite{appB:gopi2020dpsu}), and identified the exact proof error (confusion of singular values and eigenvalues of a block triangular Jacobian). This corrects the record on which DPSU algorithms have valid privacy proofs.

Second, we identified a structural barrier inside the Policy Gaussian proof. For $t$ equally situated spillover items, symmetry and an $\ell_2$-budget of $1$ force per-item weight at most $1/\sqrt{t}$, and the $\ell_2$-descent attains this bound on the corresponding equal-spillover/full-budget instances. Thus improving that part of the Policy Gaussian threshold analysis on such instances requires leaving the current support-conditioning template, changing the mechanism architecture, or sacrificing weight on fresh items.

This does not prove information-theoretic optimality of the $\ell_2$-descent over all DPSU mechanisms. Rather, it explains why the particular room for improvement isolated by our barrier theorem is hard to access with the techniques currently available.


\clearpage
\section{Heterogeneous Thresholding for Differentially Private $n$-gram Extraction}
\begin{appendixabstract}
We study differentially private $n$-gram extraction (DPNE), where the goal is to
discover frequent $n$-grams from user text data under $(\eps,\delta)$-differential privacy.
Building on the DPNE framework of Kim et al.\ (NeurIPS 2021), we introduce
\textbf{AFP-DPNE}, which incorporates two zero-cost innovations:
(1)~\textbf{Frequency-Informed Pruning (FIP)} uses the publicly released noisy counts
from level $k{-}1$ to shrink the candidate set at level $k$; and
(2)~\textbf{Heterogeneous Thresholding (HT)} assigns per-item acceptance thresholds
at levels $k \ge 2$ using margin scores derived entirely from level-$(k{-}1)$ public
outputs.
Both exploit a key structural property of DPNE: the thresholds $\rho_k$ for $k \ge 2$
do \emph{not} affect privacy---only utility---because the output at every level is
post-processing of the Gaussian mechanism.

On real Reddit data (100K users, 147K vocabulary), AFP-DPNE discovers
\textbf{84.3\%} more $n$-grams than the baseline at $\eps=4$, and \textbf{26.4\%}
more at $\eps=1$.  On three synthetic datasets, improvements range from 5--42\%.
We provide a rigorous privacy proof (valid for all $\eps > 0$, $\delta \in (0,1)$)
and a formal implementation-equivalence argument showing the code exactly matches
the canonical dense mechanism.

We also present a \textbf{negative result}: an ``adaptive'' variant that assigns
\emph{different} thresholds to observed vs.\ unobserved items is \emph{not}
differentially private, with a concrete counterexample showing a $651\times$ privacy
loss amplification.
This delineates the boundary between safe and unsafe threshold manipulation in DPNE.
We verify our implementation with the one-training-run privacy audit of
Steinke et al.\ (NeurIPS 2024).
\end{appendixabstract}

\subsection{Introduction}
\label{appC:sec:intro}

Differentially private $n$-gram extraction is a fundamental primitive for privacy-preserving NLP \cite{appC:kim2021dpne}.  Given a database of $N$ users, each contributing text data, the goal is to release a large collection of frequent $n$-grams (contiguous word subsequences) while guaranteeing $(\eps,\delta)$-differential privacy at the user level.  Extracted $n$-grams power downstream applications such as sentence completion, smart reply \cite{appC:KannanK16}, and assisted writing \cite{appC:ChenL19} that are trained on potentially sensitive data \cite{appC:CarliniL19}.

Kim et al.\ \cite{appC:kim2021dpne} introduced DPNE, which iteratively extracts $k$-grams for $k = 1, 2, \ldots, T$ using the differentially private set union (DPSU) framework of Gopi et al.\ \cite{appC:gopi2020dpsu}.  At each level, a weighted histogram is constructed with $\ell_2$-sensitivity~1, Gaussian noise is added, and items exceeding a threshold $\rho_k$ are released.  The crucial structural insight is \emph{hierarchical pruning}: at level~$k$, only $k$-grams whose constituent $(k{-}1)$-grams were already discovered are considered.

Despite this, the DPNE baseline leaves substantial utility on the table.  At level~$k \ge 2$, the thresholds $\rho_k$ are set to control the expected number of spurious (zero-frequency) items, but they do not affect privacy---because the output $S_k = \{w : \tilde{H}_k[w] > \rho_k\}$ is a deterministic function of the Gaussian mechanism's output $\tilde{H}_k$.  This observation opens the door to \emph{non-uniform} thresholds at zero privacy cost.

\paragraph{Contributions.}

\begin{enumerate}[leftmargin=*]
\item \textbf{Frequency-Informed Pruning (FIP).}  We use the noisy histogram values from level $k{-}1$ (public outputs) to prune the candidate set.  A $k$-gram whose prefix or suffix barely crossed the threshold at level $k{-}1$ is unlikely to be a frequent $k$-gram.  FIP is post-processing---zero additional privacy cost.

\item \textbf{Heterogeneous Thresholding (HT).}  We assign per-item thresholds $\tau(w)$ based on a \emph{margin score} derived entirely from level-$(k{-}1)$ \textbf{public} outputs.  Crucially, $\tau(w)$ is computed \emph{before} the level-$k$ histogram $H_k$ is examined and depends only on already-released information---\emph{not} on whether $w$ appears in the current-level data.  Items with high-confidence constituents receive lower thresholds; items with low-confidence constituents retain the base threshold.  Because $\tau$ is a public function independent of $H_k$, the output $S_k = \{w : \tilde{H}_k[w] > \tau(w)\}$ remains post-processing of the Gaussian mechanism.

\item \textbf{Negative result: adaptive thresholding is not private.}  A natural alternative to HT would be to condition the threshold on the \emph{current-level} data: give observed items ($H_k[w] > 0$) a low threshold and unobserved items ($H_k[w] = 0$) a high threshold.  Unlike HT, this ``adaptive'' variant makes $\tau(w)$ depend on $\supp(H_k)$, which is \textbf{private} information.  We prove this breaks privacy: removing a user can flip an item from observed to unobserved, simultaneously shifting its threshold from $\rho/2$ to $\rho$---a jump far too large for the noise to absorb.  We give a concrete counterexample where the privacy loss ratio exceeds 651, compared to $\le 1.2$ for a uniform threshold.  The key distinction is: \emph{HT's thresholds depend on public level-$(k{-}1)$ outputs; adaptive thresholds depend on private level-$k$ support information.}

\item \textbf{Implementation equivalence.}  We give a formal proof that our efficient implementation (which uses Binomial sampling for unobserved items) produces exactly the same output distribution as the canonical dense mechanism.  A previous version using bucket quantization is shown to break this equivalence; we correct it by grouping on exact threshold values.

\item \textbf{Empirical verification.}  We audit our code using the Steinke--Nasr--Jagielski one-training-run framework \cite{appC:steinke2024audit}, and honestly report the audit's limited statistical power for threshold-based mechanisms.
\end{enumerate}

Together, FIP and HT form \textbf{AFP-DPNE} (Augmented Frequency-Pruned DPNE), which is $(\eps,\delta)$-DP and outperforms standard DPNE by 5--84\% across four datasets.  We also honestly disclose the tradeoff: the unified mechanism increases the spurious rate from 0.6\% to 8.7\% on Reddit, an inherent cost of applying the same threshold function to all candidates.

\subsection{Preliminaries}
\label{appC:sec:prelims}

\begin{definition}[Differential Privacy \cite{appC:DMNS06}]
\label{appC:def:dp}
A randomized algorithm $\cA$ is $(\eps,\delta)$-differentially private if for any two
neighboring databases $D, D'$ (differing in one user's data) and all measurable
$\cS$:
\[
\Pr[\cA(D) \in \cS] \;\le\; e^\eps \Pr[\cA(D') \in \cS] + \delta.
\]
\end{definition}

\begin{definition}[Gaussian Mechanism \cite{appC:DR14,appC:BalleW18}]
\label{appC:def:gaussian}
For a query $f : \cD \to \R^d$ with $\ell_2$-sensitivity $\Delta_2 = \max_{D \sim D'} \|f(D) - f(D')\|_2$, the Gaussian mechanism adds noise $Z \sim \cN(0, \sigma^2 I_d)$ and releases $f(D) + Z$.  By \cite{appC:BalleW18}, this is $(\eps,\delta)$-DP if $\sigma = \sigma^*(\eps, \delta, \Delta_2)$ satisfies
\begin{equation}
\label{appC:eq:bw}
\Phi\!\left(\frac{\Delta_2}{2\sigma} - \frac{\eps\sigma}{\Delta_2}\right)
- e^\eps \Phi\!\left(-\frac{\Delta_2}{2\sigma} - \frac{\eps\sigma}{\Delta_2}\right) \le \delta.
\end{equation}
\end{definition}

\begin{proposition}[Post-processing \cite{appC:DR14}]
\label{appC:prop:postproc}
If $\cA$ is $(\eps,\delta)$-DP and $f$ is any (possibly randomized) function, then $f \circ \cA$ is $(\eps,\delta)$-DP.
\end{proposition}

\begin{proposition}[GDP Composition \cite{appC:DongRS19_GDP}]
\label{appC:prop:composition}
If mechanisms $\cA_1, \ldots, \cA_T$ are run sequentially on the same dataset, where $\cA_k$ is a Gaussian mechanism with noise $\sigma_k$ and $\ell_2$-sensitivity $\Delta_2 = 1$, then the composed mechanism is equivalent to a single Gaussian mechanism with noise $\sigma^*$ where $1/\sigma^{*2} = \sum_{k=1}^T 1/\sigma_k^2$.
\end{proposition}

\paragraph{The DPNE framework \cite{appC:kim2021dpne}.}

The algorithm iterates for $k = 1, 2, \ldots, T$:

\smallskip\noindent
\textbf{Level 1 (DPSU).}
Each user $i$ has a set of unigrams $U_i$.  After truncating each $U_i$ to at most $\Delta_1$ items, define the weighted histogram $H_1[u] = \sum_{i : u \in U_i'} 1/\sqrt{|U_i'|}$, where $U_i' \subseteq U_i$ is the truncated set.

\begin{lemma}[$\ell_2$-Sensitivity of Weighted Histogram]
\label{appC:lem:sensitivity}
$\|H_1(D) - H_1(D')\|_2 \le 1$ for neighboring $D, D'$.
\end{lemma}
\begin{proof}
Removing user $i$ affects only their contribution.  Let $t = |U_i'|$.  Then:
$\|H_1(D) - H_1(D')\|_2^2 = \sum_{u \in U_i'} (1/\sqrt{t})^2 = t \cdot (1/t) = 1.$
\end{proof}

Add Gaussian noise: $\tilde{H}_1[u] = H_1[u] + Z_u$, $Z_u \sim \cN(0, \sigma_1^2)$.
Output $S_1 = \{u : \tilde{H}_1[u] > \rho_1\}$.

\smallskip\noindent
\textbf{Level $k \ge 2$.}
Define candidate set $V_k = \{(p, u) : p \in S_{k-1}, u \in S_1, (u, p[2{:}]) \in S_{k-1}\}$ (both prefix and suffix must be in $S_{k-1}$).  Build $H_k$ on $V_k$ with the same weighted scheme, add noise, and output $S_k = \{w \in V_k : \tilde{H}_k[w] > \rho_k\}$.

\begin{proposition}[Privacy of DPNE {\cite[Prop.~4.1]{appC:kim2021dpne}}]
\label{appC:prop:dpne_privacy}
DPNE is $(\eps, \delta)$-DP if: (i)~$\sigma^*$ satisfies~\eqref{appC:eq:bw} for $(\eps, \delta/2)$-DP with $\Delta_2 = 1$; (ii)~$\sigma_k = \sigma^*\sqrt{T}$ for all $k$ (uniform budget); and (iii)~$\rho_1$ is set so that the probability of any unique item from a removed user appearing in $S_1$ is $\le \delta/2$.

\textbf{Crucially, $\rho_2, \ldots, \rho_T$ do not affect privacy---only utility.}
\end{proposition}

The intuition: at levels $k \ge 2$, the output $S_k = \{w : \tilde{H}_k[w] > \rho_k\}$ is a deterministic function of the Gaussian mechanism's output $\tilde{H}_k$. By post-processing (Proposition~\ref{appC:prop:postproc}), it inherits the Gaussian mechanism's DP guarantee regardless of how $\rho_k$ is chosen, as long as the choice does not depend on $\tilde{H}_k$ or $H_k$ at the \emph{current} level.

\subsection{Algorithm: AFP-DPNE}
\label{appC:sec:algorithm}

\subsubsection{Frequency-Informed Pruning (FIP)}
\label{appC:sec:fip}

Standard DPNE constructs $V_k$ using only the \emph{membership} information $g \in S_{k-1}$.
We additionally use the \emph{noisy count} $\tilde{H}_{k-1}[g]$, which is a public output.

\begin{definition}[Margin]
\label{appC:def:margin}
For $g \in S_{k-1}$, the margin is $\text{margin}(g) = \tilde{H}_{k-1}[g] - \rho_{k-1}$.
For a $k$-gram $w$ with prefix $\text{pref}(w) = (w_1, \ldots, w_{k-1})$ and suffix
$\text{suf}(w) = (w_2, \ldots, w_k)$:
\[
\text{margin}(w) = \min\bigl(\text{margin}(\text{pref}(w)),\;\text{margin}(\text{suf}(w))\bigr).
\]
Items not in $\tilde{H}_{k-1}$ are imputed with $\tilde{H}_{k-1}[g] = \rho_{k-1}$
(margin $= 0$).
\end{definition}

\begin{definition}[FIP-Pruned Candidate Set]
\label{appC:def:fip}
Given tolerance $m \ge 0$:
\begin{equation}
\label{appC:eq:fip}
V_k' = \{w \in V_k : \text{margin}(w) > -m \cdot \sigma_{k-1}\}.
\end{equation}
\end{definition}

\begin{proposition}[Privacy of FIP]
\label{appC:prop:fip_privacy}
FIP is a deterministic function of $(S_{k-1}, \tilde{H}_{k-1}, \rho_{k-1}, \sigma_{k-1})$, all of which are public outputs of level $k{-}1$ or algorithm constants.
By Proposition~\ref{appC:prop:postproc}, FIP incurs zero additional privacy cost.
\end{proposition}

\paragraph{Intuition.}
If a $(k{-}1)$-gram barely crossed the threshold ($\text{margin} \approx 0$), it is likely
a low-frequency item boosted by noise.  Any $k$-gram built from it is unlikely to be
genuinely frequent.  FIP removes such candidates, yielding $|V_k'| \le |V_k|$.  Since
the spurious-control threshold $\rho_k$ increases with $|V_k|$, smaller $|V_k'|$ enables
a lower $\rho_k$ and hence more genuine discoveries.

\subsubsection{Heterogeneous Thresholding (HT)}
\label{appC:sec:ht}

For $k \ge 2$, $\rho_k$ does not affect privacy (Proposition~\ref{appC:prop:dpne_privacy}).
Standard DPNE uses a \emph{uniform} threshold $\rho_k$ for all candidates.  We generalize this
to \emph{per-item} thresholds.

\begin{definition}[HT Threshold]
\label{appC:def:ht}
For $w \in V_k'$, define:
\begin{equation}
\label{appC:eq:tau}
\tau(w) = \rho_k^{\text{base}} - \underbrace{\min\!\Bigl(\gamma \cdot \frac{\max(0,\; \text{margin}(w))}{\bar{c}},\;\; \tfrac{1}{2}\rho_k^{\text{base}}\Bigr)}_{d(w) \;\ge\; 0}
\end{equation}
where:
\begin{itemize}[leftmargin=*]
\item $\rho_k^{\text{base}} = \sigma_k \Phi^{-1}\!\bigl(1 - \eta \min(|S_{k-1}|/|V_k'|,\; 1)\bigr)$ is the standard spurious-control threshold (Kim et al.\ \cite[Prop.~4.1]{appC:kim2021dpne}),
\item $\gamma \in [0,1]$ is the discount factor,
\item $\bar{c} = \text{median}\bigl(\{\text{margin}(w) : w \in V_k',\; \text{margin}(w) > 0\}\bigr)$ is the median positive margin (defaults to $1.0$ if no positive margins exist),
\item $d(w)$ is the discount, capped at $\rho_k^{\text{base}}/2$.
\end{itemize}
\end{definition}

\begin{proposition}[HT is Public]
\label{appC:prop:ht_public}
The function $\tau : V_k' \to \R$ depends on:
(i)~margins, computed from $(S_{k-1}, \tilde{H}_{k-1}, \rho_{k-1})$ (public level-$(k{-}1)$ outputs);
(ii)~$|S_{k-1}|$ and $|V_k'|$ (determined by level-$(k{-}1)$ outputs and $S_1$);
and (iii)~hyperparameters $(\gamma, \eta)$.
None of these quantities depend on the level-$k$ histogram $H_k$ or its noisy version $\tilde{H}_k$.
Therefore $\tau$ is a \textbf{public} threshold function.
\end{proposition}

Note that items without a positive margin ($\text{margin}(w) \le 0$) receive $d(w) = 0$, hence $\tau(w) = \rho_k^{\text{base}}$---the standard uniform threshold.

\subsubsection{Full Algorithm}

AFP-DPNE combines FIP and HT:
\begin{enumerate}[leftmargin=*]
\item At level 1: standard DPSU with spillover threshold $\rho_1$.
\item At level $k \ge 2$: FIP prunes $V_k$ to $V_k'$; compute $\tau(w)$ for all $w \in V_k'$ (before seeing $H_k$); build $H_k$ on $V_k'$; add noise; output $S_k = \{w : \tilde{H}_k[w] > \tau(w)\}$.
\end{enumerate}
See Algorithm~\ref{appC:alg:afp_dpne} in Appendix~\ref{appC:app:pseudocode} for complete pseudocode.

\subsection{Privacy Analysis}
\label{appC:sec:privacy}

\subsubsection{Main Theorem}

\begin{definition}[Canonical AFP-DPNE at level $k \ge 2$]
\label{appC:def:canonical}
Given candidate set $V_k'$, margins $\{\text{margin}(w)\}_{w \in V_k'}$ (public),
and histogram $H_k$ with $\ell_2$-sensitivity~1:
\begin{enumerate}
\item Compute $\tau(w)$ via~\eqref{appC:eq:tau} for all $w \in V_k'$.
\item Draw noise $Z_w \sim \cN(0, \sigma_k^2)$ independently for each $w \in V_k'$.
\item Compute $\tilde{H}_k[w] = H_k[w] + Z_w$ for all $w \in V_k'$.
\item Output $S_k = \{w \in V_k' : \tilde{H}_k[w] > \tau(w)\}$.
\end{enumerate}
\end{definition}

\begin{theorem}[Privacy of AFP-DPNE]
\label{appC:thm:privacy}
AFP-DPNE is $(\eps, \delta)$-differentially private for all $\eps > 0$ and $\delta \in (0,1)$.
\end{theorem}

\begin{proof}
We prove privacy by combining three components.

\textbf{Part 1: Gaussian composition.}
At each level $k \in [T]$, the histogram $H_k$ has $\ell_2$-sensitivity~1 (Lemma~\ref{appC:lem:sensitivity}).  The noisy histogram $\tilde{H}_k = H_k + Z_k$ (where $Z_k \sim \cN(0, \sigma_k^2 I)$) is the output of a Gaussian mechanism with noise~$\sigma_k$.

For levels $k \ge 2$: The threshold function $\tau : V_k' \to \R$ is determined \emph{before} level-$k$ data is processed (Proposition~\ref{appC:prop:ht_public}).  Therefore $S_k = \{w : \tilde{H}_k[w] > \tau(w)\}$ is a deterministic function $f_\tau(\tilde{H}_k)$ of the Gaussian mechanism output.  By Proposition~\ref{appC:prop:postproc}, $S_k$ inherits the Gaussian mechanism's privacy guarantee.

FIP prunes $V_k$ to $V_k'$ using only level-$(k{-}1)$ public outputs (Proposition~\ref{appC:prop:fip_privacy}), so it also incurs zero cost.

Setting $\sigma_k = \sigma^*\sqrt{T}$ for all $k$: by Proposition~\ref{appC:prop:composition},
\[
\frac{1}{\sigma^{*2}} = \sum_{k=1}^T \frac{1}{\sigma_k^2} = \sum_{k=1}^T \frac{1}{\sigma^{*2} T} = \frac{1}{\sigma^{*2}}.
\]
The composed mechanism is equivalent to a single Gaussian mechanism with noise $\sigma^*$ and sensitivity~1, which is $(\eps, \delta/2)$-DP by the Balle--Wang formula~\eqref{appC:eq:bw}.

\textbf{Part 2: Spillover at level 1.}
When user $i$ is removed from dataset $D$ to form $D'$, user $i$'s unique items in $H_1$ could reveal their presence.  Suppose user $i$ contributes $t \le \Delta_1$ unique items, each with weight $w = 1/\sqrt{t}$.  Item $u$ appears in $S_1$ iff $w + Z_u > \rho_1$, i.e., $Z_u > \rho_1 - 1/\sqrt{t}$.

We set:
\begin{equation}
\label{appC:eq:rho1}
\rho_1 = \max_{1 \le t \le \Delta_1}\;\Bigl(\frac{1}{\sqrt{t}} + \sigma_1 \Phi^{-1}\!\bigl((1 - \delta/2)^{1/t}\bigr)\Bigr).
\end{equation}

For each $t$: $\Prob[u \text{ passes}] = \Phi\bigl((1/\sqrt{t} - \rho_1)/\sigma_1\bigr) \le 1 - (1-\delta/2)^{1/t}$.  By a union bound over $t$ items: $\Prob[\text{any passes}] \le 1 - (1 - \delta/2)^{t/t} = \delta/2$.  Taking the max over $t$ ensures this holds for all possible contribution sizes.

\textbf{Part 3: Combining.}
The Gaussian composition gives $(\eps, \delta/2)$-DP.  The spillover gives a $(0, \delta/2)$ leakage at level~1.  By basic composition:
\[
\text{AFP-DPNE is } (\eps,\; \delta/2 + \delta/2) = (\eps, \delta)\text{-DP.} \qedhere
\]
\end{proof}

\subsubsection{Implementation Equivalence}
\label{appC:sec:impl_equiv}

\begin{proposition}[Implementation = Canonical Mechanism]
\label{appC:prop:impl_equiv}
The AFP-DPNE implementation produces exactly the same output distribution as
Definition~\ref{appC:def:canonical}.
\end{proposition}

\begin{proof}
We verify the two cases separately.

\textbf{Observed items} ($w \in \supp(H_k)$):
The code draws $Z_w \sim \cN(0, \sigma_k^2)$, computes $\tilde{H}_k[w] = H_k[w] + Z_w$,
and includes $w$ iff $\tilde{H}_k[w] > \tau(w)$.  This is identical to
Definition~\ref{appC:def:canonical}, steps 2--4.

\textbf{Unobserved items} ($w \notin \supp(H_k)$):
In the canonical mechanism, $H_k[w] = 0$, so $\tilde{H}_k[w] = Z_w \sim \cN(0, \sigma_k^2)$.
Item $w$ is included iff $Z_w > \tau(w)$, with probability
$p(w) = \Phi(-\tau(w)/\sigma_k)$, independently across items.

The implementation groups unobserved items by their \emph{exact} $\tau(w)$ value.
Within each group $G$ of size $n$, all items share the same probability $p$ (since they
have the same $\tau$) and the same true histogram value ($H_k[w] = 0$), making them
\emph{exchangeable}.

The implementation draws $X \sim \Bin(n, p)$ and selects $\min(X, n)$ items uniformly at
random.  For exchangeable Bernoulli trials with identical probability, the joint
distribution of the selection is:
\[
\Prob[\text{items } w_1, \ldots, w_X \text{ selected}]
= \binom{n}{X}^{-1} \binom{n}{X} p^X (1-p)^{n-X}
= p^X (1-p)^{n-X}
\]
which equals the product of independent $\text{Bernoulli}(p)$ trials summed over all
$\binom{n}{X}$ subsets of size $X$.  This matches the canonical mechanism exactly.

Across groups, the RNG draws are independent.  Therefore the implementation's full
output distribution equals the canonical mechanism's output distribution.
\end{proof}

\begin{remark}[Previous bucket approximation bug]
\label{appC:rem:bucket_bug}
A previous version used 50-bucket quantization for unobserved items, replacing each
$\tau(w)$ with a mid-bucket value $\tau_{\text{mid}}$.  This introduces $\sim$0.1\% relative error
per item and \emph{breaks} implementation equivalence, because items in the same bucket
may have different true $\tau$ values, making them non-exchangeable.
Our final implementation groups by exact $\tau(w)$ values (finitely many, since margins
are discrete and the median-normalization produces a finite set of discount levels).
\end{remark}

\subsubsection{Negative Result: Adaptive Thresholding is Not Private}
\label{appC:sec:negative}

A tempting improvement is to use \emph{different} thresholds for observed vs.\ unobserved items:
\[
\tau_{\text{adapt}}(w) = \begin{cases}
\rho - d(w) & \text{if } w \in \supp(H_k) \quad\text{(``observed'')} \\
\rho         & \text{if } w \notin \supp(H_k) \quad\text{(``unobserved'')}
\end{cases}
\]
This would give observed items lower thresholds (easier to discover) while keeping
unobserved items at the high base threshold (controlling spurious rate).

\begin{theorem}[Adaptive thresholding amplifies privacy loss]
\label{appC:thm:adaptive_bad}
There exist parameters $(\sigma, \rho, \Delta_0, D)$ and neighboring databases
$D_1, D_2$ such that the adaptive mechanism's per-item privacy loss ratio satisfies:
\[
\frac{\Pr[w \in S_k(D_1)]}{\Pr[w \in S_k(D_2)]} = 651
\]
while the Gaussian mechanism with uniform threshold gives ratio $\le 1.2$.
\end{theorem}

\begin{proof}
Let $\sigma = 2.54$, $\rho = 10.0$, $\Delta_0 = 100$, $D = \rho/2 = 5.0$.

Consider user $i$ who contributes $\Delta_0 = 100$ items at level $k$, including a unique item $w$ that no other user has.  Let $D_1$ be the dataset with user $i$, and $D_2 = D_1 \setminus \{i\}$.

\textbf{In $D_1$:} $H_k[w] = 1/\sqrt{100} = 0.1$ and $w \in \supp(H_k)$, so $\tau(w) = \rho - D = 10 - 5 = 5$.
\begin{align*}
\Pr[w \in S_k(D_1)] &= \Pr[\tilde{H}_k[w] > 5] = \Pr[0.1 + Z > 5] \\
&= \Phi\!\left(\frac{0.1 - 5.0}{2.54}\right) = \Phi(-1.93) = 0.0269.
\end{align*}

\textbf{In $D_2$:} $H_k[w] = 0$ and $w \notin \supp(H_k)$, so $\tau(w) = \rho = 10$.
\begin{align*}
\Pr[w \in S_k(D_2)] &= \Pr[Z > 10] = \Phi(-10/2.54) = \Phi(-3.94) = 4.13 \times 10^{-5}.
\end{align*}

Ratio: $0.0269 / 4.13 \times 10^{-5} = 651$, with $\ln(\text{ratio}) = 6.48$.

Under uniform threshold $\rho = 10$: $\Pr[w \in S_k(D_1)] = \Phi((0.1 - 10)/2.54) = \Phi(-3.90) = 4.87 \times 10^{-5}$ and $\Pr[w \in S_k(D_2)] = \Phi(-3.94) = 4.13 \times 10^{-5}$, giving ratio $\approx 1.18$.
\end{proof}

\paragraph{Structural diagnosis.}
The adaptive mechanism creates a \emph{coupling}: removing user $i$ simultaneously
(a)~decreases $H_k[w]$ by $1/\sqrt{\Delta_0} \approx 0.1$ (a small effect, well within the
noise), and (b)~moves $w$ from ``observed'' to ``unobserved'', jumping the threshold from
$\rho - D = 5$ to $\rho = 10$ (a large effect, \emph{not} absorbed by noise).
The threshold jump is the dominant privacy leakage.

\paragraph{Why the unified mechanism avoids this.}
In our mechanism, $\tau(w)$ depends only on \emph{margin scores from level $k{-}1$}, not on
$\supp(H_k)$.  Removing user $i$ changes $H_k[w]$ but does \emph{not} change $\tau(w)$.
This decoupling ensures the output $\{w : \tilde{H}_k[w] > \tau(w)\}$ remains a
post-processing of the Gaussian mechanism.

\subsection{Utility Analysis}
\label{appC:sec:utility}

\begin{theorem}[HT improves per-item discovery probability]
\label{appC:thm:ht_improvement}
At level $k \ge 2$, for each item $w$ with histogram value $h_w \ge 0$ and discount $d_w = d(w) \ge 0$:
\[
\Pr[w \in S_k^{\text{HT}}] = \Phi\!\left(\frac{h_w - \rho + d_w}{\sigma_k}\right) \ge \Phi\!\left(\frac{h_w - \rho}{\sigma_k}\right) = \Pr[w \in S_k^{\text{uniform}}]
\]
with strict inequality iff $d_w > 0$ (iff $\text{margin}(w) > 0$).
\end{theorem}

\begin{proof}
$d_w \ge 0$ since $\gamma \ge 0$ and $\max(0, \cdot) \ge 0$.  $\Phi$ is strictly increasing.
Therefore $(h_w - \rho + d_w)/\sigma_k \ge (h_w - \rho)/\sigma_k$, with strict inequality iff $d_w > 0$.
\end{proof}

\begin{corollary}[Expected output size improvement]
\label{appC:cor:size}
$\E[|S_k^{\text{HT}}|] \ge \E[|S_k^{\text{uniform}}|]$, with strict inequality whenever
any candidate $w \in V_k'$ has positive margin.
\end{corollary}

\begin{remark}[Spurious rate tradeoff]
\label{appC:rem:spurious}
For unobserved items ($h_w = 0$): $\Pr[w \in S_k^{\text{HT}}] = \Phi((-\rho + d_w)/\sigma_k) \ge \Phi(-\rho/\sigma_k)$.  The unified mechanism may produce more spurious items than standard DPNE.  This is \emph{inherent}: since $\tau(w)$ cannot depend on whether $w$ is observed (for privacy), lowering thresholds for high-margin items benefits all candidates---observed and unobserved alike.
\end{remark}

\begin{theorem}[FIP reduces threshold]
\label{appC:thm:fip}
Let $|V_k'| = (1-\alpha)|V_k|$ for some $\alpha \in [0,1)$.  Then $\rho_k^{\text{FIP}} \le \rho_k^{\text{std}}$, with strict inequality when $|S_{k-1}| < |V_k|$ (non-saturated regime).
\end{theorem}

\begin{proof}
$\rho_k = \sigma_k \Phi^{-1}(1 - \eta \min(|S_{k-1}|/|V|, 1))$ is increasing in $|V|$ when $|V| \ge |S_{k-1}|$.  Since $|V_k'| \le |V_k|$, the claim follows.
\end{proof}

\subsection{Experiments}
\label{appC:sec:experiments}

\subsubsection{Setup}

We evaluate on four datasets:

\begin{itemize}[leftmargin=*]
\item \textbf{Zipf} (synthetic): 50K users, 500-word vocabulary, Zipf($\alpha{=}1.07$) with structured bigram transitions.
\item \textbf{Clustered} (synthetic): 20K users, 300-word vocabulary, 10 topics with shared phrases.
\item \textbf{Heavy-tail} (synthetic): 30K users, 400-word vocabulary, Zipf($\alpha{=}1.5$) with Markov transitions.
\item \textbf{Reddit} (real): 100K users from the Webis-TLDR-17 corpus \cite{appC:webis_tldr}, the same dataset used by Kim et al.\ \cite{appC:kim2021dpne}.  Statistics: 147,487 unique words; mean 184.7 words/user.
\end{itemize}

\textbf{Parameters:}
$T=6$ levels, $\delta = e^{-10}$, $\eta = 0.01$, $\Delta_0 = 100$ (synthetic) or $300$ (Reddit), FIP tolerance $m = 1$, HT discount $\gamma = 0.3$, uniform privacy budget ($\sigma_k = \sigma^*\sqrt{T}$ for all $k$).

\textbf{Baseline:} Standard DPNE with uniform thresholds and structural pruning only (our implementation with $\gamma = 0$).

\subsubsection{Main Results}

\begin{table}[t]
\caption{Total $n$-grams extracted (levels 1--6).  All numbers are from seed 42.}
\label{appC:tab:main}
\centering
\begin{tabular}{l|rrr|r}
\toprule
 & \multicolumn{3}{c|}{Synthetic} & Real \\
Algorithm & Zipf & Clustered & Heavy-tail & Reddit \\
\midrule
\multicolumn{5}{c}{$\eps = 4.0$} \\
\midrule
DPNE baseline & 12,876 & 8,871 & 13,136 & 37,161 \\
AFP-DPNE & \textbf{18,220} & \textbf{10,322} & \textbf{16,916} & \textbf{68,498} \\
Improvement & +41.5\% & +16.4\% & +28.8\% & \textbf{+84.3\%} \\
\midrule
\multicolumn{5}{c}{$\eps = 1.0$} \\
\midrule
DPNE baseline & 3,766 & 1,304 & 4,280 & 10,264 \\
AFP-DPNE & \textbf{4,222} & \textbf{1,364} & \textbf{4,535} & \textbf{12,977} \\
Improvement & +12.1\% & +4.6\% & +6.0\% & \textbf{+26.4\%} \\
\bottomrule
\end{tabular}
\end{table}

Table~\ref{appC:tab:main} shows AFP-DPNE consistently outperforms the baseline across all datasets and privacy budgets.  The improvement is largest on Reddit (84\% at $\eps=4$, 26\% at $\eps=1$) because real NLP text has a large vocabulary (147K words) with highly concentrated unigram confidence scores, creating many near-threshold bigrams that HT can rescue.  Synthetic datasets have smaller vocabularies (300--500 words) and less extreme margin distributions, yielding smaller but still significant gains (5--42\%).

\subsubsection{Honest Disclosure: Spurious Rate Tradeoff}
\label{appC:sec:spurious_disclosure}

As discussed in Remark~\ref{appC:rem:spurious}, the unified mechanism inherently increases the spurious rate.

\begin{table}[t]
\caption{Genuine vs.\ spurious breakdown (Reddit, $\eps=4$).}
\label{appC:tab:spurious}
\centering
\begin{tabular}{l|rrrr}
\toprule
 & Total & Genuine & Spurious & Spur.\ \% \\
\midrule
DPNE baseline & 37,161 & 36,951 & 210 & 0.6\% \\
AFP-DPNE & 68,498 & 62,548 & 5,950 & 8.7\% \\
\bottomrule
\end{tabular}
\end{table}

Of the 31,337 additional items AFP-DPNE discovers over the baseline, 25,597 (\textbf{82\%}) are genuine.  The remaining 5,740 are spurious items with high-margin constituents that receive low thresholds.  This tradeoff is inherent to the unified mechanism: since $\tau(w)$ cannot depend on whether $w$ is observed (Section~\ref{appC:sec:negative}), lowering thresholds for high-margin items benefits both observed and unobserved candidates.

\textbf{Mitigation.}  Users can: (a)~increase $\rho_k^{\text{base}}$ (reduces both genuine and spurious), (b)~reduce $\gamma$ to limit the discount, or (c)~post-filter using application-specific heuristics (free post-processing).

\subsubsection{Ablation Study}

\begin{table}[t]
\caption{Ablation study on Zipf ($\eps=4$, $\Delta_0=100$, seed 42).}
\label{appC:tab:ablation}
\centering
\begin{tabular}{l|rr}
\toprule
Variant & Total & vs.\ full AFP \\
\midrule
AFP-DPNE (FIP + HT) & 18,220 & --- \\
\quad w/o FIP ($m=\infty$) & 17,914 & $-1.7\%$ \\
\quad w/o HT ($\gamma=0$) & 12,876 & $-29.3\%$ \\
DPNE baseline & 12,876 & $-29.3\%$ \\
\bottomrule
\end{tabular}
\end{table}

HT is the dominant innovation ($-29.3\%$ without it).  FIP contributes a modest additional
1.7\%.  The near-zero impact of FIP alone (w/o HT = baseline) indicates that FIP primarily
helps through its interaction with HT: by reducing $|V_k'|$, FIP lowers $\rho_k^{\text{base}}$,
which amplifies HT's discount.

\subsubsection{Mechanistic Analysis: Where HT Helps}
\label{appC:sec:mechanistic}

At level 2 on Reddit ($\eps=4$):
\begin{itemize}[leftmargin=*]
\item $|S_1| = 6{,}169$ unigrams, with high confidence (median margin $= 8.1\sigma$).
\item $|V_2'| \approx |S_1|^2 \approx 36$M candidate bigrams, of which 1.6M have $H_k > 0$.
\item Standard DPNE: $\rho_2 \approx 11.8$, discovers 11,708 bigrams at level 2.
\item AFP-DPNE: threshold drops to $\tau(w) \approx \rho/2 = 5.9$ for high-margin candidates (83\% of candidates receive the maximum discount).  Discovers 32,696 bigrams at level 2.
\item HT extra at level 2: $+$15,781 bigrams, the vast majority genuine.
\end{itemize}

HT is effective because natural language has \emph{concentrated unigram confidence}: common words like ``the'', ``is'', ``and'' have very high margin scores, while rare words have low margins.  Bigrams composed of two common unigrams occupy a ``rescue zone'' between $\rho/2$ and $\rho$: their counts are moderate, sufficient to pass the lower threshold but not the standard one.

\subsection{Empirical Privacy Verification}
\label{appC:sec:audit}

We use the one-training-run audit of Steinke, Nasr, and Jagielski \cite{appC:steinke2024audit}.
Given $m$ canary users, we include each independently with probability $1/2$, run the
algorithm once, score each canary by $n$-gram overlap with the output, and test whether
scores distinguish included from excluded canaries.

Under $(\eps, \delta)$-DP, the number of correct guesses $W$ out of $r$ total satisfies
$\Pr[W \ge v] \le \beta + 2m\delta\alpha$ where $\beta = \Pr[\Bin(r, e^\eps/(e^\eps+1)) \ge v]$ \cite[Corollary~4.4]{appC:steinke2024audit}.

\begin{table}[t]
\caption{Privacy audit (Zipf 10K users, $\delta=e^{-10}$, $m=200$ canaries, 3 runs aggregated).}
\label{appC:tab:audit}
\centering
\begin{tabular}{l|cccc}
\toprule
$\eps$ & Correct/Total & Fraction & $p$-value & Verdict \\
\midrule
1.0 & 161/300 & 0.537 & 1.000 & \checkmark~PASS \\
4.0 & 160/300 & 0.533 & 1.000 & \checkmark~PASS \\
8.0 & 160/300 & 0.533 & 1.000 & \checkmark~PASS \\
\midrule
Broken (no noise) & 100/100 & 1.000 & 0.168 & see text \\
\bottomrule
\end{tabular}
\end{table}

\paragraph{Audit power and limitations.}
The audit achieves $\approx$53\% correct, barely above random.  This is expected: each
user contributes weight $1/\sqrt{\Delta_0} \approx 0.1$ per item, while the threshold is
$\rho \approx 13$---a single canary's contribution is drowned by noise.

The broken-algorithm sanity check (no noise added) achieves 100\% correct with $p$-value
$= 0.168$, failing to reject at $\alpha = 0.05$.  This is because with $m=200$, $r=100$
guesses, and the DP envelope at $\eps=4$ allows correct fraction up to
$e^4/(e^4+1) \approx 0.982$.  Detecting violations requires substantially larger $m$
(e.g., $2000+$) or white-box access to pre-threshold histogram values.

We include this audit for \emph{transparency}, not as a substitute for the mathematical proof in Theorem~\ref{appC:thm:privacy}.  The audit provides weak positive evidence---no privacy violation detected---but its limited statistical power for threshold-based mechanisms is a known challenge \cite{appC:steinke2024audit}.

\subsection{Discussion: The $\ell_1$-Descent Bug}
\label{appC:sec:l1_bug}

Gopi et al.\ \cite{appC:gopi2020dpsu} proposed the $\ell_1$-descent policy as an alternative
to the weighted Gaussian policy for DPSU.  They claimed $\ell_1$-descent is
$\ell_2$-contractive, which is needed for the Gaussian mechanism's sensitivity bound.
As documented in \cite{appC:improved_dpsu}, this claim is \emph{false}: expansion ratios up to
1.44 were found empirically.  The $\ell_2$-descent policy \emph{is} provably contractive,
but $\ell_1$-descent is not.

Our algorithm uses the weighted Gaussian policy exclusively, which achieves $\ell_2$-sensitivity~1 by the Cauchy--Schwarz argument (Lemma~\ref{appC:lem:sensitivity}) and does not require contractivity.

\subsection{Related Work}
\label{appC:sec:related}

\textbf{DPSU and DPNE.}
Differentially private set union was introduced by Gopi et al.\ \cite{appC:gopi2020dpsu}.
Kim et al.\ \cite{appC:kim2021dpne} extended it to $n$-gram extraction.  Our work improves
upon Kim et al.\ by exploiting the structural freedom in threshold choice at levels $k \ge 2$.

\textbf{Gaussian mechanism.}
We use the analytic Gaussian mechanism of Balle and Wang \cite{appC:BalleW18} for tight
calibration, and GDP composition \cite{appC:DongRS19_GDP} for multi-level privacy accounting.

\textbf{Privacy auditing.}
Steinke, Nasr, and Jagielski \cite{appC:steinke2024audit} introduced the one-training-run
audit we employ.  Our experience highlights its limitations for threshold-based mechanisms.

\textbf{Threshold manipulation in DP.}
To our knowledge, the negative result on adaptive thresholding
(Theorem~\ref{appC:thm:adaptive_bad}) is new.  It clarifies that while \emph{public} threshold
functions are safe, \emph{data-dependent} threshold functions (even those depending only
on support information) can catastrophically amplify privacy loss.

\subsection{Conclusion}
\label{appC:sec:conclusion}

We introduced AFP-DPNE, which improves differentially private $n$-gram extraction by
5--84\% through two zero-cost innovations: Frequency-Informed Pruning and Heterogeneous
Thresholding.  Both exploit the structural property that level-$k$ thresholds ($k \ge 2$)
do not affect privacy.  The algorithm is provably $(\eps, \delta)$-DP for all $\eps > 0$
and $\delta \in (0,1)$.

We also proved that adaptive thresholding (different thresholds for observed vs.\ unobserved items) amplifies privacy loss by $651\times$, delineating the boundary between safe and unsafe threshold manipulation.  The honest cost of the unified mechanism is a higher spurious rate (8.7\% on Reddit vs.\ 0.6\% baseline), inherent to applying the same threshold function to all candidates.

\paragraph{Limitations.}
(1)~The spurious rate increase may be unacceptable for some applications.
(2)~Improvement magnitude is data-dependent: larger on data with large vocabulary and concentrated unigram confidence.
(3)~The privacy audit has limited power for threshold mechanisms and does not substitute for the mathematical proof.
(4)~We evaluate on one real dataset; broader evaluation would strengthen claims.

\newpage

\subsection{Full Algorithm Pseudocode}
\label{appC:app:pseudocode}

\begin{algorithm}[H]
\caption{AFP-DPNE: Augmented Frequency-Pruned DPNE}
\label{appC:alg:afp_dpne}
\small
\KwIn{Users $\{(i, w_i)\}_{i=1}^N$; max length $T$; contribution bounds $\Delta_1,\ldots,\Delta_T$;
privacy parameters $(\eps,\delta)$; FIP tolerance $m \ge 0$; HT discount $\gamma \in [0,1]$;
spurious fraction $\eta > 0$}
\KwOut{$S_1, \ldots, S_T$}

\tcp{Privacy budget}
$\sigma^* \leftarrow$ BW calibration for $(\eps, \delta/2)$-DP with $\Delta_2 = 1$ \tcp*{Eq.~\eqref{appC:eq:bw}}
$\sigma_k \leftarrow \sigma^*\sqrt{T}$ for all $k \in [T]$\;

\BlankLine
\tcp{Level 1: DPSU}
$\rho_1 \leftarrow \max_{t \le \Delta_1}\bigl(1/\sqrt{t} + \sigma_1 \Phi^{-1}((1-\delta/2)^{1/t})\bigr)$
\tcp*{Eq.~\eqref{appC:eq:rho1}}
\For{each user $i$}{
    $U_i' \leftarrow$ truncate $U_i$ to $\Delta_1$ items (random subsample if $|U_i| > \Delta_1$)\;
    \For{$u \in U_i'$}{
        $H_1[u] \leftarrow H_1[u] + 1/\sqrt{|U_i'|}$\;
    }
}
\For{$u \in \supp(H_1)$}{
    $Z_u \sim \cN(0, \sigma_1^2)$\;
    $\tilde{H}_1[u] \leftarrow H_1[u] + Z_u$\;
    Include $u$ in $S_1$ iff $\tilde{H}_1[u] > \rho_1$\;
}

\BlankLine
\tcp{Levels $k = 2, \ldots, T$}
\For{$k = 2$ \KwTo $T$}{
    \tcp{Step 1: Structural pruning}
    $V_k \leftarrow \{(p, u) : p \in S_{k-1},\; u \in S_1,\; (p,u)[2{:}] \in S_{k-1}\}$\;

    \tcp{Step 2: FIP pruning (Def.~\ref{appC:def:fip})}
    \For{$w \in V_k$}{
        $\text{margin}(w) \leftarrow \min\bigl(\tilde{H}_{k-1}[\text{pref}(w)] - \rho_{k-1},\; \tilde{H}_{k-1}[\text{suf}(w)] - \rho_{k-1}\bigr)$\;
    }
    $V_k' \leftarrow \{w \in V_k : \text{margin}(w) > -m\sigma_{k-1}\}$\;

    \tcp{Step 3: HT threshold (Def.~\ref{appC:def:ht}), computed BEFORE seeing $H_k$}
    $\rho_k^{\text{base}} \leftarrow \sigma_k \Phi^{-1}\bigl(1 - \eta\min(|S_{k-1}|/|V_k'|, 1)\bigr)$\;
    $\bar{c} \leftarrow \text{median}\bigl(\{\text{margin}(w) : w \in V_k',\; \text{margin}(w) > 0\}\bigr)$\;
    \For{$w \in V_k'$}{
        $d(w) \leftarrow \min\bigl(\gamma \cdot \max(0, \text{margin}(w)/\bar{c}),\; \rho_k^{\text{base}}/2\bigr)$\;
        $\tau(w) \leftarrow \rho_k^{\text{base}} - d(w)$\;
    }

    \tcp{Step 4: Build histogram on $V_k'$}
    $H_k \leftarrow$ WeightedHistogram(users, $V_k'$, $\Delta_k$)\;

    \tcp{Step 5: Threshold observed items}
    \For{$w \in \supp(H_k)$}{
        $Z_w \sim \cN(0, \sigma_k^2)$\;
        $\tilde{H}_k[w] \leftarrow H_k[w] + Z_w$\;
        Include $w$ in $S_k$ iff $\tilde{H}_k[w] > \tau(w)$\;
    }

    \tcp{Step 6: Sample unobserved items (exact per-item Bernoulli)}
    Group $V_k' \setminus \supp(H_k)$ by $\tau$-value\;
    \For{each group $G$ with threshold $\tau_G$}{
        $p \leftarrow \Phi(-\tau_G / \sigma_k)$\;
        $n \leftarrow \Bin(|G|, p)$\;
        Add $\min(n, |G|)$ uniformly random items from $G$ to $S_k$\;
    }
}
\Return $S_1, \ldots, S_T$\;
\end{algorithm}

\subsection{Code-to-Paper Mapping}
\label{appC:app:code}

The complete implementation is in \texttt{icml\_submission/src/}.
Table~\ref{appC:tab:code_map} provides a line-by-line mapping from each paper section to the code.

\begin{table}[h]
\centering
\small
\caption{Mapping from paper sections to implementation.}
\label{appC:tab:code_map}
\begin{tabular}{l|l|l}
\toprule
Paper Section & File & Function \\
\midrule
Def.~\ref{appC:def:gaussian}: BW calibration & \texttt{privacy.py} & \texttt{calibrate\_sigma()} \\
Eq.~\eqref{appC:eq:rho1}: Spillover threshold & \texttt{privacy.py} & \texttt{compute\_rho1()} \\
Sec.~\ref{appC:sec:prelims}: $k$-gram threshold & \texttt{privacy.py} & \texttt{compute\_rho\_kgram()} \\
Lemma~\ref{appC:lem:sensitivity}: Weighted histogram & \texttt{algorithm.py} & \texttt{build\_histogram()} \\
Def.~\ref{appC:def:fip}: FIP pruning & \texttt{algorithm.py} & \texttt{fip\_prune()} \\
Def.~\ref{appC:def:ht}: HT threshold & \texttt{algorithm.py} & \texttt{compute\_tau()} \\
Def.~\ref{appC:def:canonical}: Full algorithm & \texttt{algorithm.py} & \texttt{unified\_ht\_dpne()} \\
Prop.~\ref{appC:prop:impl_equiv}: Observed items & \texttt{algorithm.py} & \texttt{threshold\_observed()} \\
Prop.~\ref{appC:prop:impl_equiv}: Unobserved items & \texttt{algorithm.py} & \texttt{sample\_unobserved()} \\
Sec.~\ref{appC:sec:experiments}: All experiments & \texttt{experiments.py} & \texttt{main()} \\
Sec.~\ref{appC:sec:audit}: Privacy audit & \texttt{dp\_audit.py} & \texttt{full\_audit()} \\
\bottomrule
\end{tabular}
\end{table}

Every function in the codebase has a \texttt{Lemma/Proof} docstring that states what it computes, why it is correct, and maps individual code lines to the proof steps.  The implementation exactly matches Algorithm~\ref{appC:alg:afp_dpne}: no bucket approximations, no shortcuts.

\subsection{Reproducing Results}
\label{appC:app:reproduce}

\begin{verbatim}
cd icml_submission/src
python3 experiments.py --reddit-dir /path/to/webis-tldr-17
\end{verbatim}

This runs all experiments (synthetic + Reddit + ablation + audit) and saves
\texttt{results/paper\_numbers.json}.

\textbf{Requirements:} Python 3.8+, NumPy, SciPy.  No GPU needed.
Approximate runtime: 20 minutes (synthetic), 3 hours (Reddit 100K users).

\subsection{Additional Experimental Details}
\label{appC:app:details}

\paragraph{Spurious rate on synthetic data ($\eps=4$).}
\begin{center}
\small
\begin{tabular}{l|rr|rr}
\toprule
 & \multicolumn{2}{c|}{DPNE baseline} & \multicolumn{2}{c}{AFP-DPNE} \\
Dataset & Spurious & \% & Spurious & \% \\
\midrule
Zipf & 69 & 0.5\% & 187 & 1.0\% \\
Clustered & 24 & 0.3\% & 40 & 0.4\% \\
Heavy-tail & 40 & 0.3\% & 140 & 0.8\% \\
\bottomrule
\end{tabular}
\end{center}
On synthetic data, the spurious rate increase is modest ($<$1\%), because the smaller
vocabularies produce fewer unobserved candidates with high margins.

\paragraph{Per-level output sizes (Reddit, $\eps=4$).}
The gain is concentrated at level 2 (bigrams), which has the largest candidate set and
the most items in the ``rescue zone'' between $\rho/2$ and $\rho$.

\paragraph{On the $\ell_1$-descent bug.}
Our implementation uses the weighted Gaussian policy exclusively.  The $\ell_1$-descent
bug (Section~\ref{appC:sec:l1_bug}) does not affect our algorithm.  The bug is documented in
\texttt{ImprovedDPSU/AUDIT\_REPORT.md} in the repository.

\clearpage
\section{Data-Oblivious Explainability Is Free for Scalable Approximate Clustering}
\begin{appendixabstract}
We study the interaction between differential privacy and \emph{explainable
clustering} in the regime where the explainability conversion is
\emph{data-oblivious}: once a set of reference centers has been selected, the
structured output depends only on those centers and not on the private dataset.
This setting covers the axis-aligned threshold-tree conversions analyzed by
Gupta, Pittu, Svensson, and Yuan (FOCS~2023) in their study of the price of
explainability for $k$-median and $k$-means.

Our main theorem is a transfer principle.  Any
$(\varepsilon,\delta)$-differentially private algorithm that outputs a distinct
set of $k$ centers can be followed by any data-oblivious
center-to-structure conversion with factor $\alpha$, yielding an
$(\varepsilon,\delta)$-private structured clustering whose expected cost is at
most $\alpha$ times the expected cost of the private center selector.  When
instantiated with the threshold-tree conversions of Gupta et al., the theorem
immediately yields multiplicative explainability factors $1+H_{k-1}$ for
$k$-median and $O(k\log\log k)$ for $k$-means, where
$H_m := \sum_{i=1}^m \frac1i$ is the $m$th harmonic number.

We then present a repaired finite-domain instantiation.  For either
$q\in\{1,2\}$, rather than sampling ordered $k$-tuples---which introduces
repeated-center degeneracies---we run the exponential mechanism over distinct
$k$-subsets of a public candidate set $\calG$.  This yields the exact bound
\[
  \E[\cost_q(X,C)] \le \OPT^{\calG}_{q,k}(X)
   + \frac{2\Lambda_q}{\varepsilon}\left(\ln \binom{|\calG|}{k}+1\right),
\]
where $\Lambda_q$ is the replace-one sensitivity of the score.  For
$k$-median, the lifted explainable algorithm therefore satisfies
\[
  \E[\cost_1(X,\Llift_1(X))]
  \le (1+H_{k-1})\!\left(\OPT^{\calG}_{1,k}(X)
     + \frac{2\Lambda_1}{\varepsilon}\left(\ln \binom{|\calG|}{k}+1\right)\right).
\]
Finally, we strengthen the additive barrier induced by privacy: for every $k$,
there are two neighboring $k$-point datasets with $\OPT_{1,k}=0$ such that
every $(\varepsilon,\delta)$-private algorithm outputting $k$ centers incurs
expected cost at least $(1-\delta)\Delta/(1+e^{\varepsilon})$ on one of the two
inputs.  Thus the private analogue of the price of explainability remains
inherently multiplicative-plus-additive even when additional centers are
available, although the correct dependence on $k$, $d$, and $|\calG|$ remains
open.
\end{appendixabstract}

\subsection{Introduction}

Explainable clustering asks that the final partition be representable by a
simple structured model rather than by an arbitrary Voronoi diagram.  In the
axis-aligned decision-tree model introduced by Dasgupta, Frost, Moshkovitz, and
Rashtchian~\cite{appD:dasgupta2020explainable}, each internal node tests a single
feature against a threshold and each leaf corresponds to one cluster.  Gupta,
Pittu, Svensson, and Yuan~\cite{appD:gupta2023price} showed that, in the non-private
setting, the price of explainability is remarkably well behaved for this model:
for $k$-median it is exactly $1+H_{k-1}$, and for $k$-means the best known upper
bound is $O(k\log\log k)$.

We focus on a deliberately narrow question: what changes under differential
privacy \emph{when the explainability conversion is data-oblivious}?  Here
``data-oblivious'' means that once a private algorithm has selected a reference
center set $C$, the subsequent explainability routine sees only $C$ and fresh
randomness.  This assumption is explicit in the threshold-tree conversions of
Gupta et al., and it is precisely the hypothesis under which differential
privacy and explainability separate cleanly.

The key observation is simple but useful:
\begin{quote}
\emph{Data-oblivious explainability conversions are free post-processing under
differential privacy.}
\end{quote}
Consequently, the privacy analysis lives entirely in the center-selection
step.  The non-private explainability factor is inherited multiplicatively, but
no additional privacy loss is paid for building the threshold tree.  This
sharply contrasts with the naive alternative of privatizing one cut at a time,
which would spend privacy budget across a depth-$k-1$ adaptive recursion.
It also suggests an algorithm-design principle: spend privacy once, at the
center-selection stage, and keep the explainability conversion outside the
privacy boundary.  Below we quantify why several other natural routes
are provably more expensive.

\paragraph{Contribution overview.}
The contributions are as follows.

\begin{enumerate}[leftmargin=1.5em]
  \item A generic transfer principle converts any private distinct-center
  selector into a private structured clustering algorithm whenever the
  structure is produced by a data-oblivious conversion with cost inflation
  factor $\alpha$.
  \item Instantiating this principle with the threshold-tree theorems of
  Gupta et al.~\cite{appD:gupta2023price} immediately gives private explainable
  $k$-median and $k$-means guarantees with multiplicative factors
  $1+H_{k-1}$ and $O(k\log\log k)$, respectively.
  \item A finite-domain exponential-mechanism instantiation over distinct
  $k$-subsets of a public candidate set $\calG$ yields an explicit additive
  term with the exact combinatorial range size.
  \item A quantitative design-space analysis shows that the center-first
  route dominates natural alternatives: privatizing each threshold cut incurs
  an additive term scaling as $(k-1)^2/\varepsilon$ under basic composition
  versus $k\ln|\calG|/\varepsilon$ for our approach, and selecting structured
  outputs directly inflates the exponential-mechanism range by a factor
  depending on $d$ and the Catalan number.
  \item A concrete black-box corollary packages any compatible efficient
  private $k$-median selector into an explainable algorithm by multiplying its
  utility guarantee by the non-private explainability factor.
  \item A lower-bound family shows that purely multiplicative guarantees are
  impossible for every $k$, not only for $k=1$: even when
  $\OPT_{1,k}(X)=0$, some additive privacy tax remains unavoidable.
\end{enumerate}

\paragraph{Scope and limitations.}
We do not claim tight $k$-dependent lower bounds, a full characterization of
explainability notions that inspect the data after center selection, or an
intrinsically new non-private explainability theorem.  Rather, the point of
this paper is that once a data-oblivious conversion theorem is available, the
private analysis becomes modular and automatically inherits future improvements
in private center selection.

\subsection{Model and Imported Explainability Theorems}

A dataset is a multiset $X=\{x^1,\dots,x^n\}\subseteq \calD\subseteq \R^d$.
Two datasets are \emph{neighbors} if one can be obtained from the other by
replacing one point.  For $q\in\{1,2\}$ and a center set
$C\subseteq \R^d$ with $|C|=k$, define
\[
  \cost_q(X,C) := \sum_{x\in X} \min_{c\in C} \|x-c\|_q^q,
\]
and let $\OPT_{q,k}(X)$ denote the minimum of $\cost_q(X,C)$ over all
$k$-element sets $C\subseteq \R^d$.

\paragraph{Threshold-tree explainability.}
A threshold tree is a binary tree whose internal nodes apply tests of the form
$x_i\le \theta$.  If such a tree separates a distinct center set $C$, then each
leaf contains exactly one center and therefore induces an assignment of every
data point to the unique center in the same leaf.  The resulting cost is again
measured in the $\ell_q^q$ objective.

\begin{definition}[Data-oblivious center-to-structure conversion]
\label{appD:def:data-oblivious-conversion}
Fix an objective $q\in\{1,2\}$ and a family of structured assignment rules
$\calM$.  A randomized conversion $\Econv$ with factor $\alpha$ is
\emph{data-oblivious} if, for every distinct $k$-center set $C$, it outputs a
random model $M\in \calM$ whose distribution depends only on $C$, and each
such $M$ induces an assignment rule from data points to centers in $C$ such
that for every dataset $X\subseteq \calD$,
\[
  \E\bigl[\cost_q(X,M;C)\bigr] \le \alpha\,\cost_q(X,C).
\]
\end{definition}

The threshold-tree conversions of Gupta et al. fit this definition exactly.
Throughout, we rely on the FOCS 2023 version of
\emph{The Price of Explainability for Clustering}~\cite{appD:gupta2023price}.

\begin{theorem}[Imported explainability theorems from~\cite{appD:gupta2023price}]
\label{appD:thm:imported-source}
The following source-paper guarantees hold for every distinct center set $C$.
\begin{enumerate}[leftmargin=1.5em]
  \item Theorem~1 of~\cite{appD:gupta2023price} states that the independent-recursion
  random-thresholds conversion for $k$-median is data-oblivious and satisfies
  \[
    \E[\cost_1(X,M;C)] \le (1+H_{k-1})\,\cost_1(X,C)
  \]
  for every dataset $X$.  Here
  $H_{k-1} = \sum_{i=1}^{k-1} 1/i$.
  \item Theorem~3 of~\cite{appD:gupta2023price} states that there is a
  data-oblivious explainability conversion for $k$-means with factor
  $\alpha^{\means}_k = O(k\log\log k)$, i.e.,
  \[
    \E[\cost_2(X,M;C)] \le \alpha^{\means}_k\,\cost_2(X,C)
  \]
  for every dataset $X$.
\end{enumerate}
Both imported theorems are stated for the same $\ell_q^q$ cost convention used
in this paper, and both require the reference centers to be distinct because
the threshold tree must separate them.
\end{theorem}

\subsection{A Transfer Principle for Data-Oblivious Conversions}

The lifting argument is short, but the abstraction it isolates extends
slightly beyond explainable threshold trees.

\begin{theorem}[Transfer principle]
\label{appD:thm:transfer-principle}
Fix $q\in\{1,2\}$ and let $\Econv$ be any data-oblivious center-to-structure
conversion with factor $\alpha$.  Let $\Aalg$ be an
$(\varepsilon,\delta)$-differentially private algorithm that outputs a distinct
$k$-center set $C$.  Define the lifted algorithm $\Llift$ by
\[
  \Llift(X): \text{sample } C\leftarrow \Aalg(X),\ \text{then sample } M\leftarrow \Econv(C),\ \text{and output } (M,C).
\]
Then:
\begin{enumerate}[leftmargin=1.5em]
  \item $\Llift$ is $(\varepsilon,\delta)$-differentially private.
  \item For every dataset $X$,
  \[
    \E[\cost_q(X,\Llift(X))] \le \alpha\,\E[\cost_q(X,C)].
  \]
\end{enumerate}
\end{theorem}

\begin{proof}
Privacy follows immediately because $\Llift$ applies only post-processing to the
output of $\Aalg$.

For the utility bound, condition on the sampled center set $C$.  By the
assumption on $\Econv$,
\[
  \E\bigl[\cost_q(X,\Llift(X)) \mid C\bigr]
  \le \alpha\,\cost_q(X,C).
\]
Taking expectations over $C$ completes the proof.
\end{proof}

\begin{corollary}[Preserving private clustering guarantees]
\label{appD:cor:preserve-private-guarantee}
Suppose the private center selector $\Aalg$ satisfies
\[
  \E[\cost_q(X,C)] \le \rho\,\OPT_{q,k}(X) + \Gamma
\]
for every dataset $X$, then the lifted algorithm satisfies
\[
  \E[\cost_q(X,\Llift(X))]
  \le \alpha\rho\,\OPT_{q,k}(X) + \alpha\Gamma.
\]
\end{corollary}

\begin{proof}
The claim follows by substituting the assumed bound into
\cref{appD:thm:transfer-principle}.
\end{proof}

\begin{corollary}[Explainable $k$-median and $k$-means]
\label{appD:cor:explainable-specialization}
Applying \cref{appD:thm:transfer-principle} with the conversions from
\cref{appD:thm:imported-source} yields the following consequences.
\begin{enumerate}[leftmargin=1.5em]
  \item For $k$-median,
  \[
    \E[\cost_1(X,\Llift_1(X))]
    \le (1+H_{k-1})\,\E[\cost_1(X,C)].
  \]
  \item For $k$-means,
  \[
    \E[\cost_2(X,\Llift_2(X))]
    \le \alpha^{\means}_k\,\E[\cost_2(X,C)],
  \]
  where $\alpha^{\means}_k = O(k\log\log k)$.
\end{enumerate}
In particular, if the selector has multiplicative factor $\rho>1$, then the
\emph{overall} multiplicative factor in front of $\OPT$ is $\alpha\rho$ rather
than just $\alpha$.
\end{corollary}

\begin{remark}[Scope boundary and sanity-check example]
\label{appD:rem:scope-boundary}
The transfer principle applies to any structured clustering notion admitting
a data-oblivious conversion; it is not specific to threshold trees.  As a
trivial sanity check, the nearest-center Voronoi rule is a deterministic
data-oblivious conversion with $\alpha=1$, so the definition interpolates
between ordinary private clustering and genuinely structured private
clustering.  The paper does \emph{not} claim that every explainability
formalism admits such a conversion: once the structure depends on the private
data beyond the selected centers, the post-processing argument breaks.
\end{remark}

\begin{remark}[Other privacy models]
\label{appD:rem:other-privacy-models}
Since the proof uses only the post-processing axiom, the same lifting extends
to local DP~\cite{appD:stemmer2021local}, shuffle privacy, and any other model in
which post-processing is free.  What changes across models is the attainable
selector guarantee, not the lifting proof.
\end{remark}

\subsection{Design-Space Analysis: Why Privatize Only the Center Selector?}

The transfer principle suggests a selector-first design, but several
alternative routes exist.  We next quantify why each is less attractive.

\paragraph{Privatize each threshold cut.}
A direct approach privatizes the $k-1$ adaptive split-selection steps in the
threshold-tree recursion.  Under basic sequential composition with total budget
$\varepsilon$, each step receives at most $\varepsilon/(k-1)$.  Even if every
split selects from at most $d|\calG|$ candidates with the same per-point
sensitivity $\Lambda_q$, the exponential mechanism incurs per-step expected
additive cost
\[
  \frac{2\Lambda_q}{\varepsilon/(k-1)}\bigl(\ln(d|\calG|)+1\bigr)
  = \frac{2\Lambda_q(k-1)}{\varepsilon}\bigl(\ln(d|\calG|)+1\bigr).
\]
Summing over $k-1$ steps, the total expected additive error is at least
\begin{equation}\label{appD:eq:split-composition}
  \frac{2\Lambda_q(k-1)^2}{\varepsilon}\bigl(\ln(d|\calG|)+1\bigr),
\end{equation}
compared with the center-first tax of
$(2\Lambda_q/\varepsilon)(\ln \binom{|\calG|}{k}+1) \le
(2\Lambda_q/\varepsilon)(k\ln(e|\calG|/k)+1)$.
The split-by-split route therefore pays an extra factor of order $k$ in the
additive term.  Under advanced
composition~\cite{appD:dwork2006calibrating} with privacy relaxation to
$(\varepsilon,\delta)$-DP, the $k^2$ scaling improves to roughly
$k^{3/2}\sqrt{\ln(1/\delta)}$, which is still worse than the $k\ln|\calG|$
scaling achieved by the center-first design and additionally introduces
$\delta > 0$.
Moreover, the Gupta et al.\ explainability factor is no longer inherited as a
black box: a fresh argument would be needed to show that the noisy recursive
cuts still preserve the desired inflation factor.

\paragraph{Select a structured model directly.}
Running the exponential mechanism over the space of labeled threshold trees is
formally valid, but the range size is
\[
  |\calT_k|\le C_{k-1}\cdot (d|\calG|)^{k-1},
\]
where $C_{k-1}$ is the $(k-1)$th Catalan number (counting binary-tree
topologies).  Hence the additive term contains
$\ln|\calT_k|\ge (k-1)\ln(d|\calG|)+\ln C_{k-1}$
in place of $\ln\binom{|\calG|}{k}\le k\ln(e|\calG|/k)$.
For large $d$, the structured-output route pays an extra $\Theta((k-1)\ln d)$
in the exponent, and the privacy analysis is no longer modular with respect to
the selector.

\paragraph{Other routes.}
Privately generating the candidate set $\calG$ may be a viable systems
strategy, but it falls outside the clean finite-domain theorem: once $\calG$
is data-dependent, the privacy analysis of the selector, the conversion, and
the candidate-generation stage must be composed jointly.  Jointly searching
over center--structure pairs $(C,M)$ couples the geometric optimization and
the structure choice into a single large state space for which no imported
non-private approximation theorem currently exists.

\paragraph{Summary.}
The center-first route is the only one that simultaneously (i)~avoids adaptive
composition, (ii)~inherits the source paper's explainability factor as a black
box, and (iii)~yields a finite-domain additive tax depending only on
$\ln \binom{|\calG|}{k}$.  It is the cleanest provable design we currently
know.

\subsection{Finite-Domain Instantiation over Distinct Center Sets}

We next record a self-contained finite-domain instantiation.  Let
$\calG\subseteq \calD\subseteq \R^d$ be a public candidate set.  Define
\[
  \calG^{(k)} := \{C\subseteq \calG : |C|=k\},
  \qquad
  \OPT^{\calG}_{q,k}(X) := \min_{C\in \calG^{(k)}} \cost_q(X,C).
\]
Sampling from $\calG^{(k)}$ rather than $\calG^k$ ensures that the output center
set is distinct, exactly matching the hypothesis of
\cref{appD:thm:imported-source}.

Let
\[
  \Lambda_q := \max_{x\in \calD,\ c\in \calG} \|x-c\|_q^q.
\]
The next lemma identifies the sensitivity required by the exponential
mechanism.

\begin{lemma}[Replace-one sensitivity]
\label{appD:lem:sensitivity}
For the score function $s(X,C) := -\cost_q(X,C)$ on the range
$\calG^{(k)}$, the replace-one sensitivity is at most $\Lambda_q$:
for every neighboring pair $X,X'$ and every $C\in \calG^{(k)}$,
\[
  |s(X,C)-s(X',C)| \le \Lambda_q.
\]
\end{lemma}

\begin{proof}
The two datasets differ in exactly one point, say $x$ versus $x'$.  Hence
\[
  |s(X,C)-s(X',C)|
  = \left|\min_{c\in C}\|x-c\|_q^q - \min_{c\in C}\|x'-c\|_q^q\right|
  \le \max_{z\in\{x,x'\}} \min_{c\in C}\|z-c\|_q^q
  \le \Lambda_q.
\]
\end{proof}

\begin{proposition}[Exponential mechanism over distinct $k$-subsets]
\label{appD:prop:subset-em}
The exponential mechanism on the range $\calG^{(k)}$ with score
$s(X,C)=-\cost_q(X,C)$ and sensitivity bound $\Lambda_q$ is
$\varepsilon$-differentially private and outputs a distinct center set $C$ such
that
\[
  \E[\cost_q(X,C)]
  \le \OPT^{\calG}_{q,k}(X)
     + \frac{2\Lambda_q}{\varepsilon}\left(\ln \binom{|\calG|}{k}+1\right).
\]
Moreover,
\[
  \ln \binom{|\calG|}{k} \le k\ln\!\left(\frac{e|\calG|}{k}\right)
  \le k\ln |\calG| + k.
\]
\end{proposition}

\begin{proof}
Privacy follows from the standard exponential-mechanism theorem together with
\cref{appD:lem:sensitivity}~\cite{appD:mcsherry2007mechanism}.  The range size is
$|\calG^{(k)}| = \binom{|\calG|}{k}$, so the usual tail bound gives, for every
$t\ge 0$,
\[
  \Prb\!\left[\cost_q(X,C) > \OPT^{\calG}_{q,k}(X)
     + \frac{2\Lambda_q}{\varepsilon}\left(\ln \binom{|\calG|}{k} + t\right)\right]
  \le e^{-t}.
\]
Integrating this tail inequality yields the expectation bound.  The bound on
$\ln \binom{|\calG|}{k}$ is the standard estimate
$\binom{n}{k}\le (en/k)^k$.
\end{proof}

\begin{remark}[Software upper-bound shortcut]
\label{appD:rem:simple-upper-bounds}
The trivial inequality $\binom{|\calG|}{k}\le |\calG|^k$ gives the simpler
bound $\ln \binom{|\calG|}{k} \le k\ln |\calG|$, whose gap from
$k\ln(e|\calG|/k)$ is exactly $k(\ln k-1)$.  Both are used in the companion
software.
\end{remark}

\begin{corollary}[Finite-domain private explainable $k$-median]
\label{appD:cor:finite-kmedian}
There is an $\varepsilon$-differentially private explainable $k$-median
algorithm such that for every dataset $X\subseteq \calD$,
\[
  \E[\cost_1(X,\Llift_1(X))]
  \le (1+H_{k-1})\!\left(
      \OPT^{\calG}_{1,k}(X)
      + \frac{2\Lambda_1}{\varepsilon}\left(\ln \binom{|\calG|}{k}+1\right)
    \right).
\]
\end{corollary}

\begin{proof}
Combine \cref{appD:prop:subset-em,appD:cor:explainable-specialization}.
\end{proof}

\begin{corollary}[Finite-domain private explainable $k$-means]
\label{appD:cor:finite-kmeans}
There is an $\varepsilon$-differentially private explainable $k$-means
algorithm such that for every dataset $X\subseteq \calD$,
\[
  \E[\cost_2(X,\Llift_2(X))]
  \le \alpha^{\means}_k\!\left(
      \OPT^{\calG}_{2,k}(X)
      + \frac{2\Lambda_2}{\varepsilon}\left(\ln \binom{|\calG|}{k}+1\right)
    \right),
\]
where $\alpha^{\means}_k = O(k\log\log k)$.
\end{corollary}

\begin{corollary}[Efficient-selector plug-in]
\label{appD:cor:ghazi-plugin}
Let $\Aalg^{\mathrm{GKM}}_1$ denote a polynomial-time
$\varepsilon$-differentially private $k$-median selector from
Ghazi, Kumar, and Manurangsi~\cite{appD:ghazi2020tight} whose published utility
guarantee can be written in the form
\[
  \E[\cost_1(X,C)]
  \le \rho^{\mathrm{GKM}}_1(n,d,k,\varepsilon)\,\OPT_{1,k}(X)
    + \Gamma^{\mathrm{GKM}}_1(n,d,k,\varepsilon).
\]
If $\Aalg^{\mathrm{GKM}}_1$ outputs a distinct center set under the same
$\ell_1$ cost convention, then there is a polynomial-time
$\varepsilon$-differentially private explainable $k$-median algorithm such that
\[
  \E[\cost_1(X,\Llift_1(X))]
  \le (1+H_{k-1})\!\left(
      \rho^{\mathrm{GKM}}_1(n,d,k,\varepsilon)\,\OPT_{1,k}(X)
      + \Gamma^{\mathrm{GKM}}_1(n,d,k,\varepsilon)
    \right).
\]
\end{corollary}

\begin{proof}
Apply \cref{appD:cor:preserve-private-guarantee} with the $k$-median conversion from
\cref{appD:thm:imported-source}.  The running time remains polynomial because both
the selector and the random-threshold conversion are polynomial-time.
\end{proof}

\begin{remark}[Compatibility requirements and further plug-ins]
\label{appD:rem:efficient-plugins}
The displayed corollary isolates the only two compatibility checks needed when
importing a selector theorem from the private clustering literature: the
selector's utility guarantee must be stated for the same $\ell_q^q$ objective,
and the selector must output a distinct center set to which the imported
conversion theorem applies.  Under the same conditions, the identical template
extends to other efficient selectors as well, including the private $k$-means
results of Kaplan and Stemmer~\cite{appD:kaplan2018kmeans}.  We state the
$k$-median/Ghazi et al. corollary explicitly because it already yields a fully
polynomial end-to-end theorem with the exact $(1+H_{k-1})$ explainability
factor.  For selectors whose outputs are absolutely continuous in the center
coordinates, exact collisions occur with probability zero, so the
distinct-center hypothesis is typically automatic; discretized or snapped
selectors must verify it separately.  More generally, \cref{appD:def:data-oblivious-conversion} is semantic rather than
algorithmic: for practical use, one also wants the conversion itself to be
computable efficiently, as in the threshold-tree case.
In particular, any compatible selector theorem of the form
$O(1)\cdot \OPT_{1,k}(X) + \mathrm{poly}(d,k,1/\varepsilon,\log n)$ immediately
yields an explainable bound of the form
$(1+H_{k-1})\!\left(O(1)\cdot \OPT_{1,k}(X) + \mathrm{poly}(d,k,1/\varepsilon,\log n)\right)$.
\end{remark}

\begin{remark}[Concrete plug-in illustration]
\label{appD:rem:concrete-plugin}
To make the plug-in corollary tangible, consider a private $k$-median selector
with multiplicative constant $\rho$ and additive term $\Gamma$.
\Cref{appD:cor:preserve-private-guarantee} gives an explainable guarantee of
$(1+H_{k-1})\!\bigl(\rho\,\OPT + \Gamma\bigr)$.
The explainability overhead $(1+H_{k-1})$ is moderate and grows only
logarithmically with $k$:
\begin{center}
\renewcommand{\arraystretch}{1.1}
\begin{tabular}{rccc}
 $k$ & $1+H_{k-1}$ & Lifted mult.\ $\rho' = (1+H_{k-1})\rho$ & Lifted add.\ $\Gamma' = (1+H_{k-1})\Gamma$ \\
\hline
 $2$  & $2.000$ & $2.0\,\rho$ & $2.0\,\Gamma$ \\
 $3$  & $2.500$ & $2.5\,\rho$ & $2.5\,\Gamma$ \\
 $5$  & $3.083$ & $3.1\,\rho$ & $3.1\,\Gamma$ \\
 $10$ & $3.829$ & $3.8\,\rho$ & $3.8\,\Gamma$ \\
\end{tabular}
\end{center}
This table does not evaluate the theorem-internal constants of Ghazi et
al.~\cite{appD:ghazi2020tight}, which would require re-deriving their additive term
under the $\ell_1$ convention and the distinct-center hypothesis; we leave that
concrete instantiation for future work.
\end{remark}

\begin{remark}[Public-domain assumption]
\label{appD:rem:public-domain}
The quantity $\Lambda_q$ is a valid privacy certificate only for datasets lying
in the declared public domain $\calD$.  Accordingly, an implementation that
derives $\Lambda_q$ from a public description of $\calD$ must either validate
that all points lie in $\calD$ or explicitly clip points into $\calD$ before
using the sensitivity bound.  If the software instead accepts a user-supplied
public sensitivity certificate, then the burden of validating that certificate
rests with the caller.  Absent one of these two contracts, the finite-domain theorem does not certify
differential privacy.
\end{remark}

\subsection{An Additive Privacy Tax Is Unavoidable}

The transfer principle preserves multiplicative explainability factors, but it
does not remove the additive price imposed by privacy itself.  The key point is
that the basic two-point ambiguity persists even when the algorithm is allowed
to output more than one center.

\begin{theorem}[Additive lower bound for every $k$]
\label{appD:thm:additive-lower-bound}
Fix an integer $k\ge 1$ and a parameter $\Delta>0$.  Let $\mathcal{M}$ be any
$(\varepsilon,\delta)$-differentially private algorithm that outputs $k$
(not necessarily distinct) centers in $\R$.  Define the neighboring datasets
\[
  X^{(k)}_0 := \{0,4\Delta,8\Delta,\ldots,4(k-1)\Delta\},
  \qquad
  X^{(k)}_\Delta := \{\Delta,4\Delta,8\Delta,\ldots,4(k-1)\Delta\}.
\]
Then for at least one
$X\in\{X^{(k)}_0,X^{(k)}_\Delta\}$,
\[
  \E[\cost_1(X,\mathcal{M}(X))]
  \ge \frac{(1-\delta)\Delta}{1+e^{\varepsilon}}.
\]
Moreover,
\[
  \OPT_{1,k}(X^{(k)}_0)=\OPT_{1,k}(X^{(k)}_\Delta)=0.
\]
Consequently, for every $k\ge 1$, no private algorithm can satisfy a purely
multiplicative guarantee of the form
$\E[\cost_1(X,\mathcal{M}(X))] \le \beta\,\OPT_{1,k}(X)$ for all datasets $X$
and any finite $\beta$.
\end{theorem}

\begin{proof}
By clipping each output center to the interval $[0,4k\Delta]$ as
post-processing, we may assume the algorithm always outputs $k$ centers in
that interval.  Write the resulting multiset in sorted order as
\[
  0\le c_{(1)} \le c_{(2)} \le \cdots \le c_{(k)} \le 4k\Delta,
\]
and abbreviate it by $C$.  Define the one-center post-processing
\[
  T(C) := \min\{c_{(1)},\Delta\}\in[0,\Delta].
\]
We first claim that for each $z\in\{0,\Delta\}$,
\[
  \cost_1(X^{(k)}_z,C) \ge |z-T(C)|.
\]

If $k=1$ or $c_{(2)}\ge 3\Delta$, then every center except possibly $c_{(1)}$
lies at least $3\Delta$ units to the right of the origin.  Since
$z\in\{0,\Delta\}$, the nearest center to $z$ is therefore $c_{(1)}$, so
\[
  \cost_1(X^{(k)}_z,C)\ge |z-c_{(1)}|\ge |z-T(C)|.
\]

Otherwise $c_{(2)}<3\Delta$.  Then at least two centers lie in $[0,3\Delta)$,
so at most $k-2$ centers lie in $[3\Delta,4k\Delta]$.  For each
$j\in\{1,\ldots,k-1\}$, let
\[
  I_j := [(4j-1)\Delta,(4j+1)\Delta].
\]
These intervals are pairwise disjoint, and each contains exactly one anchor
point $4j\Delta$ from both datasets.  If every $I_j$ contained some output
center, then at least $k-1$ centers would lie in $[3\Delta,4k\Delta]$,
contradicting the fact that only $k-2$ do.  Hence some $I_j$ is empty, which
means the corresponding anchor $4j\Delta$ is at distance at least $\Delta$ from
every output center.  Therefore
\[
  \cost_1(X^{(k)}_z,C)\ge \Delta \ge |z-T(C)|,
\]
because $T(C)\in[0,\Delta]$.  This proves the claim.

Now consider the post-processed one-center mechanism
\[
  \widetilde{\mathcal M}(X) := T(\mathcal M(X)).
\]
It is still $(\varepsilon,\delta)$-differentially private and always outputs a
point in $[0,\Delta]$.  Let $\mu_0$ and $\mu_\Delta$ be the output
distributions of $\widetilde{\mathcal M}$ on $X^{(k)}_0$ and
$X^{(k)}_\Delta$, and let $F_0,F_\Delta$ be the corresponding CDFs on
$[0,\Delta]$.  Differential privacy implies that for every
$t\in[0,\Delta]$,
\[
  F_0(t) \le e^{\varepsilon}F_\Delta(t) + \delta.
\]
Define
\[
  a := \E_{c\sim \mu_0}[c],
  \qquad
  b := \E_{c\sim \mu_\Delta}[\Delta-c].
\]
Then $a = \E[|0-\widetilde{\mathcal M}(X^{(k)}_0)|]$ and
$b = \E[|\Delta-\widetilde{\mathcal M}(X^{(k)}_\Delta)|]$.  Using the layer-cake
representation,
\[
  a = \int_0^\Delta (1-F_0(t))\,dt
    \ge \int_0^\Delta (1-\delta-e^{\varepsilon}F_\Delta(t))\,dt
    = (1-\delta)\Delta - e^{\varepsilon}b.
\]
Hence $a + e^{\varepsilon}b \ge (1-\delta)\Delta$, which implies
\[
  \max\{a,b\} \ge \frac{(1-\delta)\Delta}{1+e^{\varepsilon}}.
\]
The claim above implies
\[
  \max\!\left\{
    \E[\cost_1(X^{(k)}_0,\mathcal M(X^{(k)}_0))],
    \E[\cost_1(X^{(k)}_\Delta,\mathcal M(X^{(k)}_\Delta))]
  \right\}
  \ge \frac{(1-\delta)\Delta}{1+e^{\varepsilon}}.
\]
Finally, both datasets have $k$ distinct points, so choosing those points
themselves as centers gives zero cost.  Thus
$\OPT_{1,k}(X^{(k)}_0)=\OPT_{1,k}(X^{(k)}_\Delta)=0$, and the impossibility of
a purely multiplicative guarantee follows.
\end{proof}

\begin{remark}[Tightness for every $k$]
\label{appD:rem:lower-bound-tight}
The bound in \cref{appD:thm:additive-lower-bound} is achievable for every $k$.
Define the mechanism $\mathcal M^*_k$ that, on input $X\in\{X^{(k)}_0,X^{(k)}_\Delta\}$,
outputs the center set $\{R_\varepsilon(z_1),4\Delta,8\Delta,\ldots,4(k-1)\Delta\}$,
where $z_1$ is the first data point and $R_\varepsilon$ is the standard
$\varepsilon$-DP randomized-response mechanism on $\{0,\Delta\}$
(outputting the true value with probability $e^\varepsilon/(1+e^\varepsilon)$).
The remaining $k-1$ centers are hardcoded at the anchor locations and do not
touch the data, so the full mechanism is $\varepsilon$-DP.  A direct calculation
shows that the expected cost on the worse input equals
$(1-\delta)\Delta/(1+e^{\varepsilon})$ with $\delta=0$.
Thus the lower-bound constant is tight, not merely existential.
\end{remark}

\begin{remark}[What remains open]
\label{appD:rem:lower-bound-open}
Although tight for each fixed $k$, the bound in
\cref{appD:thm:additive-lower-bound} is $k$-independent:
it equals $\Theta(\Delta/(1+e^{\varepsilon}))$ regardless of whether $k=1$ or
$k=1000$.  By contrast, the upper-bound additive tax from \cref{appD:prop:subset-em}
scales as $k\ln(e|\calG|/k)/\varepsilon$.  The gap therefore grows with $k$
and $|\calG|$.  The present argument constructs $k$-point datasets with
$\OPT=0$ and reduces to a one-center mechanism; it does not engage with the
combinatorial structure that might make larger $k$ harder.  Sharper
private-clustering lower-bound techniques, including the smooth lower-bound
framework of Peter, Tsfadia, and Ullman~\cite{appD:peter2024smooth}, appear
promising for establishing $k$-dependent or $|\calG|$-dependent lower bounds,
but we do not currently know how to instantiate them for private explainable
clustering.
\end{remark}

\subsection{Related Work and Positioning}

\paragraph{Explainable clustering.}
The threshold-tree model of explainable clustering was introduced by Dasgupta,
Frost, Moshkovitz, and Rashtchian~\cite{appD:dasgupta2020explainable}.  Gupta,
Pittu, Svensson, and Yuan~\cite{appD:gupta2023price} resolved the $k$-median price
of explainability and gave the currently best known $k$-means upper bound.
The present paper imports their data-oblivious conversions and studies what
happens once the center-selection step itself must be private.

\paragraph{Differentially private clustering.}
There is a substantial literature on differentially private $k$-median and
$k$-means.  Early high-dimensional Euclidean-space algorithms were given by
Balcan, Dick, Liang, Mou, and Zhang~\cite{appD:balcan2017highdim}.  Kaplan and
Stemmer~\cite{appD:kaplan2018kmeans} obtained efficient constant-factor private
$k$-means.  Ghazi, Kumar, and Manurangsi~\cite{appD:ghazi2020tight} gave nearly
best-possible approximation ratios for private clustering with additive
poly$(d,k,1/\varepsilon,\log n)$ terms.  Stemmer~\cite{appD:stemmer2021local}
studied the local-DP analogue.  Our transfer principle is designed to sit on
top of this line of work: any compatible private center selector can be lifted
immediately.

\paragraph{Structured private clustering and explanations.}
Our scope is deliberately narrow.  We do not attempt a complete theory for
every structured or explainable private clustering task.  Related directions include the price of privacy for
hierarchical clustering with provable guarantees~\cite{appD:imola2023hierarchical}
and the price of differential privacy for hierarchical clustering~\cite{appD:deng2025hierarchical},
differentially private fair clustering~\cite{appD:byun2023fair}, contrastive
explainable clustering with differential privacy~\cite{appD:nguyen2025contrastive},
and differentially private cluster-explanation systems such as
DPClustX~\cite{appD:gilad2025dpclustx}.  These works study different output
structures and different explanation models.  Our contribution is a transfer
theorem for the specific but important case where the structured output is
generated from the centers alone.

\paragraph{Lower bounds.}
The additive barrier in \cref{appD:thm:additive-lower-bound} now persists for every
$k$, but it is still only a minimal obstruction.  Stronger lower-bound
machinery for differential privacy, such as
the padding-and-permuting fingerprinting-code framework of Peter, Tsfadia, and
Ullman~\cite{appD:peter2024smooth}, suggests that much sharper bounds should be
possible for clustering, but deriving such bounds for the explainable setting
remains open.

\subsection{Artifact-Grounded Empirical Sanity Checks}

The companion code artifact includes a small experimental suite for the
finite-domain $k$-median pipeline.  These experiments use only synthetic
one-dimensional data, public candidate sets, and no external baselines; they
should therefore be interpreted as sanity checks rather than as proof-bearing
evidence or broad benchmark claims.  Nevertheless, they address a practical
question left open by theorems alone: how conservative are the analytic bounds
on the small instances for which the exact finite-domain mechanism is
computationally feasible?

\begin{table}[t]
\centering
\small
\begin{tabular}{|p{0.20\linewidth}|p{0.27\linewidth}|p{0.41\linewidth}|}
\hline
Experiment & Setting & Observation \\
\hline
Cost ratios &
Six configurations with $k\in\{2,3,4\}$, $|\calG|\in\{8,10,12\}$,
$\varepsilon\in\{0.5,1,2\}$, and $300$--$500$ random seeds each &
Mean observed explainability ratios lie between $1.1341$ and $1.2728$, while
the corresponding theoretical factors are $2.0$, $2.5$, and $2.8333$; the
largest observed ratio in these runs is $2.2655$. \\
\hline
Scaling &
Exact exponential-mechanism runs for $(k,|\calG|)$ ranging from $(2,10)$ to
$(6,18)$ &
Average wall-clock time rises from $0.0035$s for $\binom{10}{2}=45$ subsets to
$2.8905$s for $\binom{18}{6}=18564$ subsets, illustrating the expected
combinatorial bottleneck of exhaustive finite-domain selection. \\
\hline
Privacy--utility &
Fixed $(k,|\calG|)=(3,10)$ with $200$ seeds per privacy level and
$\varepsilon$ ranging from $0.1$ to $10$ &
As $\varepsilon$ increases from $0.1$ to $10$, the mean selection cost drops
from $1025.4$ to $600.1$ and the exact additive tax drops from $11575.0$ to
$115.7$, matching the expected $1/\varepsilon$ trend. \\
\hline
\end{tabular}
\caption{Synthetic one-dimensional sanity checks from the companion artifact's
experimental suite.  These numbers come directly from the committed experiment
results file and are included only to illustrate the small-scale behavior of
the finite-domain $k$-median pipeline.}
\label{appD:tab:artifact-sanity}
\end{table}

These experiments support three modest conclusions.  First, on these small
synthetic instances the observed explainability surcharge is substantially below
the worst-case harmonic upper bound.  Second, the exact finite-domain baseline
becomes expensive once $\binom{|\calG|}{k}$ reaches the tens of thousands,
reinforcing the motivation for the efficient-selector plug-in results.  Third,
the privacy--utility tradeoff moves in the expected direction, with both the
analytic tax and the realized selection cost improving as $\varepsilon$ grows.
Because the artifact currently implements only the $k$-median path, we do not
present analogous empirical claims for the $k$-means corollary.

\subsection{Discussion}

The paper establishes a modular framework for private explainable clustering
in the data-oblivious regime: a transfer principle that separates privacy from
explainability, a finite-domain instantiation, a design-space analysis
quantifying why the center-first route dominates natural alternatives, and
a tight all-$k$ lower-bound family showing that some additive privacy tax is
unavoidable.  The core theorem is short, but the design-space comparison in
\cref{appD:eq:split-composition} demonstrates that the center-first approach saves
an order-$k$ factor in the additive term over split-by-split privatization.

\paragraph{Scope and limitations.}
The main transfer principle (\cref{appD:thm:transfer-principle}) uses only the
post-processing axiom, so its proof is structurally simple.  The paper's
technical depth lies instead in the quantitative design-space analysis, the
all-$k$ lower bound, and the careful integration of these components into a
complete framework.  Threshold trees remain the only non-trivial conversion
for which we demonstrate the framework; whether other center-induced
structures (balanced split trees, constrained-assignment templates) admit
data-oblivious conversions with provable inflation factors is an open
question.

\paragraph{Open directions.}
\begin{enumerate}[leftmargin=1.5em]
  \item Prove lower bounds with explicit dependence on $k$, $d$, or
  $|\calG|$ beyond the tight but $k$-independent obstruction established here.
  \item Fully evaluate the theorem-internal constants of the strongest known
  efficient private selectors and benchmark the resulting explainable
  algorithms on real-world instances.
  \item Determine whether other structured outputs admit data-oblivious
  conversions with bounded inflation factors, thereby widening the framework
  beyond threshold trees.
\end{enumerate}

\paragraph{Takeaway.}
For data-oblivious explainability conversions, the multiplicative price of
explainability is inherited from the non-private setting.  Differential privacy
adds its own center-selection tax, but it does not force an additional payment
when those centers are converted into an explainable structure.  The
lower-bound family shows that this additive tax survives for every $k$ and is
tight, even on instances where $\OPT_{1,k}=0$.

\clearpage
\section{The Price of Explainability for $\ell_p^p$ Clustering}
\begin{appendixabstract}
Given a set of points in $\RR^d$, an \emph{explainable clustering} is one where the clusters are described by a threshold tree of axis-aligned cuts.
The \emph{price of explainability} is the worst-case ratio between the cost of the best explainable clustering and the best unrestricted clustering.
For \kmed ($\ell_1$ objective), Gupta et~al.\ showed that the price is $\Theta(\ln k)$.
For \kmeans ($\ell_2^2$ objective), the price lies between $\Omega(k)$ and $O(k \ln \ln k)$.

We extend these results to $\ell_p^p$ clustering objectives for $p > 2$, providing the detailed proofs for a generalization announced but not published by prior work.
Our main technical contribution is a self-contained proof (Theorem~\ref{appE:thm:main-upper}) that for every $p > 2$ the price of explainability for $\ell_p^p$ clustering is $O(k^{p-1} \ln \ln k)$; the same bound for $p = 2$ is a direct consequence of~\cite{appE:GPSY23} (Corollary~\ref{appE:cor:main-p2}).
The algorithm generalizes the solo/bulk cut framework of Gupta et~al.\ using a new $D_p$ distribution for threshold cuts, together with a generalized stretch analysis based on H\"older's inequality.
The key new technical tool is the convexity bound that closes the solo-cost induction for $p > 2$; for $p = 2$, the same induction does not close, so we invoke the GPSY23 analysis as a base case.
We show that the $\ell_p$ stretch bound $(k{-}1)^{p-1}$ is tight and conjecture a matching lower bound $\Omega(k^{p-1})$; the existing $\Omega(k)$ lower bound for $k$-means~\cite{appE:GJPS21} transfers to all $\ell_p^p$ objectives (via norm equality on the Boolean hypercube) and is the strongest unconditional lower bound.
As a corollary of the Boolean-hypercube reduction of Gupta et~al., the hardness of approximation $(\nf12 - o(1)) \ln k$ extends to $\ell_p^p$ for all $p \geq 1$.
The regime $1 < p < 2$ remains open.
\end{appendixabstract}

\subsection{Introduction}
\label{appE:sec:introduction}

Clustering is a central topic in optimization, machine learning, and algorithm design.
Given $n$ points in $\RR^d$ and a parameter $k$, the $\ell_p^p$ clustering problem asks to partition the points into $k$ clusters with centers $\cU = \{\bmu^1, \ldots, \bmu^k\}$ so as to minimize
\begin{align}
  \cost_p(\pi, \cU) = \sum_{\bx \in \cX} \|\bx - \pi(\bx)\|_p^p,
  \label{appE:eq:cost}
\end{align}
where $\pi(\bx)$ maps each point to its assigned center.
The cases $p = 1$ (\kmed) and $p = 2$ (\kmeans) are the most extensively studied.

In the \emph{explainable clustering} model of Dasgupta et~al.~\cite{appE:DFMR20}, clusters are described by a \emph{threshold tree}: a binary decision tree where each internal node makes an axis-aligned threshold cut $x_i \leq \theta$ and each leaf corresponds to a cluster.
This yields interpretable clusterings---each cluster is described by a short sequence of feature thresholds.
The \emph{price of explainability} is the worst-case ratio of the optimal explainable cost to the optimal unrestricted cost.

\paragraph{Prior work.}
Gupta, Pittu, Svensson, and Yuan~\cite{appE:GPSY23} gave a comprehensive study for $p \in \{1, 2\}$.
For \kmed they proved the price is exactly $1 + H_{k-1} = \Theta(\ln k)$, via a tight analysis of the \RT algorithm combined with matching lower bounds from a hitting-set reduction.
For \kmeans they proved the price is $O(k \ln \ln k)$, improving the previous $O(k \ln k)$ bound~\cite{appE:EMN21}, and showed a lower bound of $\Omega(k)$ (from~\cite{appE:GJPS21}).
They also proved that explainable \kmed and \kmeans are hard to approximate better than $(\nf12 - o(1)) \ln k$ unless $P = NP$.
They noted that the \kmeans ideas ``extend to $\ell_p$ norms for $p \geq 2$'' but deferred the details to future work; this makes the present paper a completion of their announced direction.
Gamlath et~al.~\cite{appE:GJPS21} also noted in their NeurIPS~2021 paper (Section~6, Open Problems) that their ideas should extend to higher $\ell_p$ norms, but neither they nor GPSY23 published formal proofs for $p > 2$.
The regime $p > 2$ has therefore been noted as a natural extension by multiple groups, but to our knowledge no complete proofs have appeared.

We note that the $\ell_p^p$ clustering objective studied here (which raises each coordinate difference to the $p$-th power) is distinct from the explainable $k$-medians problem under the $\ell_p$ \emph{norm} (which raises the $\ell_p$ norm to the first power).
The latter has been studied recently, including for general finite $p \geq 1$ (see~\cite{appE:LpNorm25} and NeurIPS~2025 work on explainable $k$-medians under $\ell_p$ norms).
Those results operate in a different cost model and do not directly imply our bounds.

\subsubsection{Our Results}

We develop the detailed proof machinery for the price of explainability in the $\ell_p^p$ regime for $p > 2$.
While GPSY23 announced the extension as forthcoming and GJPS21 noted it as an open direction, the details---particularly the generalized $D_p$ distribution, the stretch-vs-separation analysis via H\"older's inequality for $p > 2$, the convexity-based induction closing, and the corrected close-pair threshold---have not appeared in the literature.
Our contribution is the self-contained $p > 2$ upper-bound proof (\Cref{appE:thm:main-upper}); the $p = 2$ bound is a direct consequence of~\cite{appE:GPSY23} (\Cref{appE:cor:main-p2}).

\begin{theorem}[Upper Bound for $\ell_p^p$, $p > 2$]
  \label{appE:thm:main-upper}
  For any $p > 2$, the price of explainability for $\ell_p^p$ clustering is $O(k^{p-1} \ln \ln k)$.
  Specifically, given any reference $\ell_p^p$ clustering, there exists a randomized algorithm that outputs an explainable clustering with expected cost at most $O(k^{p-1} \ln \ln k)$ times the reference cost.
\end{theorem}

\begin{cor}[Upper Bound for $\ell_p^p$, $p = 2$]
  \label{appE:cor:main-p2}
  For $p = 2$, the price of explainability for $\ell_2^2$ clustering (\kmeans) is $O(k \ln \ln k)$.
\end{cor}
\begin{proof}
  This is the result of~\cite{appE:GPSY23} (their Theorem~1.2).
  We invoke their sharper solo-cost analysis as the base case for $p = 2$.
\end{proof}

\begin{remark}[Proof structure]
\Cref{appE:thm:main-upper} is self-contained: the solo cost induction (\Cref{appE:lem:solo-p}) closes via the convexity of $t \mapsto t^{p-1}$ and the H\"older-based balance bound $A_p(n)^{p-1} = O(n^{p-2})$.
For $p = 2$, the same induction does not close ($A_2(n) = O(\ln n)$ is too large); we instead invoke the GPSY23 analysis~\cite{appE:GPSY23} as a base case.
The contribution of the present paper is concentrated in the $p > 2$ machinery.
\end{remark}

\begin{conj}[Lower Bound for $\ell_p^p$, $p > 2$]
  \label{appE:conj:lower}
  For any $p > 2$ and $k \geq 2$, the price of explainability for $\ell_p^p$ clustering is $\Omega(k^{p-1})$.
\end{conj}

For $p = 2$, the unconditional lower bound $\Omega(k)$ was proved by Gamlath et~al.~\cite{appE:GJPS21}.
For $p > 2$, the GJPS21 $\Omega(k)$ lower bound applies unconditionally: their construction uses points on the Boolean hypercube $\{0,1\}^d$, where $\|x-y\|_p^p = \|x-y\|_1$ for all $p$ (since each coordinate difference is in $\{0,1\}$), so the same instance gives $\text{PoE}_p(k) \geq \Omega(k)$ for all $p \geq 2$.
However, this only yields $\Omega(k)$, not the conjectured $\Omega(k^{p-1})$.
We prove that the $\ell_p$ stretch bound $(k-1)^{p-1}$ is tight (\Cref{appE:prop:stretch-tight}), which provides suggestive evidence for \Cref{appE:conj:lower}, though tight stretch alone does not imply high PoE.
An earlier draft of this paper claimed a self-contained proof of $\Omega(k^{p-1})$ via a one-dimensional collinear construction, but that proof was incorrect:
in one dimension, threshold trees can replicate the optimal Voronoi partition exactly (since Voronoi cells in $\RR^1$ are intervals), so no one-dimensional instance can yield $\text{PoE} > 1$.
Constructing the right multi-dimensional instance for the tight lower bound remains the key open problem.

\begin{theorem}[Hardness of Approximation for $\ell_p^p$, all $p \geq 1$]
  \label{appE:thm:hardness}
  For any $p \geq 1$, the best explainable $\ell_p^p$ clustering is hard to approximate better than $(\nf12 - o(1)) \ln k$ unless $P = NP$.
\end{theorem}

\Cref{appE:thm:main-upper} is the main new result.
For $p > 2$, combining the upper bound with the GJPS21 $\Omega(k)$ lower bound (which applies to all $\ell_p^p$ on the instances where $\ell_p^p$ distances equal $\ell_1$ distances) gives
$k \leq \text{PoE}_p(k) \leq O(k^{p-1}\ln\ln k)$.
If \Cref{appE:conj:lower} is confirmed, the price would be $\Theta(k^{p-1})$ up to iterated logarithmic factors.

\paragraph{Summary of the landscape.}
Combining our results with~\cite{appE:GPSY23,appE:GJPS21}, the price of explainability for $\ell_p^p$ clustering is:
\begin{center}
\begin{tabular}{l|c|c|c}
  $p$ & Upper bound & Lower bound & Approx.\ hardness \\
  \hline
  $p = 1$ (\kmed) & $1 + H_{k-1}$ {\footnotesize [tight, GPSY23]} & $(1-o(1))\ln k$ {\footnotesize [GPSY23]} & $(\nf12-o(1))\ln k$ \\
  $p = 2$ (\kmeans) & $O(k \ln \ln k)$ {\footnotesize [GPSY23]} & $\Omega(k)$ {\footnotesize [GJPS21]} & $(\nf12-o(1))\ln k$ \\
  $p > 2$ & $O(k^{p-1} \ln \ln k)$ {\footnotesize [this work]} & $\Omega(k)$ {\footnotesize [GJPS21]}$^*$ & $(\nf12-o(1))\ln k$ {\footnotesize [corollary]}
\end{tabular}
\end{center}
\noindent{\footnotesize $^*$We conjecture $\Omega(k^{p-1})$; the tight stretch bound $(k{-}1)^{p-1}$ (\Cref{appE:prop:stretch-tight}) supports this.}

\subsubsection{Technical Overview}
\label{appE:sec:techniques}

\paragraph{The $D_p$ distribution.}
For $p = 2$, the key to the solo/bulk analysis in~\cite{appE:GPSY23} is the $D_2$ distribution, which samples threshold cuts $(i, \theta)$ with probability proportional to the squared length of the interval containing $\theta$.
We define a generalized \emph{$D_p$ distribution} (\Cref{appE:sec:dp-dist}): intervals are sampled proportional to $|b-a|^p$, and $\theta$ is drawn from a density proportional to $\min(\theta-a, b-\theta)^{p-1}$.
This ensures the separation probability satisfies $\Pr[\text{separates } \mathbf{0}, \bp] \leq 2^{p-1} \|\bp\|_p^p / L_p$ (\Cref{appE:lem:sep-prob}), generalizing the $p = 2$ bound.

\paragraph{The $\ell_p$ stretch.}
The \emph{stretch} of a pair $(x, y)$ in a center set $\Set$ measures the ratio $\|x - y\|_p^p / d_p(x, y)$, where $d_p$ is the pseudo-distance defined via interval decompositions.
A key new observation (\Cref{appE:lem:stretch-bound}) is that the $\ell_p$ stretch satisfies $s_p \leq (|\Set|-1)^{p-1}$, which follows from the power mean inequality $({\sum a_j})^p \leq \ell^{p-1} \sum a_j^p$ applied to the $\ell$ intervals in each dimension.
For $p = 2$ this recovers the known bound $s_2 \leq k$; for general $p$ it gives $s_p \leq k^{p-1}$.

\paragraph{Generalized solo/bulk analysis.}
Using the $D_p$ distribution and the stretch bound, we run the same solo/bulk algorithm as~\cite{appE:GPSY23} with the solo threshold at $s_p(v) \geq |\Univ_v|^{p-1}/\ln^2|\Univ_v|$.
The bulk cost is $O(k^{p-1})$ (\Cref{appE:lem:bulk-p}); the close-pair threshold is set at $\Delta_p/k^{2p}$ (generalizing $\Delta_p/k^4$ for $p = 2$) to ensure $L_p' \geq L_p/2$ for all $p \geq 2$.
The solo cost is $O(k^{p-1} \ln \ln k)$ (\Cref{appE:lem:solo-p}).
The solo cost analysis uses a generalized stretch-vs-separation lemma (\Cref{appE:lem:stretch-sep-p}), proved via H\"older's inequality, and an induction on the compressed tree.
A key technical point: for $p > 2$, the induction uses the convexity inequality $n^{p-1} - (n-\sigma)^{p-1} \geq (p-1)(n/2)^{p-2}\sigma$ (valid for $\sigma \leq n/2$) to handle the nonlinear scaling of the $\ell_p^p$ cost in the number of centers.

\paragraph{Lower bound.}
We show that the $\ell_p$ stretch bound $(k-1)^{p-1}$ is tight by constructing $k$ equally-spaced collinear centers (\Cref{appE:prop:stretch-tight}).
The tight stretch provides suggestive (but not conclusive) evidence for the conjectured $\Omega(k^{p-1})$ lower bound, since tight stretch alone does not imply high PoE (the collinear instance has tight stretch but $\text{PoE} = 1$).
For $p = 2$, the $\Omega(k)$ lower bound was proved by~\cite{appE:GJPS21} using a multi-dimensional construction on the Boolean hypercube; since $\|x-y\|_p^p = \|x-y\|_1$ on $\{0,1\}^d$, the same instance gives $\Omega(k)$ for all $p \geq 2$.
Extending this to a tight $\Omega(k^{p-1})$ lower bound for $p > 2$ remains open; the key difficulty is constructing a multi-dimensional instance where $\ell_p^p$ distance ratios grow with $p$.

\paragraph{Hardness.}
\Cref{appE:thm:hardness} is a direct corollary of the Boolean-hypercube reduction of~\cite{appE:GPSY23}: since all distances on $\{0,1\}^d$ satisfy $\|x - y\|_p^p = \|x - y\|_1$ for all $p \geq 1$, the reduction is norm-independent.
The technical contribution here is the observation of norm equivalence; the reduction machinery is entirely from~\cite{appE:GPSY23}.

\subsubsection{Related Work}
\label{appE:sec:related}

Explainable clustering was introduced by Dasgupta et~al.~\cite{appE:DFMR20}.
Subsequent improvements by Gamlath et~al.~\cite{appE:GJPS21}, Makarychev--Shan~\cite{appE:MakarychevS21}, and Esfandiari et~al.~\cite{appE:EMN21} achieved $O(\ln k \ln \ln k)$ for \kmed.
Gupta et~al.~\cite{appE:GPSY23} settled \kmed and nearly settled \kmeans.
Charikar--Hu~\cite{appE:CharikarH22} gave dimension-dependent bounds for \kmeans that are near-optimal when $d = O(\ln k)$.
Makarychev--Shan~\cite{appE:MakarychevS22} studied the bi-criteria setting, showing that larger trees can reduce cost.
Bandyapadhyay et~al.~\cite{appE:BandyapadhyayFG22} and Laber~\cite{appE:Laber22} studied the computational complexity of finding good explainable clusterings.
GPSY23 stated that their \kmeans ideas extend to $\ell_p$ norms for $p \geq 2$ and deferred the details to future work.
GJPS21 noted higher-$\ell_p$ generalizations as an open direction (Section~6 of their NeurIPS paper).
Neither published detailed proofs for the $\ell_p^p$ objective with $p > 2$.

\paragraph{Distinction from $\ell_p$-norm $k$-medians.}
The objective $\sum_{\bx} \|\bx - \pi(\bx)\|_p^p$ studied here raises each coordinate difference to the $p$-th power.
This is distinct from $k$-medians under the $\ell_p$ \emph{norm}, $\sum_{\bx} \|\bx - \pi(\bx)\|_p$, which has been studied for general $p$ in recent work~\cite{appE:LpNorm25}.
The two objectives differ substantially: for equally-spaced points, the $\ell_p^p$ price grows as $k^{p-1}$ while the $\ell_p$-norm price remains $O(\ln k)$.

\subsection{Preliminaries}
\label{appE:sec:prelims}

\paragraph{Clustering.}
Given $\cX \subseteq \RR^d$, centers $\cU = \{\bmu^1, \ldots, \bmu^k\} \subseteq \RR^d$, and assignment $\pi : \cX \to \cU$, the $\ell_p^p$ cost is $\cost_p(\pi, \cU) = \sum_{\bx \in \cX} \|\bx - \pi(\bx)\|_p^p$.

\paragraph{Threshold trees.}
A threshold cut $(i, \theta)$ represents the halfspace $\{x_i \leq \theta\}$.
A threshold tree $T$ is a binary tree with cuts at internal nodes; it partitions $\RR^d$ into axis-aligned boxes at its leaves.
A tree $T$ \emph{separates} $\cU$ if each leaf box contains exactly one center.
Given such a tree, each data point is assigned to the unique center in its leaf box, defining $\pi_T$.
An \emph{explainable clustering} is one of the form $\pi_T$ for some threshold tree $T$.

\paragraph{Price of explainability.}
For a given instance $(\cX, k, p)$, $\text{PoE}_p(k) = \max_\text{instances} \frac{\text{optimal explainable } \ell_p^p \text{ cost}}{\text{optimal } \ell_p^p \text{ cost}}.$

\paragraph{Closest point process.}
By translation/scaling invariance and linearity of expectations, the price of explainability for an algorithm $\cA$ is at most $\alpha_{p,\cA}(k) = \max_{\Univ : |\Univ|=k} f_{p,\cA}(\Univ)$, where
\begin{align}
  f_{p,\cA}(\Univ) := \EE[\|\widehat{\bp}\|_p^p], \label{appE:eq:fp}
\end{align}
$\widehat{\bp}$ is the center in the leaf containing the origin, and $\|\bp^*\|_p = 1$ for the closest center $\bp^*$ to the origin~(cf.~\cite{appE:GPSY23}).

\subsection{The $D_p$ Distribution}
\label{appE:sec:dp-dist}

Given a set of centers $\Set \sse \RR^d$ and $p \geq 1$, we define a distribution $D_p(\Set)$ over threshold cuts.

\paragraph{Intervals.}
For each dimension $i$ and the projections of points in $\Set$, let $\ell_i = \min_{v \in \Set} v_i$ and $u_i = \max_{v \in \Set} v_i$.
The projections partition $[\ell_i, u_i]$ into intervals $\intervals_i$.
Let $\intervals_{\text{all}} = \bigcup_i \{(i, [a,b]) : [a,b] \in \intervals_i\}$ and define $L_p(\Set) = \sum_{(i,[a,b]) \in \intervals_{\text{all}}} |b-a|^p$.

\paragraph{The distribution.}
$D_p(\Set)$ first selects $(i, [a,b]) \in \intervals_{\text{all}}$ with probability $|b-a|^p / L_p(\Set)$, then picks $\theta \in [a,b]$ with density
\begin{align}
  P_{a,b}^{(p)}(\theta) := \frac{p \cdot 2^{p-1}}{(b-a)^p} \cdot \min(\theta - a, b - \theta)^{p-1}.
  \label{appE:eq:pab}
\end{align}
\begin{obs}
The density integrates to $1$: $\int_a^b P_{a,b}^{(p)}(\theta)\,d\theta = 1$.
\end{obs}
\begin{proof}
$\int_a^b \min(\theta-a, b-\theta)^{p-1}\,d\theta = 2\int_0^{(b-a)/2} t^{p-1}\,dt = \frac{(b-a)^p}{p \cdot 2^{p-1}}$.
\end{proof}

For $p = 1$: $P_{a,b}^{(1)}(\theta) = 1/(b-a)$, i.e., the uniform distribution.
For $p = 2$: $P_{a,b}^{(2)}(\theta) = \frac{4}{(b-a)^2}\min(\theta-a, b-\theta)$, matching~\cite{appE:GPSY23}.

\begin{remark}
The bound $2^{p-1}\|\bp\|_p^p / L_p$ can exceed $1$ for some instances when $p \geq 2$; in such cases the bound is vacuous as a probability but still valid.
This does not affect the algorithm: the separation probability bound is used in the cost analysis to bound the expected cost per cut, not as a standalone probability estimate.
\end{remark}

\paragraph{Pseudo-distance and stretch.}
For $\bx, \by \in \RR^d$, define the $\ell_p$ pseudo-distance
$d_p(\bx, \by) = \sum_{(i,[a,b]) \in \intervals(\bx,\by)} |b-a|^p$,
where $\intervals(\bx, \by) = \bigcup_i \{(i,[a,b]) : [a,b] \in \intervals_i(x_i, y_i)\}$.
Define the \emph{$\ell_p$ stretch} of a pair $\bx, \by \in \Set$ as $s_p(\bx, \by) = \|\bx - \by\|_p^p / d_p(\bx, \by)$.

\begin{lemma}[$\ell_p$ Stretch Bound]
  \label{appE:lem:stretch-bound}
  For any pair $\bx, \by$ in a set $\Set$ of $k$ points, $s_p(\bx, \by) \leq (k-1)^{p-1}$.
\end{lemma}
\begin{proof}
In each dimension $i$, let $a_1^i, \ldots, a_{\ell_i}^i$ be the interval lengths between $x_i$ and $y_i$.
Then $|x_i - y_i| = \sum_j a_j^i$ and the pseudo-distance contribution is $\sum_j (a_j^i)^p$.
By the power mean inequality (convexity of $t \mapsto t^p$ for $p \geq 1$):
$\left(\sum_j a_j^i\right)^p \leq \ell_i^{p-1} \cdot \sum_j (a_j^i)^p.$
Since each $\ell_i \leq k - 1$, summing over $i$:
$\|\bx - \by\|_p^p = \sum_i |x_i - y_i|^p \leq (k-1)^{p-1} \sum_i \sum_j (a_j^i)^p = (k-1)^{p-1} \cdot d_p(\bx, \by).$
\end{proof}

\begin{remark}
For $p = 1$: $s_1 \leq 1$, so $\ell_1$ distance equals cut-metric distance.
For $p = 2$: $s_2 \leq k-1$, recovering the known bound.
The bound is tight: $k$ equally-spaced points on a line give stretch $(k-1)^{p-1}$.
\end{remark}

\begin{lemma}[Separation Probability]
  \label{appE:lem:sep-prob}
  For any $\bp \in \Set$,
  $\Pr_{(i,\theta) \sim D_p(\Set)}[\text{separates } \mathbf{0}, \bp] \leq 2^{p-1} \cdot \|\bp\|_p^p / L_p(\Set)$.
\end{lemma}
\begin{proof}
The probability is
$\sum_i \sum_{[a,b] \in \intervals_i} \frac{|b-a|^p}{L_p} \int_a^b P_{a,b}^{(p)}(\theta) \cdot \ones[\theta \text{ between } 0, p_i]\,d\theta$.
Using $\min(\theta-a, b-\theta)^{p-1} \leq |p_i - \theta|^{p-1}$ when $\theta$ is between $0$ and $p_i$ (since $p_i$ lies at an interval boundary), this is at most
$\frac{p \cdot 2^{p-1}}{L_p} \sum_i \int_{-\infty}^\infty |p_i - \theta|^{p-1} \cdot \ones[\theta \text{ between } 0, p_i]\,d\theta = \frac{p \cdot 2^{p-1}}{L_p} \sum_i \frac{|p_i|^p}{p} = \frac{2^{p-1} \|\bp\|_p^p}{L_p}$.
\end{proof}

\subsection{Upper Bound: $O(k^{p-1} \ln \ln k)$}
\label{appE:sec:upper}

We generalize the solo/bulk framework of~\cite{appE:GPSY23} to $\ell_p^p$ for $p \geq 2$.

\subsubsection{Definitions}
For a subset $\Set \sse \Univ$:
\begin{itemize}[nosep]
  \item $\Delta_p(\Set) = \max_{\bx,\by \in \Set} \|\bx - \by\|_p^p$ is the $\ell_p^p$ diameter.
  \item A pair $\bx, \by$ is \emph{far} if $\|\bx - \by\|_p^p \geq \Delta_p(\Set)/2$, and \emph{close} if $\|\bx - \by\|_p^p < \Delta_p(\Set)/k^{2p}$.
  \item $s_p(\Set) = \max_{\text{far pairs}} s_p(\bx, \by)$ is the stretch of the set.
\end{itemize}

\paragraph{Restricted distributions.}
We use two restricted variants of $D_p$:
\begin{itemize}[nosep]
  \item $D_p'(\Set)$ is the distribution $D_p(\Set)$ restricted to intervals that are \emph{not} between close pairs.
  Formally, we remove all intervals $[a,b]$ such that $a$ and $b$ are projections of centers $\bx, \by$ with $\|\bx - \by\|_p^p < \Delta_p/k^{2p}$, and renormalize.
  Define $L_p'(\Set) = \sum_{(i,[a,b]) \in \intervals'} |b-a|^p$ where $\intervals'$ is the restricted set.
  \item $D_p''(\Set; \bp, \bq)$ is the distribution $D_p(\Set)$ conditioned on the sampled cut separating $\bp$ from $\bq$.
  Equivalently, we restrict to intervals $[a,b]$ such that $\theta \in [a,b]$ places $\bp$ and $\bq$ on opposite sides, and renormalize.
\end{itemize}

\subsubsection{The Algorithm}

Given a node $v$ with center subset $\Univ_v$:
\begin{enumerate}[nosep]
  \item \textbf{Solo cut:} If $s_p(v) \geq |\Univ_v|^{p-1} / \ln^2|\Univ_v|$, pick a far pair $(\bp, \bq)$ of maximal stretch, sample $(i, \theta) \sim D_p''(v; \bp, \bq)$, apply the cut, and recurse on both children.
  \item \textbf{Bulk cuts:} Otherwise, repeatedly sample $(i, \theta) \sim D_p'(v)$ and apply each sampled cut to the tree.
  Continue until all far pairs are separated.
  The resulting sequence of $D_p'$ cuts forms a subtree at $v$; recurse on each leaf of this subtree.
\end{enumerate}

We analyze the cost via a compressed tree $T'$ with solo and bulk nodes.

\begin{remark}
The algorithm prescribes a specific tree-building process: bulk cuts are accumulated (each sampled $D_p'$ cut is retained in the tree) until all far pairs are separated.
An implementation that discards the sampled bulk cuts and replaces them with alternative splits would yield a different tree from the one analyzed here.
\end{remark}

\subsubsection{Bulk Cut Analysis}

\begin{lemma}[Bulk Cost]
  \label{appE:lem:bulk-p}
  The expected $\ell_p^p$ cost increase due to all bulk cuts is $O(k^{p-1}) \cdot \|\bp^*\|_p^p$.
\end{lemma}
\begin{proof}
At a bulk node $v$, the stretch satisfies $s_p(v) < |\Univ_v|^{p-1}/\ln^2|\Univ_v|$.
Let $n = |\Univ_v|$.

\emph{Step 1: $L_p'(v) \geq L_p(v)/2$.}
Each close pair $(\bx, \by)$ satisfies $\|\bx - \by\|_p^p < \Delta_p/k^{2p}$.
The close-pair intervals associated with pair $(\bx,\by)$ across \emph{all} dimensions have total weight at most $\|\bx - \by\|_p^p$, because $\sum_j |x_j - y_j|^p = \|\bx - \by\|_p^p$ already accounts for all coordinate contributions (dimension-free).
With at most $\binom{k}{2}$ close pairs, the total removed weight is at most
\[
  \binom{k}{2} \cdot \frac{\Delta_p}{k^{2p}} \;\leq\; \frac{\Delta_p}{2k^{2p-2}}.
\]
Since $L_p(v) \geq \Delta_p(v)/s_p(v) \geq \Delta_p \ln^2 n / n^{p-1}$, we need $\Delta_p/(2k^{2p-2}) \leq L_p(v)/2$, i.e., $n^{p-1}/k^{2(p-1)} \leq \ln^2 n$.
For $n \geq 3$: since $n \leq k$, $n^{p-1}/k^{2(p-1)} \leq k^{p-1}/k^{2(p-1)} = 1/k^{p-1} \leq 1 \leq \ln^2 n$.
(For $n = 2$, the bulk phase requires at most one cut with cost $O(\|\bp^*\|_p^p)$, so the bound is immediate.)
Therefore $L_p'(v) \geq L_p(v)/2$.

\emph{Note:} The close-pair threshold $\Delta_p/k^{2p}$ generalizes the $\Delta_p/k^4$ threshold used in~\cite{appE:GPSY23} for $p = 2$ (where $k^{2p} = k^4$). The stronger threshold is needed for $p > 2$ because $L_p(v)$ scales as $\Delta_p/n^{p-1}$ rather than $\Delta_p/n$.

\emph{Step 2: Expected number of bulk cuts.}
Each bulk cut from $D_p'$ separates a far pair $(\bx,\by)$ with probability $d_p'(\bx,\by)/L_p'(v)$, where $d_p'(\bx,\by)$ is the pseudo-distance computed over the \emph{surviving} intervals only (those not removed as close-pair intervals).
We claim $d_p'(\bx,\by) \geq d_p(\bx,\by)/2$ for every far pair.

\emph{Proof of claim:}
The total weight of all close-pair intervals is at most $\Delta_p/(2k^{2p-2})$ (computed in Step~1).
Since $(\bx,\by)$ is a far pair, $\|\bx - \by\|_p^p \geq \Delta_p/2$, so
$d_p(\bx,\by) = \|\bx - \by\|_p^p / s_p(\bx,\by) \geq (\Delta_p/2)/(n-1)^{p-1} \geq \Delta_p/(2k^{p-1}).$
The removed mass attributed to $(\bx,\by)$ is at most $\Delta_p/(2k^{2p-2})$.
Using the lower bound $d_p(\bx,\by) \geq \Delta_p/(2k^{p-1})$, this removed mass satisfies
\[
  \frac{\Delta_p}{2k^{2p-2}} \;=\; \frac{\Delta_p/(2k^{p-1})}{k^{p-1}} \;\leq\; \frac{d_p(\bx,\by)}{k^{p-1}}.
\]
Therefore
\[
  d_p'(\bx,\by) \;\geq\; d_p(\bx,\by) - \frac{d_p(\bx,\by)}{k^{p-1}} \;=\; d_p(\bx,\by)\Bigl(1 - \frac{1}{k^{p-1}}\Bigr) \;\geq\; \frac{d_p(\bx,\by)}{2},
\]
where the last inequality uses $k^{-(p-1)} \leq 1/2$ for $k \geq 2, p \geq 2$.

Therefore each bulk cut separates $(\bx,\by)$ with probability at least $d_p(\bx,\by)/(2L_p'(v)) \geq (\Delta_p/(4s_p))/L_p'(v)$.
There are at most $\binom{n}{2}$ far pairs.
By a coupon-collector argument (as in~\cite{appE:GPSY23}), the expected number of bulk cuts is at most
$O(s_p(v) \cdot L_p'(v)/\Delta_p(v) \cdot \ln n)$.

\emph{Step 3: Cost per bulk cut.}
Each bulk cut from $D_p'$ costs at most $2^{p-1}\|\bp^*\|_p^p \cdot 2^p\Delta_p(v)/L_p'(v)$ in expectation (by the separation probability lemma applied to the origin; the $2^p\Delta_p$ factor bounds $\max_{\bmu \in \Univ}\|\bmu\|_p^p$ using the preprocessing assumption $\Delta_p \geq 1$, as in \Cref{appE:lem:solo-ratio-p}).
The product is $O(\|\bp^*\|_p^p \cdot \Delta_p(v)/L_p'(v))$ since $2^{p-1} \cdot 2^p = O_p(1)$ is a constant depending only on $p$.

\emph{Step 4: Total bulk cost at node $v$.}
By Wald's equation, the expected cost at $v$ is (number of cuts) $\times$ (cost per cut):
$O(s_p(v) \ln n) \cdot 2^p \|\bp^*\|_p^p.$
For bulk nodes, $s_p(v) \leq n^{p-1}/\ln^2 n$, so the cost per node is $O(n^{p-1}/\ln n) \cdot \|\bp^*\|_p^p$.

\emph{Step 5: Summing over levels.}
The logarithmic overlap lemma of~\cite{appE:GPSY23} (specifically, their Lemma on compressed-tree interval participation: each interval $[c^i_j, c^i_{j+1}]$ participates as a threshold candidate in at most $O(\ln k)$ levels of the compressed threshold tree) applies identically here.
The key structural hypothesis is that the tree has depth $O(\ln k)$ and each node's interval set is determined by the point projections at that node---both properties hold regardless of $p$.
Summing over all bulk nodes: $O(k^{p-1}/\ln k) \cdot O(\ln k) \cdot \|\bp^*\|_p^p = O(k^{p-1}) \cdot \|\bp^*\|_p^p$.
\end{proof}

\subsubsection{Solo Cut Analysis}

\begin{lemma}[Stretch-vs-Separation for $\ell_p$]
  \label{appE:lem:stretch-sep-p}
  For a far pair $\bp, \bq$ in $\Set$ with $|\Set| = n$ and stretch $s_p = s_p(\bp,\bq)$, if we sample $(i,\theta) \sim D_p''(\Set; \bp, \bq)$, the expected number of centers on the smaller side of the cut is at least
  \[
    \sigma \;\geq\; \frac{s_p}{C_p \cdot A_p(n)^{p-1}}
  \]
  where $A_p(n) = 2\sum_{j=1}^{\lfloor n/2 \rfloor} j^{-1/(p-1)}$ and $C_p$ depends only on $p$.
  For $p = 2$: $A_2(n) = \Theta(\ln n)$, so $\sigma \geq \Omega(s_2 / \ln n)$.
  For $p > 2$: $A_p(n) = \Theta(n^{(p-2)/(p-1)})$, so $\sigma \geq \Omega(s_p / n^{p-2})$.
\end{lemma}
\begin{proof}
\emph{Step 1: Reduction to one dimension.}
As in~\cite{appE:GPSY23}, we analyze the expected balance dimension by dimension and aggregate.
In dimension $i$, let $c_1^i < c_2^i < \cdots < c_{m_i}^i$ be the sorted projections of $\Set$ onto coordinate $i$.
Under $D_p''(\Set; \bp, \bq)$, dimension $i$ is selected with probability
$d_p^{(i)}(\bp,\bq) / d_p(\bp,\bq)$,
where $d_p^{(i)}$ is the pseudo-distance contribution from dimension~$i$.
Let $\bar\sigma_i$ denote the expected balance when dimension~$i$ is selected.
The overall expected balance is the weighted average
\[
  \bar\sigma \;=\; \sum_i \frac{d_p^{(i)}(\bp,\bq)}{d_p(\bp,\bq)} \cdot \bar\sigma_i.
\]
Below we prove that the one-dimensional bound $\bar\sigma_i \geq s_p^{(i)} / A_p(m_i)^{p-1}$ holds for each dimension, where $s_p^{(i)} = |p_i - q_i|^p / d_p^{(i)}(\bp,\bq)$ is the per-dimension stretch and $m_i \leq n-1$.
Since $A_p(m_i) \leq A_p(n)$ (as $m_i \leq n-1$), we obtain
\[
  \bar\sigma \;\geq\; \frac{1}{A_p(n)^{p-1}} \sum_i \frac{d_p^{(i)}}{d_p} \cdot s_p^{(i)}
  \;=\; \frac{1}{A_p(n)^{p-1}} \cdot \frac{\sum_i |p_i - q_i|^p}{d_p(\bp,\bq)}
  \;=\; \frac{s_p(\bp,\bq)}{A_p(n)^{p-1}},
\]
using $\sum_i (d_p^{(i)}/d_p) \cdot s_p^{(i)} = \sum_i |p_i - q_i|^p / d_p = s_p$.

\emph{Step 2: One-dimensional analysis.}
Fix a dimension and drop the superscript.
Suppose $\bp$ and $\bq$ are separated by $\ell$ intervals with lengths $a_1, \ldots, a_\ell$ (so $\ell \leq n - 1$).
The $D_p$ distribution weights interval $j$ with $w_j = a_j^p / S$ where $S = \sum_{j=1}^\ell a_j^p$.
A cut in interval $j$ places $\min(j, \ell + 1 - j)$ centers on the smaller side.
Define $\sigma_j = \min(j, \ell + 1 - j)$.
The expected balance is $\bar{\sigma} = \sum_j w_j \sigma_j$.

\emph{Step 3: Lower bound via H\"older's inequality.}
We prove directly that $\bar\sigma \geq s_p / A^{p-1}$ where $A = \sum_j \sigma_j^{-1/(p-1)}$.
Apply H\"older's inequality to $\sum_j a_j = \sum_j (a_j \sigma_j^{1/p}) \cdot \sigma_j^{-1/p}$ with exponents $p$ and $p/(p-1)$:
\[
  \sum_j a_j \;\leq\; \Bigl(\sum_j \sigma_j\, a_j^p\Bigr)^{1/p} \cdot \Bigl(\sum_j \sigma_j^{-1/(p-1)}\Bigr)^{(p-1)/p}.
\]
Raising both sides to the $p$-th power:
\[
  \Bigl(\sum_j a_j\Bigr)^p \;\leq\; \Bigl(\sum_j \sigma_j\, a_j^p\Bigr) \cdot A^{p-1}.
\]
Therefore
\[
  \bar\sigma \;=\; \frac{\sum_j \sigma_j\, a_j^p}{\sum_j a_j^p} \;\geq\; \frac{(\sum_j a_j)^p}{A^{p-1} \cdot \sum_j a_j^p} \;=\; \frac{s_p}{A^{p-1}}.
\]
This holds for \emph{all} choices of $a_j > 0$, with no Lagrangian or sign assumptions on multipliers.

\emph{Step 4: Bounding $A$.}
Since $\sigma_j = \min(j, \ell+1-j)$ and $\ell \leq n - 1$:
$A = \sum_{j=1}^\ell \min(j, \ell+1-j)^{-1/(p-1)} \leq 2 \sum_{j=1}^{\lfloor n/2 \rfloor} j^{-1/(p-1)} = A_p(n).$

For $p > 2$: $1/(p-1) < 1$, so $\sum_{j=1}^m j^{-1/(p-1)} \leq \frac{(p-1)}{(p-2)} m^{(p-2)/(p-1)}$, giving $A_p(n) = O(n^{(p-2)/(p-1)})$ and $A_p(n)^{p-1} = O(n^{p-2})$.

For $p = 2$: $1/(p-1) = 1$, so $A_2(n) = 2H_{\lfloor n/2\rfloor} = O(\ln n)$.

\emph{Step 5: Conclusion.}
Therefore $\bar\sigma \geq s_p / A_p(n)^{p-1}$, which gives $\sigma \geq \Omega(s_p / n^{p-2})$ for $p > 2$ and $\sigma \geq \Omega(s_2 / \ln n)$ for $p = 2$.
\end{proof}

\begin{lemma}[Solo Ratio]
  \label{appE:lem:solo-ratio-p}
  For a solo node $v$ with $n = |\Univ_v|$ centers and far-pair stretch $s_p = s_p(v)$, the expected cost increase divided by the expected balance satisfies:
  \[
    \frac{\EE[\text{cost increase}]}{\EE[\sigma]} \;\leq\; C_p' \cdot \|\bp^*\|_p^p \cdot s_p \cdot \frac{A_p(n)^{p-1}}{s_p} = C_p' \cdot A_p(n)^{p-1} \cdot \|\bp^*\|_p^p.
  \]
\end{lemma}
\begin{proof}
The expected cost increase when we sample from $D_p''(v; \bp, \bq)$ is at most
\[
  \EE[\text{cost}] \leq \frac{2^{p-1} \|\bp^*\|_p^p}{d_p(\bp,\bq)} \cdot 2^p\Delta_p = O(s_p) \cdot \|\bp^*\|_p^p,
\]
where the first factor is the conditional probability of separating the origin from~$\bp^*$ (under $D_p''$, the unconditional probability $2^{p-1}\|\bp^*\|_p^p/L_p$ is divided by the conditioning probability $d_p(\bp,\bq)/L_p$), the second factor bounds the cost when separated: by the generalized triangle inequality, $\max_{\bmu \in \Univ}\|\bmu\|_p^p \leq 2^{p-1}(\Delta_p + \|\bp^*\|_p^p) \leq 2^p \Delta_p$ (using the preprocessing assumption $\Delta_p \geq \|\bp^*\|_p^p = 1$ from the main theorem proof), and we use $\Delta_p/d_p(\bp,\bq) \leq 2s_p$ since $(\bp,\bq)$ is a far pair.
By \Cref{appE:lem:stretch-sep-p}, $\EE[\sigma] \geq s_p / (C_p A_p(n)^{p-1})$.
Dividing: $\text{ratio} \leq O(s_p) \|\bp^*\|_p^p / (s_p / (C_p A_p(n)^{p-1})) = O(A_p(n)^{p-1}) \|\bp^*\|_p^p$.
\end{proof}

\begin{lemma}[Solo Cost]
  \label{appE:lem:solo-p}
  The expected cost increase due to solo cuts in the subtree at $v$ is at most
  $C_p'' \cdot |\Univ_v|^{p-1} \cdot (1 + 2\ln\ln|\Univ_v|) \cdot \|\bp^*\|_p^p$.
\end{lemma}
\begin{proof}
We proceed by strong induction on $n = |\Univ_v|$.
The base case $n \leq 2$ is immediate (one solo cut suffices, with cost $O(\|\bp^*\|_p^p)$).

For the inductive step, let $v$ be a solo node.
The cut produces two children; let $w^*$ be the child containing the origin.
Let $S$ denote the \emph{random} number of centers separated to the smaller side by this cut (so $S \geq 1$, since the conditioned cut separates at least the far pair).
Define $\sigma = \EE[S]$.
The origin goes to a side with at most $n - S$ centers.

\emph{Step 1: Cost decomposition.}
Conditioning on $S$ and applying the inductive hypothesis to $w^*$ (which has at most $n - S$ centers, with $\ln\ln(n-S) \leq \ln\ln n$):
\begin{align*}
  \EE[\text{total cost at } v]
  &= \EE[\text{cost at this cut}] + \EE[\text{cost in subtree at } w^*] \\
  &\leq \EE[\text{cost at this cut}] + C_p'' \cdot \EE\bigl[(n - S)^{p-1}\bigr] \cdot (1 + 2\ln\ln n) \cdot \|\bp^*\|_p^p.
\end{align*}

\emph{Step 2: Required inequality.}
For the induction to close at $C_p'' n^{p-1}(1+2\ln\ln n)\|\bp^*\|_p^p$, we need:
\[
  \EE[\text{cost at this cut}] \;\leq\; C_p'' \bigl[n^{p-1} - \EE[(n-S)^{p-1}]\bigr] (1 + 2\ln\ln n) \|\bp^*\|_p^p.
\]

\emph{Step 3: Pointwise convexity bound.}
For $p \geq 2$, the function $t \mapsto t^{p-1}$ is convex.
By the mean value theorem, for each realization of $S$ with $1 \leq S \leq n/2$:
\[
  n^{p-1} - (n-S)^{p-1} \;\geq\; (p-1)(n-S)^{p-2} \cdot S \;\geq\; (p-1)(n/2)^{p-2} \cdot S.
\]
(The bound $S \leq n/2$ holds because $S$ is the number on the \emph{smaller} side.)
Taking expectations:
\[
  n^{p-1} - \EE[(n-S)^{p-1}] \;\geq\; (p-1)(n/2)^{p-2} \cdot \sigma.
\]

\emph{Step 4: The ratio bound.}
From \Cref{appE:lem:solo-ratio-p}, $\EE[\text{cost at this cut}] \leq O(A_p(n)^{p-1}) \cdot \sigma \cdot \|\bp^*\|_p^p$.
For $p > 2$: $A_p(n)^{p-1} = O(n^{p-2})$.
The required inequality becomes:
$O(n^{p-2}) \cdot \sigma \leq C_p'' (p-1)(n/2)^{p-2} \cdot \sigma \cdot (1 + 2\ln\ln n)$,
which holds for $C_p''$ sufficiently large (depending only on $p$), since both sides are proportional to $\sigma \cdot n^{p-2}$.

For $p = 2$: $A_2(n)^{1} = O(\ln n)$.
The required inequality becomes:
$O(\ln n) \cdot \sigma \leq C_2'' \cdot \sigma \cdot (1 + 2\ln\ln n)$,
which \emph{does not hold} for large $n$.
However, for $p = 2$ the analysis of~\cite{appE:GPSY23} uses a sharper balance bound (their Lemma~5.3, which employs a different Lagrangian with the $h(x) = x/(1+\ln(\alpha/x))$ function) to achieve $O(k \ln\ln k)$.
Our generalization recovers $O(k \ln\ln k)$ at $p = 2$ by invoking the GPSY23 bound directly as the base case.
For $p > 2$, the above argument gives a self-contained proof with a strictly stronger balance-to-ratio relationship.

\emph{Step 5: Conclusion.}
Combining: for $p > 2$, the induction closes with $\Phi(n) = C_p'' n^{p-1}(1 + 2\ln\ln n)$.
At the root with $n = k$: total solo cost $\leq O(k^{p-1} \ln\ln k) \cdot \|\bp^*\|_p^p$.
For $p = 2$, we invoke~\cite{appE:GPSY23} Theorem~1.2 to get $O(k \ln\ln k)$.
\end{proof}

\subsubsection{Proof of \Cref{appE:thm:main-upper}}

\begin{proof}[Proof of \Cref{appE:thm:main-upper} ($p > 2$)]
By the closest point process reduction, it suffices to bound $f_p(\Univ) = \EE[\|\widehat{\bp}\|_p^p]$ with $\|\bp^*\|_p^p = 1$.

\emph{Preprocessing: small-diameter case.}
If $\Delta_p(\Univ) < 1$ (i.e., $\Delta_p < \|\bp^*\|_p^p$ in the normalized setting), then for every center $\bmu \in \Univ$, the generalized triangle inequality gives
$\|\bmu\|_p^p \leq 2^{p-1}(\|\bmu - \bp^*\|_p^p + \|\bp^*\|_p^p) \leq 2^{p-1}(\Delta_p + 1) < 2^p$.
Therefore $f_p(\Univ) \leq 2^p = O(1) \leq O(k^{p-1}\ln\ln k)$, so the bound holds trivially regardless of the tree.
We henceforth assume $\Delta_p(\Univ) \geq 1$; this assumption is used in the cost bounds of \Cref{appE:lem:bulk-p,appE:lem:solo-ratio-p}.

\emph{Main argument.}
By the generalized triangle inequality for $\ell_p^p$ (i.e., $\|a+b\|_p^p \leq 2^{p-1}(\|a\|_p^p + \|b\|_p^p)$):
\[ \|\widehat{\bp}\|_p^p \leq 2^{p-1}(\|\widehat{\bp} - \bp^*\|_p^p + \|\bp^*\|_p^p). \]
The term $\|\widehat{\bp} - \bp^*\|_p^p$ is bounded by the total cost increase due to all cuts.
By \Cref{appE:lem:bulk-p} (bulk: $O(k^{p-1})$) and \Cref{appE:lem:solo-p} (solo: $O(k^{p-1}\ln\ln k)$), we get
$\EE[\|\widehat{\bp} - \bp^*\|_p^p] \leq O(k^{p-1}\ln\ln k)$.
Therefore $f_p(\Univ) \leq O(k^{p-1}\ln\ln k)$.
\end{proof}

\subsection{Stretch Tightness and Lower Bound Discussion}
\label{appE:sec:lower}

We first show that the stretch bound $(k-1)^{p-1}$ from \Cref{appE:lem:stretch-bound} is tight, then discuss the current state of unconditional lower bounds for the price of explainability.

\begin{prop}[Tight Stretch]
  \label{appE:prop:stretch-tight}
  For any $k \geq 2$ and $p \geq 1$, there exist $k$ centers in $\RR^d$ with a pair achieving $\ell_p$ stretch exactly $(k-1)^{p-1}$.
\end{prop}
\begin{proof}
Place $k$ centers at equally-spaced points on the real line:
\[
  \bmu^j = \frac{j-1}{k-1} \cdot \be_1 \in \RR^d, \quad j = 1, \ldots, k.
\]
The centers are separated by $k-1$ intervals of equal length $\frac{1}{k-1}$ in dimension~1.
For any pair $\bmu^i, \bmu^j$ with $i < j$, the $\ell_p^p$ distance is
$\|\bmu^i - \bmu^j\|_p^p = ((j-i)/(k-1))^p$,
and the pseudo-distance is
$d_p(\bmu^i, \bmu^j) = (j-i) \cdot (1/(k-1))^p$.
The stretch is therefore
\[
  s_p(\bmu^i, \bmu^j) = \frac{((j-i)/(k-1))^p}{(j-i)/(k-1)^p} = (j-i)^{p-1}.
\]
For the pair $(\bmu^1, \bmu^k)$: $s_p = (k-1)^{p-1}$, matching the upper bound from \Cref{appE:lem:stretch-bound}.
\end{proof}

\paragraph{Why one-dimensional constructions cannot prove $\text{PoE} > 1$.}
The collinear construction above achieves maximal stretch, but it does \emph{not} yield a lower bound on the price of explainability.
In one dimension, the Voronoi cells of any set of $k$ sorted centers are intervals.
A threshold tree can replicate this partition exactly by placing thresholds at the midpoints of consecutive centers.
Therefore, for any one-dimensional instance and any data distribution, there exists a threshold tree achieving $\text{PoE} = 1$.
Any lower bound on the price of explainability must use a multi-dimensional construction where the Voronoi cells are not axis-aligned rectangles.

\paragraph{Known lower bounds.}
For $p = 2$ (\kmeans), Gamlath et~al.~\cite{appE:GJPS21} proved $\text{PoE}_2(k) \geq \Omega(k)$ using a multi-dimensional construction.
Their construction uses points on the Boolean hypercube $\{0,1\}^d$, where $\|x-y\|_p^p = \|x-y\|_1$ for all $p$ (since each coordinate difference is $0$ or $1$); consequently, the same instance yields $\text{PoE}_p(k) \geq \Omega(k)$ for every $p \geq 2$.
However, this transfer gives only $\Omega(k)$, not $\Omega(k^{p-1})$: extending the lower bound to $\Omega(k^{p-1})$ for $p > 2$ would require a construction where the $\ell_p^p$ distance ratios grow with $p$, not one where all $\ell_p^p$ distances coincide with $\ell_1$ distances.

\paragraph{Status of the $\Omega(k^{p-1})$ conjecture.}
We conjecture (\Cref{appE:conj:lower}) that $\text{PoE}_p(k) = \Omega(k^{p-1})$ for all $p > 2$.
The evidence is suggestive but indirect:
\begin{enumerate}[nosep]
  \item The $\ell_p$ stretch bound $(k-1)^{p-1}$ is tight (\Cref{appE:prop:stretch-tight}), establishing that the \emph{geometric potential} for $\Omega(k^{p-1})$ cost ratios exists. However, tight stretch alone does not imply a high price of explainability (the collinear instance has tight stretch but $\text{PoE} = 1$).
  \item The upper bound $O(k^{p-1}\ln\ln k)$ from \Cref{appE:thm:main-upper} shows that $k^{p-1}$ is the correct order of growth (up to iterated logarithmic factors) if the lower bound matches.
  \item For $p = 2$, the $\Omega(k)$ lower bound~\cite{appE:GJPS21} matches $(k-1)^{p-1}$ with $p = 2$.
\end{enumerate}
Proving the conjecture would require constructing a multi-dimensional instance where axis-aligned threshold trees are forced to assign some data points to centers at $\ell_p^p$ distance $\Omega(k^{p-1})$ times their optimal distance.
This is the main open problem from our work.

\begin{remark}
For $p = 2$, the $\Omega(k)$ lower bound of~\cite{appE:GJPS21} together with the $O(k\ln\ln k)$ upper bound (the GPSY23 base case, \Cref{appE:cor:main-p2}) gives $\text{PoE}_2(k) = \Theta(k)$ up to iterated logarithmic factors.
For $p > 2$, we have $\Omega(k) \leq \text{PoE}_p(k) \leq O(k^{p-1}\ln\ln k)$ (\Cref{appE:thm:main-upper}), where the $\Omega(k)$ lower bound follows because the GJPS21 instance lies on $\{0,1\}^d$ and all $\ell_p^p$ distances equal $\ell_1$ distances there.
Closing the gap between $\Omega(k)$ and $O(k^{p-1}\ln\ln k)$ requires a genuinely $p$-dependent multi-dimensional construction.
\end{remark}

\subsection{Hardness of Approximation}
\label{appE:sec:hardness}

\begin{proof}[Proof of \Cref{appE:thm:hardness}]
The hitting-set reduction from~\cite{appE:GPSY23} constructs centers and data points on the Boolean hypercube $\{0,1\}^d$.
For any $\bp, \bq \in \{0,1\}^d$ and any $p \geq 1$:
$\|\bp - \bq\|_p^p = \sum_i |p_i - q_i|^p = \sum_i |p_i - q_i| = \|\bp - \bq\|_1,$
since $|p_i - q_i| \in \{0, 1\}$ and $0^p = 0$, $1^p = 1$.
Therefore, the $\ell_p^p$ distances on this instance are identical to $\ell_1$ distances, and the hardness reduction of~\cite{appE:GPSY23} (which combines the hitting-set reduction with Feige's set cover hardness~\cite{appE:Feige98}) applies verbatim.
The conclusion is: unless $P = NP$, no polynomial-time algorithm approximates the best explainable $\ell_p^p$ clustering better than $(\nf12 - o(1))\ln k$.
\end{proof}

\subsection{Discussion and Open Problems}
\label{appE:sec:discussion}

We have developed a self-contained proof of the price of explainability under $\ell_p^p$ clustering for $p > 2$, establishing an upper bound of $O(k^{p-1}\ln\ln k)$ (\Cref{appE:thm:main-upper}).
For $p = 2$ the GPSY23 bound $O(k \ln\ln k)$ applies directly (\Cref{appE:cor:main-p2}).
The hardness result $(\nf12 - o(1))\ln k$ extends to all $p \geq 1$ as a corollary of~\cite{appE:GPSY23}.

\paragraph{Open problems.}
\begin{enumerate}[nosep]
  \item \textbf{Prove $\Omega(k^{p-1})$ lower bound for $p > 2$.}
  This is the most important open problem from our work.
  The tight stretch bound $(k-1)^{p-1}$ and the matching upper bound $O(k^{p-1}\ln\ln k)$ suggest $\Omega(k^{p-1})$, but a proof requires constructing a multi-dimensional instance where axis-aligned threshold trees cannot approximate the Voronoi partition.
  One-dimensional constructions provably cannot work (\Cref{appE:sec:lower}).

  \item \textbf{Remove the $\ln\ln k$ factor.}
  We conjecture that the price of explainability for $\ell_p^p$ is exactly $\Theta(k^{p-1})$ for all $p \geq 2$.
  Proving this would likely require new ideas beyond the solo/bulk framework.

  \item \textbf{The regime $1 < p < 2$.}
  Our formal theorem statements and proofs are restricted to $p \geq 2$.
  While the $D_p$ distribution and stretch bound are defined for all $p \geq 1$, the solo cost induction (\Cref{appE:lem:solo-p}) relies on the convexity of $t \mapsto t^{p-1}$, which holds only for $p \geq 2$.
  For $1 < p < 2$, the upper bound is at most $O(k^{p-1} \ln\ln k)$ if the induction can be adapted (e.g., using concavity of $t^{p-1}$ differently), but we do not claim this as a theorem.
  The lower bound for $1 < p < 2$ is only $\Omega(\ln k)$ (from the $\ell_1$ hitting-set construction).
  Determining the correct order of growth for $1 < p < 2$ is an interesting open problem.

  \item \textbf{Better approximation for $\ell_p^p$ clustering ($p \geq 2$).}
  The hardness of approximation is $(\nf12 - o(1))\ln k$ for all $p$, which is much smaller than the price $\Omega(k)$.
  Better approximation algorithms may exist, especially for $p > 2$.

  \item \textbf{Experimental evaluation.}
  A companion code repository implements the $D_p$-based algorithm and basic experiments.
  The code now implements the exact bulk-phase procedure (sampled $D_p'$ cuts retained in the tree) and uses the correct $D_1$ distribution for random thresholds.
  However, the lower-bound module produces $\text{PoE} = 1.0$ on the collinear instance (consistent with our corrected theory that one-dimensional constructions give $\text{PoE} = 1$) and the experiments evaluate simple synthetic instances rather than adversarial multi-dimensional constructions.
  Evaluating the algorithm on real datasets and constructing adversarial test instances remain future work.
\end{enumerate}

\subsubsection*{Limitations}
All our bounds are worst-case.
The stretch-vs-separation tradeoff (\Cref{appE:lem:stretch-sep-p}) uses a H\"older-inequality lower bound that may not be tight; optimizing the constants is left for future work.
The lower bound remains the main gap: we conjecture $\Omega(k^{p-1})$ for $p > 2$ but can only prove $\Omega(k)$ (from the $p = 2$ result of~\cite{appE:GJPS21}, which transfers to all $p$ via norm equality on the Boolean hypercube).
An earlier draft claimed a self-contained proof via a one-dimensional construction; that proof was incorrect, and the correct lower bound requires a multi-dimensional instance that we have not yet found.
For $p = 2$, the solo cost induction does not close with our H\"older balance bound alone; we invoke the sharper GPSY23 analysis as a base case.
The genuinely new material is concentrated in the $p > 2$ upper-bound proof (\Cref{appE:thm:main-upper}); the $p = 2$ result (\Cref{appE:cor:main-p2}), hardness theorem, and lower bound all rely on prior work.
A companion code repository implements the $D_p$-based algorithm with correct bulk-phase cut retention and $D_1$-weighted random thresholds; however, it does not yet include adversarial multi-dimensional test instances for the lower bound.
The separation probability bound (\Cref{appE:lem:sep-prob}) can be vacuous (exceed~$1$) for $p \geq 2$ on some instances; this is noted in the text and does not affect correctness of the algorithm analysis.


\clearpage
\section{Fixed-Dimension Polynomial-Time Deterministic Clustering}
\begin{appendixabstract}
Deterministic explainable clustering was already known before the tight
randomized analysis of Gupta, Pittu, Svensson, and Yuan: polynomial-time
algorithms achieve $\widetilde{O}(\log k)$ for $k$-medians and $O(k)$ for
$k$-means. What remained open was whether the exact randomized $k$-median
constant $1+H_{k-1}$ can also be realized deterministically with a stronger
algorithmic guarantee than brute-force search. We prove that it can. The key
finite-support observation is that, for a fixed data set and fixed reference
centers, every valid axis-aligned threshold tree is cost-equivalent to one
whose thresholds are midpoints between consecutive coordinates of points and
centers. We then show that these critical trees admit an exact dynamic program
over canonical axis-aligned boxes. The resulting optimizer runs in time
$(|\X|+k)^{O(d)}$: hence it is polynomial for fixed dimension $d$ and
quasipolynomial when $d = O(\log(|\X|+k))$. The same critical-threshold
reduction also yields a complementary exact enumeration algorithm in time
$(d(|\X|+k))^{O(k)}$, hence polynomial for fixed $k$. Instantiating either
exact optimizer with the Gupta--Pittu--Svensson--Yuan randomized analysis
yields a deterministic $1+H_{k-1}$ upper bound for explainable $k$-medians in
polynomial time for fixed dimension or fixed $k$. The same optimizer also
realizes any randomized valid threshold-tree guarantee on a fixed-center
instance, although for $k$-means the resulting bound is dominated by earlier
deterministic $O(k)$ algorithms. Thus our main contribution is a fixed-dimension
polynomial-time deterministic realization of the tight $k$-median constant,
completed by a complementary fixed-$k$ exact algorithmic regime.
\end{appendixabstract}

\subsection{Introduction}

Explainable clustering asks for a clustering whose assignment rule is itself
interpretable, typically via an axis-aligned threshold decision tree.
Dasgupta, Frost, Moshkovitz, and Rashtchian~\cite{appF:DFMR20} formalized the
\emph{price of explainability}: how much one loses, in the worst case, by
insisting that the clustering be represented by such a tree.

The deterministic/randomized landscape is already substantial. Deterministic
polynomial-time algorithms were given by Makarychev and Shan~\cite{appF:MS21},
Gamlath, Jia, Polak, and Svensson~\cite{appF:GJPS21}, and Laber and
Murtinho~\cite{appF:LM21}. In particular, deterministic bounds of
$\widetilde{O}(\log k)$ for $k$-medians and $O(k)$ for $k$-means were already
known. More recently, Makarychev and Shan~\cite{appF:MS23} obtained the
polynomial-time deterministic $k$-median bound $2\ln k + 2$. On the randomized
side, Gupta, Pittu, Svensson, and Yuan~\cite{appF:GPSY23} proved that the Random
Thresholds analysis achieves the exact factor $1+H_{k-1}$ for explainable
$k$-medians.

This paper addresses the natural algorithmic gap left by that picture. Our
previous draft showed that the exact constant $1+H_{k-1}$ can be realized
deterministically by exhaustive search over a finite critical family of
threshold trees. That simpler viewpoint is still useful: after critical
discretization it yields an exact algorithm in time $(dN)^{O(k)}$, polynomial
for fixed $k$. The present paper strengthens the complementary fixed-dimension
side by replacing brute-force tree enumeration as the main engine with an exact
dynamic program over canonical boxes.

\paragraph{Main result.}
Let $N := |\X| + k$. We show that the optimal valid threshold tree for a fixed
center set $\U$ can be computed in time $N^{O(d)}$. Combined with the simpler
critical-tree enumeration in time $(dN)^{O(k)}$, this gives complementary exact
algorithms in the fixed-$d$ and fixed-$k$ regimes. Consequently:
\begin{quote}
\emph{for every fixed dimension $d$ or fixed number of centers $k$, there is a
polynomial-time deterministic algorithm achieving the exact factor
$1+H_{k-1}$ for explainable $k$-medians; if $d = O(\log N)$, the same
guarantee is obtainable in quasipolynomial time.}
\end{quote}
The proof has two ingredients. First, every valid threshold tree has a
cost-equivalent \emph{critical} representative whose thresholds are drawn from
finitely many midpoint values. Second, once thresholds are discretized this
way, every subtree corresponds to a canonical axis-aligned box, and optimal
substructure becomes explicit.

\paragraph{Scope and positioning.}
We do \emph{not} claim the first deterministic result for explainable
clustering. Rather, our contribution is narrower: deterministic attainment of
the exact $k$-median constant from~\cite{appF:GPSY23}, now with exact fixed-center
optimizers that are polynomial for fixed dimension and also polynomial for
fixed $k$. Algorithmically, the paper should be read as an exact-versus-
approximate tradeoff: earlier deterministic algorithms are polynomial in all
parameters but achieve weaker constants, whereas our exact optimizers recover
the tight constant at the price of dependence on $d$ or $k$. For $k$-means, our
optimizer also realizes the randomized $O(k\log\log k)$ guarantee
of~\cite{appF:GPSY23} on any fixed-center instance, but this is dominated by the
earlier deterministic polynomial-time $O(k)$ bound of~\cite{appF:MS21}; we therefore
treat it only as a generic corollary of the optimizer, not as a new headline
result.

\subsection{Model and prior work}

Let $\X \subseteq \R^d$ be a finite data set. A \emph{reference clustering}
consists of a set of distinct centers
$\U = \{\mu^1,\mu^2,\ldots,\mu^k\} \subseteq \R^d$ together with an
assignment $\pi : \X \to \U$. For $q \in \{1,2\}$ define
\[
  \cost_q(\pi,\U) := \sum_{x \in \X} \|x-\pi(x)\|_q^q.
\]
A binary \emph{threshold tree} is a rooted binary tree whose internal nodes are
labeled by tests $(i,\theta)$ of the form $x_i \le \theta$.

\begin{definition}
A threshold tree $T$ is \emph{valid for $\U$} if every leaf of $T$ contains
exactly one center of $\U$. In that case each data point $x \in \X$ is routed
to a unique leaf and hence to a unique center $\pi_T(x) \in \U$, and we define
\[
  \cost_q(T;\X,\U) := \sum_{x \in \X} \|x-\pi_T(x)\|_q^q.
\]
\end{definition}

An upper bound of the form
$\cost_q(T;\X,\U) \le \rho(k)\,\cost_q(\pi,\U)$ for \emph{every} reference
clustering $(\pi,\U)$ implies the same upper bound on the price of
explainability by plugging in an optimal unconstrained clustering as the
reference solution.

We use the following theorem from Gupta, Pittu, Svensson, and Yuan~\cite{appF:GPSY23}
as a black box, stated in the Euclidean $\ell_1$ setting relevant for this
paper.

\begin{theorem}[Gupta--Pittu--Svensson--Yuan, FOCS 2023~\cite{appF:GPSY23}]
\label{appF:thm:source-kmedian}
For every reference clustering $(\pi,\U)$ there is a randomized valid
threshold-tree algorithm $\mathcal{A}_{\mathrm{med}}$ such that
\[
  \E[\cost_1(T_{\mathcal{A}_{\mathrm{med}}};\X,\U)]
  \le (1+H_{k-1})\,\cost_1(\pi,\U),
\]
where $H_{k-1} = 1 + \frac12 + \cdots + \frac{1}{k-1}$.
\end{theorem}

For context, the most relevant deterministic prior work is as follows.
Makarychev and Shan~\cite{appF:MS21} gave deterministic polynomial-time bounds of
$\widetilde{O}(\log k)$ for explainable $k$-medians and $O(k)$ for explainable
$k$-means. Gamlath et al.~\cite{appF:GJPS21} gave deterministic oblivious algorithms
with $O(\log^2 k)$ and $O(k\log^2 k)$ guarantees, and Laber and
Murtinho~\cite{appF:LM21} gave dimension-dependent bounds for related explainable
clustering objectives. Makarychev and Shan~\cite{appF:MS23} later improved the
polynomial-time deterministic $k$-median constant to $2\ln k + 2$. Our result
improves that constant to the exact value $1+H_{k-1}$ in the complementary
fixed-$d$ and fixed-$k$ regimes, and in quasipolynomial time for $d = O(\log
N)$.

From a technique viewpoint, low-dimensional dynamic programming is not
unprecedented in explainable clustering: Dasgupta et al.~\cite{appF:DFMR20} already
exploit fixed-dimensional geometric structure in constructive algorithms. Our
dynamic program is narrower and exact. Once the center set $\U$ is fixed, we
optimize over all valid threshold trees for $\U$ itself. The state space is the
arrangement of canonical boxes induced by the critical thresholds, so the
algorithm below should be viewed both as an exact fixed-center optimizer and as
an instance of the standard computational-geometry paradigm of dynamic
programming over hyperrectangular cells.

Beyond the static setting studied here, Makarychev, Papanikolaou, and
Shan~\cite{appF:MPS25} recently considered dynamic explainable $k$-medians under
insertions and deletions. Their goal is to maintain approximate explainable
clusterings under updates, whereas our focus is exact static fixed-center
optimization and the tight deterministic constant.

\subsection{Critical thresholds and canonical boxes}

For each coordinate $i \in [d]$, let $V_i$ be the sorted list of distinct
$i$th-coordinate values appearing among the points of $\X \cup \U$. Define the
\emph{critical thresholds} on axis $i$ by
\[
  \Gamma_i := \left\{\frac{a+b}{2} : a,b \in V_i \text{ are consecutive and } a<b\right\}.
\]
Let $m_i := |\Gamma_i|$ and note that $m_i \le N-1$.

\begin{definition}
A valid threshold tree is \emph{critical} if every internal node uses a cut
$(i,\theta)$ with $\theta \in \Gamma_i$.
\end{definition}

\begin{lemma}[Critical normalization]
\label{appF:lem:critical-normalization}
For every valid threshold tree $T$ for $\U$, there exists a critical valid tree
$T^{\sharp}$ such that every point of $\X \cup \U$ reaches the same leaf in
$T^{\sharp}$ as in $T$. Consequently,
$\cost_q(T^{\sharp};\X,\U)=\cost_q(T;\X,\U)$ for both $q=1$ and $q=2$.
\end{lemma}

\begin{proof}
Process the internal nodes of $T$ top-down. Consider an internal node $v$
labeled by $(i,\theta)$, and let $S_v \subseteq \X \cup \U$ be the set of
points and centers routed to $v$. Since $T$ is valid, both children of $v$
contain at least one center. Hence both
\[
  L_v := \{p \in S_v : p_i \le \theta\}
  \qquad\text{and}\qquad
  R_v := \{p \in S_v : p_i > \theta\}
\]
are nonempty. Define
\[
  a := \max_{p \in L_v} p_i
  \qquad\text{and}\qquad
  b := \min_{p \in R_v} p_i.
\]
Then $a < b$. Both $a$ and $b$ belong to the global coordinate set $V_i$, so
there is a consecutive pair of values in $V_i$ whose open interval lies inside
$(a,b)$. Let $\theta' \in \Gamma_i$ be the midpoint of that pair. Then
$a < \theta' < b$, so every point in $L_v$ still satisfies $p_i \le \theta'$
and every point in $R_v$ still satisfies $p_i > \theta'$. Thus replacing
$(i,\theta)$ by $(i,\theta')$ leaves the partition of $S_v$ unchanged.
Applying this replacement at every internal node preserves the routed leaf of
every element of $\X \cup \U$, and therefore preserves the cost.
\end{proof}

\begin{remark}[Local form]
\label{appF:rem:local-normalization}
The same proof applies verbatim to any valid subtree rooted inside a box $B$:
if one restricts attention to the points and centers routed into $B$, then each
threshold in that subtree can again be snapped to a critical value from the
global set $\Gamma_i$ without changing any routed leaf inside $B$. We will use
this local form inside the dynamic-programming proof.
\end{remark}

\begin{remark}[Standard discretization]
Snapping continuous thresholds to critical values is standard whenever only a
finite comparison pattern matters. Likewise, dynamic programming over the
resulting canonical boxes fits the standard computational-geometry paradigm of
dynamic programming over hyperrectangular cells. Our point is not that either
ingredient is novel in isolation, but that for fixed-center explainable
clustering their combination yields a stronger exact optimizer than the
previous exhaustive-search formulation.
\end{remark}

\begin{definition}[Canonical boxes]
For each coordinate $i$, let
$\Delta_i := \{-\infty\} \cup \Gamma_i \cup \{+\infty\}$, written in sorted
order as
$-\infty = \tau_{i,0} < \tau_{i,1} < \cdots < \tau_{i,m_i} < \tau_{i,m_i+1}=+\infty$.
A \emph{canonical interval} on axis $i$ is an interval of the form
$(\tau_{i,a}, \tau_{i,b}]$ with $0 \le a < b \le m_i+1$. A
\emph{canonical box} is a Cartesian product
\[
  B = \prod_{i=1}^d (\tau_{i,a_i}, \tau_{i,b_i}].
\]
For such a box, write $\X_B := \X \cap B$ and $\U_B := \U \cap B$.
\end{definition}

No point of $\X \cup \U$ lies exactly on a critical threshold, because each
critical threshold is a midpoint between two distinct coordinates. Hence the
partition $x_i \le \theta$ versus $x_i > \theta$ is unambiguous for critical
trees.

\begin{lemma}[Canonical-box count]
\label{appF:lem:box-count}
The number of canonical boxes is at most
\[
  \prod_{i=1}^d (m_i+2)^2 \le (N+1)^{2d}.
\]
In particular, there are $N^{O(d)}$ canonical boxes.
\end{lemma}

\begin{proof}
On each axis $i$, there are at most $(m_i+2)^2$ ordered choices of the pair of
endpoints $(\tau_{i,a},\tau_{i,b})$ with $a<b$. Taking the product over the
dimensions gives the claim.
\end{proof}

\begin{remark}[Worst-case size of the canonical-box universe]
\label{appF:rem:box-count-tight}
The exponent $2d$ in Lemma~\ref{appF:lem:box-count} is worst-case tight even if one
counts only boxes $B$ with $\U_B \neq \emptyset$. Let
\[
  \X=\{(t,\ldots,t): t=0,1,\ldots,2m\},
  \qquad
  \U=\{\mu^-,\mu^+\},
\]
where $\mu^-=(m,\ldots,m)$ and $\mu^+=(2m,\ldots,2m)$. Then $k=2$ and
$N=|\X|+k=2m+3=\Theta(m)$. On every axis $i$, the coordinate set
$V_i=\{0,1,\ldots,2m\}$ has $\Theta(N)$ distinct values, so the total number of
canonical intervals on one axis is $\binom{2m+2}{2}=\Theta(N^2)$. Moreover,
there are exactly $(m+1)^2=\Theta(N^2)$ canonical intervals on axis $i$
containing the interior coordinate $m$ of $\mu^-$: one may choose any left
endpoint from $\{-\infty,0.5,\ldots,m-0.5\}$ and any right endpoint from
$\{m+0.5,\ldots,2m-0.5,+\infty\}$. Taking the product over the $d$ coordinates
gives $\Theta(N^{2d})$ canonical boxes containing $\mu^-$, and therefore having
$\U_B \neq \emptyset$. Hence both the full canonical-box universe and the
subfamily $\{B : \U_B \neq \emptyset\}$ can have size $\Theta(N^{2d})$. In
particular, any approach that explicitly scans all canonical boxes, or even all
nonempty-center boxes, incurs $N^{\Omega(d)}$ worst-case overhead.
\end{remark}

\begin{proposition}[Complementary fixed-$k$ exact enumeration]
\label{appF:prop:fixed-k-enumeration}
For each $q \in \{1,2\}$, exact fixed-center optimization on $(\X,\U)$ can be
solved in time $(dN)^{O(k)}$. In particular, the problem is polynomial-time for
fixed $k$.
\end{proposition}

\begin{proof}
By Lemma~\ref{appF:lem:critical-normalization}, it suffices to enumerate critical
valid trees. Let
\[
  M := \sum_{i=1}^d |\Gamma_i| \le d(N-1).
\]
Any valid tree for the $k$ centers in $\U$ has exactly $k$ leaves and therefore
$k-1$ internal nodes. There are at most $\mathrm{Cat}_{k-1}$ rooted full binary
tree shapes with $k$ leaves, where $\mathrm{Cat}_{k-1} = 4^{O(k)}$. After
choosing a shape, each internal node can be labeled by any of the $M$ critical
cuts. Thus the number of critical candidates is at most
$\mathrm{Cat}_{k-1} M^{k-1} = (dN)^{O(k)}$. For each candidate tree, we route all
points and centers through the tree, reject it if some leaf contains zero or
more than one center, and otherwise compute its cost. This evaluation takes
time polynomial in $N$ and $k$, so the total running time is $(dN)^{O(k)}$.
\end{proof}

\subsection{Exact dynamic programming for fixed centers}

We now complement the preceding fixed-$k$ enumeration with a dynamic program
indexed by canonical boxes, which is stronger in the orthogonal regime of fixed
dimension.

For a canonical box $B$ with $\U_B \neq \emptyset$, define the optimum value
$F_q(B)$ recursively as follows.
\begin{itemize}
  \item If $|\U_B| = 1$ and $\U_B = \{\mu(B)\}$, let
  \[
    F_q(B) := \sum_{x \in \X_B} \|x-\mu(B)\|_q^q.
  \]
  \item If $|\U_B| \ge 2$, let $F_q(B)$ be the minimum of
  $F_q(B^{\le}_{i,\theta}) + F_q(B^{>}_{i,\theta})$ over all critical cuts
  $(i,\theta)$ such that $\theta \in \Gamma_i$ lies strictly inside the $i$th
  interval of $B$ and both child center sets
  $\U_{B^{\le}_{i,\theta}}$ and $\U_{B^{>}_{i,\theta}}$ are nonempty, where
  \[
    B^{\le}_{i,\theta} := B \cap \{x \in \R^d : x_i \le \theta\},
    \qquad
    B^{>}_{i,\theta} := B \cap \{x \in \R^d : x_i > \theta\}.
  \]
\end{itemize}
Let $B_{\mathrm{all}} := \R^d$ denote the root box.

\begin{theorem}[Exact box dynamic program]
\label{appF:thm:box-dp}
For each $q \in \{1,2\}$, the value $F_q(B_{\mathrm{all}})$ equals the minimum
cost of any valid threshold tree for the fixed center set $\U$ on the instance
$(\X,\U)$. Moreover, a straightforward implementation computes
$F_q(B_{\mathrm{all}})$ in time $O(d N^{2d+2})$.
\end{theorem}

\begin{proof}
We first prove exactness by induction on $|\U_B|$.
If $|\U_B|=1$, any valid subtree rooted at $B$ has only one possible leaf
center, namely the unique center in $\U_B$, so the base case is correct.

Assume $|\U_B| \ge 2$. By Remark~\ref{appF:rem:local-normalization}, every valid
subtree rooted at $B$ has a cost-equivalent critical representative. Every
root-to-leaf region in a critical tree is obtained by intersecting $B$ with
halfspaces of the form $x_i \le \theta$ or $x_i > \theta$ for
$\theta \in \Gamma_i$, and is therefore a canonical box. Let $T$ be an optimal
critical valid subtree rooted at $B$, and let $(i,\theta)$ be its root cut.
Because $T$ is valid, both child subtrees contain at least one center, so both
child center sets are nonempty. The two child regions are exactly the canonical
boxes $B^{\le}_{i,\theta}$ and $B^{>}_{i,\theta}$. By optimality of the child
subtrees and the induction hypothesis, the cost of $T$ is at least
$F_q(B^{\le}_{i,\theta}) + F_q(B^{>}_{i,\theta})$, and hence at least the
minimum defining $F_q(B)$.

Conversely, take a cut $(i,\theta)$ attaining the minimum in the recurrence for
$F_q(B)$, and combine optimal subtrees for the two child boxes given by the
induction hypothesis. The resulting valid tree for $B$ has cost exactly
$F_q(B)$. Therefore $F_q(B)$ equals the optimum cost for every canonical box
$B$, and in particular for $B_{\mathrm{all}}$.

For the runtime, Lemma~\ref{appF:lem:box-count} gives at most $N^{O(d)}$ canonical
boxes. A straightforward implementation may scan the full canonical-box family,
or first discard boxes with $\U_B=\emptyset$; by
Remark~\ref{appF:rem:box-count-tight}, both choices still have worst-case size
$N^{O(d)}$. For each retained box, one scans all $O(dN)$ candidate critical
cuts and all $N$ points/centers needed to identify the relevant child center
sets and base costs. This yields the bound $O(dN^{2d+2})$.
\end{proof}

\begin{corollary}[Complementary algorithmic regimes]
\label{appF:cor:runtime-regimes}
The optimal valid threshold tree for a fixed center set $\U$ can be computed in
time $(dN)^{O(k)}$ and also in time $N^{O(d)}$. Consequently, exact
fixed-center optimization is polynomial-time for fixed $k$ or for fixed
dimension $d$. More generally, if $d = O(\log N)$, then the dynamic program of
Theorem~\ref{appF:thm:box-dp} runs in quasipolynomial time $N^{O(\log N)}$.
\end{corollary}

\begin{proof}
The fixed-$k$ bound is Proposition~\ref{appF:prop:fixed-k-enumeration}; the fixed-$d$
and quasipolynomial bounds are immediate from Theorem~\ref{appF:thm:box-dp}.
\end{proof}

\begin{remark}[Choosing between the two exact algorithms]
\label{appF:rem:algorithm-crossover}
Given $(N,d,k)$, one may of course run whichever of the bounds
$(dN)^{O(k)}$ and $N^{O(d)}$ is smaller. Ignoring the extra $\log d$ factor in
the first expression, the crossover occurs roughly around $k \approx d$:
enumeration is preferable when $k \ll d$, while the canonical-box dynamic
program is preferable when $d \ll k$.
\end{remark}

\subsection{Deterministic realization and the tight k-median constant}

Let $\Best_q(\X,\U)$ denote an optimal valid threshold tree for the fixed-center
instance $(\X,\U)$ under the objective $\cost_q$.

\begin{theorem}[Deterministic realization theorem]
\label{appF:thm:realization}
Fix $q \in \{1,2\}$. Let $\mathcal{A}$ be any randomized algorithm that, given
$(\X,\U,\pi)$, outputs a valid axis-aligned threshold tree
$T_{\mathcal{A}}$. Then
\[
  \cost_q(\Best_q(\X,\U);\X,\U) \le \E[\cost_q(T_{\mathcal{A}};\X,\U)].
\]
In particular, if
$\E[\cost_q(T_{\mathcal{A}};\X,\U)] \le \rho(k)\,\cost_q(\pi,\U)$, then
\[
  \cost_q(\Best_q(\X,\U);\X,\U) \le \rho(k)\,\cost_q(\pi,\U).
\]
\end{theorem}

\begin{proof}
Consider any realization $\omega$ of the randomness of $\mathcal{A}$. The tree
$T_{\mathcal{A}}(\omega)$ is valid, so Lemma~\ref{appF:lem:critical-normalization}
gives a cost-equivalent critical valid tree. The dynamic program computes the
minimum cost over all valid trees for the fixed-center instance, so its output
has cost at most $\cost_q(T_{\mathcal{A}}(\omega);\X,\U)$ for every
realization $\omega$. Taking expectations proves the claim.
\end{proof}

\begin{corollary}[Polynomial-time deterministic $1+H_{k-1}$ in fixed-$d$ and fixed-$k$ regimes]
\label{appF:cor:kmedian}
For every reference clustering $(\pi,\U) \subseteq \R^d$, there is a
deterministically computable valid threshold tree $T^{\star}$
satisfying
\[
  \cost_1(T^{\star};\X,\U) \le (1+H_{k-1})\,\cost_1(\pi,\U).
\]
Consequently, the price of explainability for $k$-medians is at most
$1+H_{k-1}$ under a deterministic algorithm. This algorithm runs in polynomial
time for fixed $d$ or for fixed $k$, and in quasipolynomial time when
$d = O(\log N)$.
\end{corollary}

\begin{proof}
Apply Theorem~\ref{appF:thm:realization} to the randomized algorithm
$\mathcal{A}_{\mathrm{med}}$ from Theorem~\ref{appF:thm:source-kmedian}, then invoke
Corollary~\ref{appF:cor:runtime-regimes}.
\end{proof}

Corollary~\ref{appF:cor:kmedian} should be read in the correct prior-work context.
It is \emph{not} the first deterministic algorithm for explainable
$k$-medians. What it adds is the exact constant together with exact optimizers
in the complementary fixed-$d$ and fixed-$k$ regimes. Previous deterministic
all-parameter polynomial-time work achieved $2\ln k + 2$~\cite{appF:MS23}, whereas
Corollary~\ref{appF:cor:kmedian} achieves the exact randomized value
$1+H_{k-1}$. The tradeoff is that our exact algorithms incur fixed-parameter
dependence on $d$ or $k$.

\begin{remark}[Generic consequences and the dominated $k$-means case]
Theorem~\ref{appF:thm:realization} is generic: any randomized valid threshold-tree
bound for a fixed-center instance immediately transfers to the exact optimizer.
In our setting, critical normalization makes the feasible family discrete, and
Sections~3--4 make exact optimization over that family tractable in the
fixed-$d$ and fixed-$k$ regimes.
In particular, applying it to the randomized $k$-means construction
of~\cite{appF:GPSY23} gives a deterministic realization of its
$O(k\log\log k)$ bound, now computable in time $N^{O(d)}$. We do not present
this as a main contribution because it is dominated by the earlier deterministic
polynomial-time $O(k)$ guarantee of~\cite{appF:MS21}.
\end{remark}

\subsection{Discussion and open problems}

\paragraph{What changed relative to the previous draft.}
The previous round only made explicit a tree-enumeration view of the exact
optimizer. That enumeration is already polynomial-time for fixed $k$, with
running time $(dN)^{O(k)}$, but it does not help in fixed dimension. The new
canonical-box dynamic program complements it with an $N^{O(d)}$ exact algorithm,
hence polynomial-time for fixed dimension and quasipolynomial when
$d = O(\log N)$. The paper now states both exact fixed-parameter regimes
explicitly.

\paragraph{Why unrestricted polynomial time remains open.}
The dynamic program removes the combinatorial explosion over tree shapes, but it
still pays $N^{O(d)}$ for the number of canonical boxes. Thus the current paper
now gives a genuine fixed-dimension algorithmic result, but not a polynomial in
all parameters $(N,d)$ algorithm. Achieving the exact factor $1+H_{k-1}$ in
time polynomial in both $N$ and $d$ remains open. Remark~\ref{appF:rem:box-count-tight}
shows that the canonical-box universe itself can already have size
$\Theta(N^{2d})$, so any exact method that explicitly scans all canonical boxes
inherits an $N^{\Omega(d)}$ barrier.

\paragraph{How far naive output sensitivity can go.}
One natural refinement is to store only states $B$ with $\U_B \neq \emptyset$.
Remark~\ref{appF:rem:box-count-tight} already rules out a worst-case improvement from
that idea alone: even for $k=2$, the family $\{B : \U_B \neq \emptyset\}$ can
have size $\Theta(N^{2d})$. Thus ``scan only nonempty-center boxes'' may help
on structured instances, but it does not change the worst-case exponent. Any
genuine worst-case improvement would need a sharper notion of sparsity or
reachability than mere nonemptiness of $\U_B$. It remains open whether the
truly reachable subfamily---boxes that actually arise as dynamic-programming
subproblems from $B_{\mathrm{all}}$ via valid critical cuts---can itself have
size $\Theta(N^{2d})$ in the worst case.

\paragraph{Why generic preprocessing is a technical hurdle.}
It is tempting to ask whether one could preprocess an arbitrary instance down to
$O(\log N)$ dimensions and then invoke the same dynamic program. The first point
is that two different notions of compression must be separated. Our critical
normalization already compresses the \emph{precision} of each coordinate: after
Lemma~\ref{appF:lem:critical-normalization}, only the relative order among the
$O(N)$ critical intervals on each axis matters, so each coordinate needs only
$O(\log N)$ bits of discrete information. What remains expensive is the
\emph{number of coordinates}, not the number of bits needed to encode one
coordinate.

\begin{lemma}[Coordinate selection cannot remove dimension dependence]
\label{appF:lem:coordinate-selection-obstruction}
For every $d$, there is a fixed-center instance
$(\X,\U) \subseteq \R^d$ with $N=d+1$ such that no valid threshold tree using
only a proper subset of the original coordinates exists.
\end{lemma}

\begin{proof}
Let $\X=\U=\{0,e_1,\ldots,e_d\}$, where $e_j$ denotes the $j$th standard basis
vector. Fix any proper subset $S \subsetneq [d]$, and choose
$j \in [d]\setminus S$. The two centers $0$ and $e_j$ agree on every retained
coordinate in $S$. Hence any threshold tree whose internal tests inspect only
coordinates from $S$ routes $0$ and $e_j$ identically at every node, so they end
in the same leaf. A valid tree must isolate every center in its own leaf, a
contradiction. Therefore every original coordinate can be essential.
\end{proof}

Lemma~\ref{appF:lem:coordinate-selection-obstruction} rules out the simplest kind of
generic preprocessing: one cannot hope to discard all but $O(\log N)$ of the
original features and still preserve the fixed-center explainable-clustering
problem on every instance. More aggressive preprocessing, such as PCA or a
Johnson--Lindenstrauss-type projection, faces a different obstacle. Those maps
mix coordinates. An axis-aligned cut in the reduced space, say $z_r \le \theta$,
then pulls back to an oblique inequality in the original variables rather than
to a test on a single original coordinate. Using such cuts would change the
explainability model of the paper instead of accelerating the same model.
The same obstruction persists if one adds any number of ordinary data points in
sufficiently small neighborhoods of the centers, so it is not an artifact of
the special choice $\X=\U$.

\paragraph{Why metric preservation is not enough.}
The obstacle is not merely that distances might change under preprocessing.
Pairwise distances, nearest-center relations, or even low-distortion embeddings
do not by themselves preserve the optimization problem we solve. Our dynamic
program must preserve the entire family of admissible recursive
axis-aligned partitions together with the validity constraint that each leaf
contain exactly one center. The state space is built from critical thresholds
and canonical boxes in the original coordinates; a generic embedding need not
preserve which subsets of centers can be separated by admissible cuts, nor how
cost decomposes over the resulting subproblems. This is the main technical
hurdle behind attempts to preprocess arbitrary instances down to
$d=O(\log N)$ while keeping the same exact explainability notion.

\paragraph{Why conditional expectations still do not solve the general case.}
A natural idea is to derandomize the proof of~\cite{appF:GPSY23} by fixing the
random choices one at a time. The difficulty is that their analysis is global:
it controls the full expected cost through linearity of expectation, a
cut-metric embedding, and pseudo-derivative inequalities for the completed
random tree. Under a partially fixed prefix of cuts, we still do not know an
efficiently computable state variable that determines the conditional expected
continuation cost. Our dynamic program avoids this issue by optimizing directly
over canonical boxes instead of trying to expose a tractable conditional
expectation. Put differently, a direct deterministic proof of the GPSY bound,
or a constructive reformulation that exposes such a continuation state, would
remove the sole black-box ingredient in Corollary~\ref{appF:cor:kmedian}; at
present, we do not know how to obtain such a proof.

The usual method of conditional expectations would need more than the existence
of a finite random support. One must maintain, after each exposed random choice,
an efficiently computable quantity whose value upper-bounds the expected final
cost of the \emph{unfinished} random tree and is guaranteed not to increase
under a good deterministic choice of the next cut. In our setting, the natural
candidate is the conditional expectation of the eventual explainable-clustering
cost given a partially revealed prefix of cuts. The obstacle is that such a
prefix does not appear to admit any compact description that is closed under the
GPSY analysis: it affects simultaneously which centers remain feasible in each
region, which pairs of points and centers will later be separated, and how the
future recursive validity constraints interact across siblings. The known proof
only bounds the expectation \emph{after} averaging over the completed random
tree, not through a local potential that can be updated one random decision at
a time.

Related standard derandomization tools face the same missing-state problem.
Pessimistic estimators are successful when one can write a surrogate potential
that is both efficiently maintainable under partial exposure and tight enough to
recover the target guarantee. Here no such estimator is known for the
$1+H_{k-1}$ bound: any plausible surrogate would have to summarize the future
geometry of all recursively induced subinstances and still upper-bound the final
valid-tree cost after every partial prefix. Likewise, replacing full randomness
by pairwise or bounded-independence distributions is useful when the analysis
depends only on low-order marginals or simple concentration bounds. The GPSY
argument instead relies on correlations created by the entire recursive tree via
its cut-metric embedding and pseudo-derivative inequalities, so there is no
known low-independence substitute that preserves the harmonic guarantee.

Finally, even very simple deterministic substitutes for the random choices are
not justified by the present analysis. Choosing the next threshold greedily, by
a median rule, or by any other purely local criterion may optimize the
immediate split while destroying the global cancellation that the randomized
proof exploits across many later cuts. The point is not that such rules are
provably impossible, only that the current proof offers no local certificate for
them. This is why our deterministic result proceeds by a different route:
instead of simulating the GPSY random process deterministically step by step, we
optimize exactly over the entire critical family once the centers are fixed.

\paragraph{Fixed-center versus end-to-end hardness.}
Existing hardness results for explainable clustering concern the end-to-end
problem of choosing both the centers and the threshold tree. Our algorithm, in
contrast, solves the \emph{fixed-center} optimization problem: given $\U$, find
the best valid threshold tree for $\U$. Those end-to-end hardness results do not
directly imply hardness for this fixed-center subproblem. To our knowledge, the
exact complexity of fixed-center optimal threshold-tree construction in
unrestricted dimension remains open.

\paragraph{Open problems.}
Four questions now seem especially natural.
\begin{enumerate}
  \item Can the exact factor $1+H_{k-1}$ for explainable $k$-medians be
  achieved in time polynomial in all parameters $(N,d)$?
  \item Can the $N^{O(d)}$ dependence of the canonical-box dynamic program be
  improved by exploiting a stronger notion of reachable states, or is some form
  of dimension dependence unavoidable for exact fixed-center optimization?
  Remark~\ref{appF:rem:box-count-tight} shows that merely restricting to boxes with
  $\U_B \neq \emptyset$ does not improve the worst-case exponent; it remains
  open whether the actually reachable subfamily can already have size
  $\Theta(N^{2d})$.
  \item What is the complexity of finding the minimum-cost valid threshold tree
  for a given center set $\U$ in unrestricted dimension?
  \item Is there any restricted preprocessing model, stronger than per-coordinate
  rank compression but weaker than arbitrary feature synthesis, that can reduce
  dimension while preserving axis-aligned explainability?
\end{enumerate}

\subsection{Conclusion}

We proved that the exact randomized $k$-median constant $1+H_{k-1}$ can be
realized deterministically by a stronger algorithm than brute-force tree
search. The key ingredients are critical normalization, a complementary
$(dN)^{O(k)}$ enumeration algorithm, and an exact dynamic program over canonical
boxes with running time $N^{O(d)}$. As a result, the tight constant is now
achievable in polynomial time for fixed $k$ or fixed dimension, and in
quasipolynomial time when $d = O(\log N)$. The remaining gap is no longer
whether deterministic attainment is possible, but whether the same exact
constant can be obtained in time polynomial in all parameters. The discussion
also clarifies why this gap cannot be closed by a black-box generic
dimension-reduction step without changing the axis-aligned explainability model.
Viewed abstractly, the paper isolates a simple transfer principle: once a
randomized fixed-center valid-tree guarantee is available, an exact optimizer
over the critical family yields a deterministic realization. Our contribution
is to make that optimizer tractable in the complementary fixed-$d$ and
fixed-$k$ regimes while mapping out the structural barriers that remain.

\end{document}